\documentclass[twoside]{article}
\usepackage[a4paper]{geometry}
\usepackage[latin1]{inputenc} % ou \usepackage[utf8]{inputenc}
\usepackage[T1]{fontenc} % ou \usepackage[OT1]{fontenc}
\usepackage{RR}

\usepackage{afterpage}
\usepackage{amsmath, amssymb, enumitem, amsthm}
\usepackage{bm} % bold mathematical symbols
\usepackage{caption}
\usepackage{color}
\usepackage{graphicx}
\usepackage{indentfirst}
\usepackage{mathtools}
\usepackage{lmodern} % load a font with all the characters
\usepackage{siunitx}
\usepackage{wrapfig}
\usepackage{xcolor, colortbl}

\makeatletter
\let\save@mathaccent\mathaccent
\newcommand*\if@single[3]{%
  \setbox0\hbox{${\mathaccent"0362{#1}}^H$}%
  \setbox2\hbox{${\mathaccent"0362{\kern0pt#1}}^H$}%
  \ifdim\ht0=\ht2 #3\else #2\fi
  }
\newcommand*\rel@kern[1]{\kern#1\dimexpr\macc@kerna}
\newcommand*\widebar[1]{\@ifnextchar^{{\wide@bar{#1}{0}}}{\wide@bar{#1}{1}}}
\newcommand*\wide@bar[2]{\if@single{#1}{\wide@bar@{#1}{#2}{1}}{\wide@bar@{#1}{#2}{2}}}
\newcommand*\wide@bar@[3]{%
  \begingroup
  \def\mathaccent##1##2{%
    \let\mathaccent\save@mathaccent
    \if#32 \let\macc@nucleus\first@char \fi
    \setbox\z@\hbox{$\macc@style{\macc@nucleus}_{}$}%
    \setbox\tw@\hbox{$\macc@style{\macc@nucleus}{}_{}$}%
    \dimen@\wd\tw@
    \advance\dimen@-\wd\z@
    \divide\dimen@ 3
    \@tempdima\wd\tw@
    \advance\@tempdima-\scriptspace
    \divide\@tempdima 10
    \advance\dimen@-\@tempdima
    \ifdim\dimen@>\z@ \dimen@0pt\fi
    \rel@kern{0.6}\kern-\dimen@
    \if#31
      \overline{\rel@kern{-0.6}\kern\dimen@\macc@nucleus\rel@kern{0.4}\kern\dimen@}%
      \advance\dimen@0.4\dimexpr\macc@kerna
      \let\final@kern#2%
      \ifdim\dimen@<\z@ \let\final@kern1\fi
      \if\final@kern1 \kern-\dimen@\fi
    \else
      \overline{\rel@kern{-0.6}\kern\dimen@#1}%
    \fi
  }%
  \macc@depth\@ne
  \let\math@bgroup\@empty \let\math@egroup\macc@set@skewchar
  \mathsurround\z@ \frozen@everymath{\mathgroup\macc@group\relax}%
  \macc@set@skewchar\relax
  \let\mathaccentV\macc@nested@a
  \if#31
    \macc@nested@a\relax111{#1}%
  \else
    \def\gobble@till@marker##1\endmarker{}%
    \futurelet\first@char\gobble@till@marker#1\endmarker
    \ifcat\noexpand\first@char A\else
      \def\first@char{}%
    \fi
    \macc@nested@a\relax111{\first@char}%
  \fi
  \endgroup
}
\makeatother

\newcommand{\ud}{\mathrm{d}}
\newcommand\R{\mathbb{R}}
\newcommand\E{\mathbb{E}}

\newcommand\abs[1]{\left\lvert#1\right\rvert}
\newcommand\norm[1]{\left\lVert #1 \right\rVert}
\newcommand{\aff}{\operatorname{aff}}

\newcommand{\ie}{\emph{i.e.}\ }

\newcommand{\del}{\mathrm{Del}}
\newcommand{\vor}{\mathrm{Vor}}
\newcommand{\str}{\mathrm{S}}
\newcommand{\phant}{{\vphantom{-1}}} % used to have invisible superscripts so the subscripts are on the same height
\newcommand{\canvas}{\mathcal{C}}
\newcommand{\VD}{\operatorname{VD}}
\newcommand{\BS}{\operatorname{BS}}
\newcommand{\Vor}{\operatorname{Vor}}
\newcommand{\RVD}{\operatorname{Vor}}
\newcommand{\EV}{\operatorname{EV}}
\newcommand{\DV}{\operatorname{DV}}
\newcommand{\RDT}{\widetilde{\del}_g(\mathcal{P})}
\newcommand{\SRDT}{\widebar{\del}_g(\mathcal{P})}
\newcommand{\DC}{\del_g(\mathcal{P})}
\newcommand{\DDC}{\del^{\textrm{d}}_g(\mathcal{P})}
\newcommand{\Vertices}{\operatorname{Vert}}

\DeclareMathOperator*{\argmin}{argmin}

\DeclareMathOperator{\ver}{Vert}

\DeclareMathOperator{\V}{V}
\newcommand{\simp}{\sigma}

\newcommand{\tet}{\tau}

\newcommand\numberthis{\addtocounter{equation}{1}\tag{\theequation}}

\newtheorem{theorem}{Theorem}[section]

\newtheorem{lemma}[theorem]{Lemma}
\newtheorem{definition}[theorem]{Definition}

\newtheorem{remark}[theorem]{Remark}
\newtheorem*{lemma*}{\textbf{Lemma}}
\newtheorem*{theorem*}{\textbf{Theorem}}

\RRNo{XXXX}
\RRdate{March 2017}
\RRauthor{% les auteurs
Jean-Daniel Boissonnat
  \and
Mael Rouxel-Labb\'e
\and Mathijs Wintraecken}
\authorhead{Boissonnat \& Rouxel-Labb\'e \& Wintraecken}
\RRtitle{Triangulations anisotropiques via diagrammes de Voronoi riemanniens discrets}
\RRetitle{Anisotropic triangulations via discrete Riemannian Voronoi diagrams}
\titlehead{Anisotropic triangulations via discrete Riemannian Voronoi diagrams}
\RRnote{Long version of the paper accepted at the 33rd International Symposium on Computational Geometry (2017)}
\RRresume{
	L'utilisation de triangulations anisotropes est souhaitable dans de nombreux domaines, tels que la r\'esolution d'\'equations aux d\'eriv\'ees partielles ou la visualisation de surfaces.
Pour r\'esoudre ce probl\`eme notoirement difficile, nous proposons l'utilisation de diagrammes de Voronoi riemanniens discrets, une structure discrete qui approxime le diagramme de Voronoi riemannien.
Cette structure a \'et\'e impl\'ement\'ee et nous avons montr\'e dans une publication empirique associ\'ee qu'elle produisait de bonnes triangulations pour des domaines de $\mathbb{R}^2$ et des surfaces plong\'ees dans $\mathbb{R}^3$.

Dans ce papier, nous \'etudions les aspects th\'eoriques de notre structure.
Etant donn\'e un ensemble fini de points $\mathcal{P}$ dans un domaine $\Omega$ \'equipp\'e d'une m\'etrique riemannienne, nous comparons le diagramme de Voronoi riemannien discret de $\mathcal{P}$ à sa version exacte.
Ces deux diagrammes ont chacun une structure dualle, respectivement appel\'ees le complexe de Delaunay riemannien discret et le complex de Delaunay riemannien.
Nous donnons des conditions qui garantissent que ces deux complexes sont identiques.
Il en r\'esulte de r\'esultats pr\'ecedemment \'etablis que le complexe de Delaunay riemannien discret peut être plong\'e dans $\Omega$ sous certaines conditions et une triangulation anisotrope faite de simplexes courbes est obtenue.
En outre, nous montrons que, sous des conditions analogues, les simplexes de cette triangulation peuvent être rendus droit.

}
\RRabstract{
The construction of anisotropic triangulations is desirable for various applications, such as the numerical solving of partial differential equations and the representation of surfaces in graphics.
To solve this notoriously difficult problem in a practical way, we introduce the discrete Riemannian Voronoi diagram, a discrete structure that approximates the Riemannian Voronoi diagram.
This structure has been implemented and was shown to lead to good triangulations in $\mathbb{R}^2$ and on surfaces embedded in $\mathbb{R}^3$ as detailed in our experimental companion paper.

In this paper, we study  theoretical aspects of our structure.
Given a finite set of points $\cal P$ in a domain $\Omega$ equipped with a Riemannian metric, we compare the discrete Riemannian Voronoi diagram of $\cal P$ to its Riemannian Voronoi diagram.
Both diagrams have dual structures called the discrete Riemannian Delaunay and the Riemannian Delaunay complex.
We provide conditions that guarantee that these dual structures are identical.
It then follows from previous results  that the discrete Riemannian Delaunay complex can be embedded in $\Omega$ under sufficient conditions, leading to an anisotropic triangulation with curved simplices.
Furthermore, we show that, under similar conditions, the simplices of this triangulation can be straightened.
}
\RRmotcle{G\'eom\'etrie riemannienne, Diagramme de Voronoi, Triangulation de Delaunay}
\RRkeyword{Riemannien geometry, Voronoi diagram, Delaunay triangulation}
\RRprojet{Datashape}
\RCSophia % Sophia Antipolis M\'editerran\'ee

\begin{document}

\makeRR   % cas d'un rapport de recherche

\section{Introduction}
\emph{Anisotropic triangulations} are triangulations whose elements are elongated along prescribed directions.
Anisotropic triangulations are known to be well suited when solving PDE's~\cite{dazevedo1989, Mirebeau2010, shewchuk2002}.
They can also significantly enhance the accuracy of a surface representation if the anisotropy of the triangulation conforms to the curvature of the surface~\cite{garland1997}.

Many methods to generate anisotropic triangulations are based on the notion of Riemannian metric and create triangulations whose elements adapt locally to the size and anisotropy prescribed by the local geometry.
The numerous theoretical and practical results~\cite{aurenhammer2000} of the Euclidean Voronoi diagram and its dual structure, the Delaunay triangulation, have pushed authors to try and extend these well-established concepts to the anisotropic setting.
Labelle and Shewchuk~\cite{labelle2003} and Du and Wang~\cite{du2005anisotropic} independently introduced two anisotropic Voronoi diagrams whose anisotropic distances are based on a discrete approximation of the Riemannian metric field.
Contrary to their Euclidean counterpart, the fact that the dual of these anisotropic Voronoi diagrams is an embedded triangulation is not immediate, and, despite their strong theoretical foundations, the anisotropic Voronoi diagrams of Labelle and Shewchuk and Du and Wang have only been proven to yield, under certain conditions, a good triangulation in a two-dimensional setting~\cite{canas2011, canas2012, cheng2006anisotropic, du2005anisotropic, labelle2003}.

Both these anisotropic Voronoi diagrams can be considered as an approximation of the exact Riemannian Voronoi diagram, whose cells are defined as $V_{g}(p_i) = \{ x\in\Omega \mid \, d_{g} (p_i, x) \leq d_{g} (p_j, x), \forall p_j\in\mathcal{P}\backslash p_i \}$, where $d_g(p,q)$ denotes the geodesic distance.
Their main advantage is to ease the computation of the anisotropic diagrams.
However, their theoretical and practical results are rather limited.
The exact Riemannian Voronoi diagram comes with the benefit of providing a more favorable theoretical framework and recent works have provided sufficient conditions for a point set to be an embedded Riemannian Delaunay complex~\cite{DBLP:journals/corr/BoissonnatDG13a, dyer2008sgp, leibon1999}.
We approach the Riemannian Voronoi diagram and its dual Riemannian Delaunay complex with a focus on both practicality and theoretical robustness.
We introduce the discrete Riemannian Voronoi diagram, a discrete approximation of the (exact) Riemannian Voronoi diagram.
Experimental results, presented in our companion paper~\cite{rouxel2016discretized}, have shown that this approach leads to good anisotropic triangulations for two-dimensional domains and surfaces, see Figure~\ref{fig-RDT}.

\begin{figure}[!hbt]
  \centering
  \includegraphics[height=0.37\linewidth]{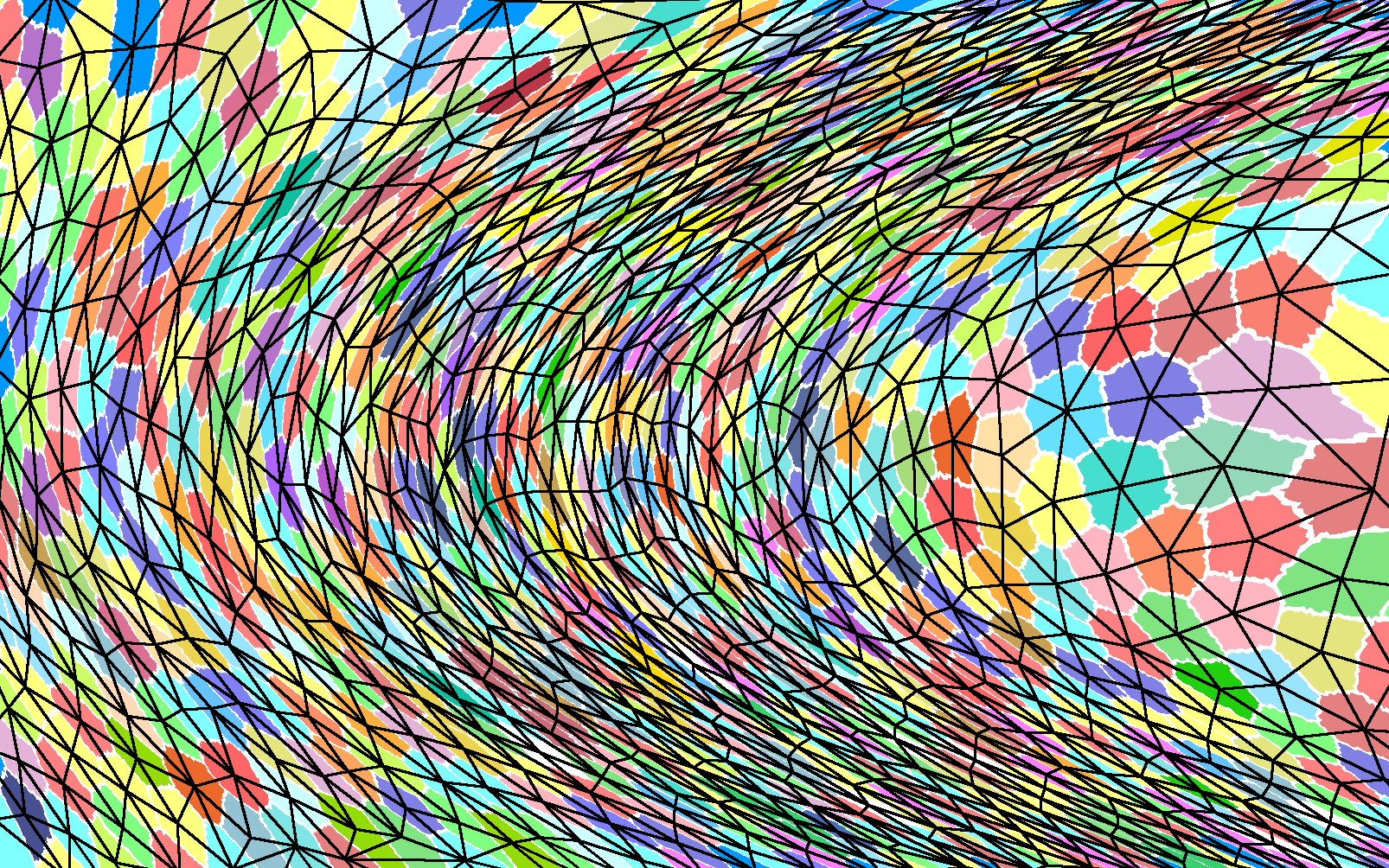}
  \hfill
  \includegraphics[height=0.37\linewidth]{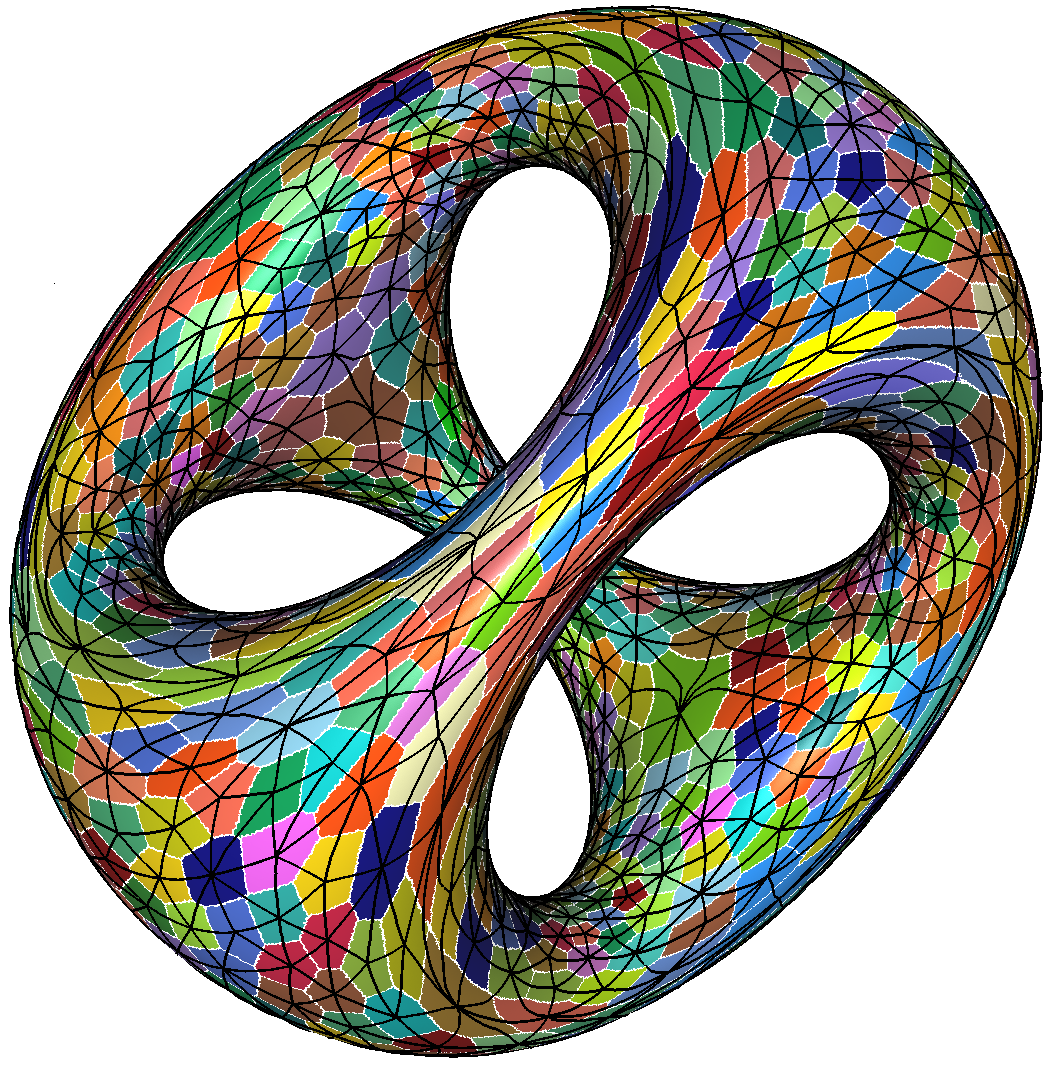}
  \caption{Left, the discrete Riemannian Voronoi diagram (colored cells with bisectors in white) and its dual complex (in black) realized with straight simplices of a two-dimensional domain endowed with a hyperbolic shock-based metric field. 
           Right, the discrete Riemannian Voronoi diagram and the dual complex realized with curved simplices of the ``chair'' surface endowed with a curvature-based metric field~\cite{rouxel2016discretized}.}
  \label{fig-RDT}
\end{figure}

We introduce in this paper the theoretical side of this work, showing that our approach is theoretically sound in all dimensions.
% By building upon theoretical results on the Riemannian Delaunay complex, we first show that the simplices of the Riemannian Delaunay complex can be ``straightened'' while preserving the embedded property of the triangulation.
We prove that, under sufficient conditions, the discrete Riemannian Voronoi diagram has the same combinatorial structure as the (exact) Riemannian Voronoi diagram and that the dual discrete Riemannian Delaunay complex can be embedded as a triangulation of the point set, with either curved or straight simplices.
Discrete Voronoi diagrams have been independently studied, although in a two-dimensional isotropic setting by Cao et al.~\cite{cao2015proof}.

%%%%%%%%%%%%%%%%%%%%%%%%%%%%%%%%%%%%%%%%%%%%%%%%%%%%%%%%%%%%%%%%%%%%%%%%%%%%%%%%%%%%%%%%%%%%%%%%%%%%%%%%%%%%%%%%%%%%%%%%%%%%%%%%%%%%%%%%%%%%%%%%%%%%%%%%%%
\section{Riemannian geometry}
In the main part of the text we consider an (open) domain $\Omega$ in $\R^n$ endowed with a Riemannian metric $g$, which we shall discuss below.
We assume that the metric $g$ is Lipschitz continuous.
%We ignore the boundary of $\Omega$.
%Our results extend to $n$-manifolds (embedded in $\R^m$) without boundary, as we shall see in Section~\ref{section-extension}.
The structures of interest will be built from a finite set of points $\mathcal{P}$, which we call \emph{sites}.

\subsection{Riemannian metric}
A \emph{Riemannian metric field}~$g$, defined over $\Omega$, associates a \emph{metric}~$g(p) = G_p$ to any point~$p$ of the domain.
This means that for any $v,w \in \mathbb{R}^n$ we associate an inner product $\langle v,w \rangle_g= v^t g(p)w$, in a way that smoothly depends on $p$.
Using a Riemannian metric, we can associate lengths to curves and define the geodesic distance $d_g$ as the minimizer of the lengths of all curves between two points.
When the map~${g: p\mapsto G}$ is constant, the metric field is said to be \emph{uniform}.
In this case, the distance between two points $x$ and $y$ in $\Omega$ is $d_{G} (x, y) = \norm{x - y}_G = \sqrt{(x - y)^t G^\phant (x - y)^\phant}$.

Most traditional geometrical objects can be generalized using the geodesic distance.
For example, the geodesic (closed) ball centered on~$p\in \Omega$ and of radius $r$ is given by $B_g(p, r) = \{x\in \Omega \mid d_g(p, x) \leq r \}$.
In the following, we assume that~$\Omega\subset\R^n$ is endowed with a Lipschitz continuous metric field~$g$.

We define the \emph{metric distortion} between two distance functions $d_g(x, y)$ and $d_{g'}(x, y)$ to be the function $\psi(g, g') $ such that for all $x, y$ in a small-enough neighborhood we have: $1 / \psi(g,g') \, d_{g}(x,y)\leq d_{g'}(x,y)\leq \psi(g,g')\, d_{g}(x,y)$.
Observe that $\psi(g, g') \geq 1$ and $\psi(g, g') = 1$ when $g = g'$.
Our definition generalizes the concept of distortion between two metrics $g(p)$ and $g(q)$, as defined by Labelle and Shewchuk~\cite{labelle2003} (see Appendix~\ref{appendix-distortion_properties}).

\subsection{Geodesy}
Let $v \in \mathbb{R}^n$.
From the unique geodesic $\gamma$ satisfying $\gamma(0) = p$ with initial tangent vector~$\dot{\gamma} = v$, one defines the \emph{exponential map} through $\exp(v) = \gamma(1)$.
The \emph{injectivity radius} at a point $p$ of $\Omega$ is the largest radius for which the exponential map at $p$ restricted to~a~ball of that radius is a diffeomorphism.
The injectivity radius $\iota_\Omega$ of $\Omega$ is defined as the infimum of the injectivity radii at all points.
For any $p\in \Omega$ and for a two-dimensional linear subspace~$H$ of the tangent space at $p$, we define the \emph{sectional curvature} $K$ at $p$ for $H$ as the Gaussian curvature at $p$ of the surface $\exp_p(H)$.

In the theoretical studies of our algorithm, we will assume that the injectivity radius of~$\Omega$ is strictly positive and its sectional curvatures are bounded. 

\subsection{Power protected nets}
Controlling the quality of the Delaunay and Voronoi structures will be essential in our proofs.
For this purpose, we use the notions of net and of power protection.

~\\
\textbf{Power protection of point sets} \label{section-protection}
Power protection of simplices is a concept formally introduced by Boissonnat, Dyer and Ghosh~\cite{DBLP:journals/corr/BoissonnatDG13a}. %, which specifically targets Delaunay triangulations.
Let $\sigma$ be a simplex whose vertices belong to $\mathcal{P}$, and
let $B_g(\sigma)= B_g(c,r)$ denote a circumscribing ball of $\sigma$
where $r=d_g(c,p)$ for any vertex $p$ of $\sigma$. We call $c$  the circumcenter of $\sigma$ and $r$ its circumradius.

For $0\leq\delta\leq r$, we associate to $B_g(\sigma)$ the dilated ball $B_g^{+\delta}(\sigma) = B(c, \sqrt{r^2 + \delta^2})$.
We say that $\sigma$ is \emph{$\delta$-power protected} if $B_g^{+\delta}(\sigma)$ does not contain any point of $\mathcal{P}\,\backslash \ver(\sigma)$ where $\ver(\sigma)$ denotes the vertex set of $\sigma$.
The ball $B_g^{+\delta}$ is the \emph{power protected} ball of $\sigma$.
Finally, a point set $\mathcal{P}$ is $\delta$-power protected if the Delaunay ball of its simplices are $\delta$-power protected.

~\\
\textbf{Nets} \label{section-net}
To ensure that the simplices of the structures that we shall consider are well shaped, we~will need to control the density and the sparsity of the point set.
The concept of net conveys these requirements through \emph{sampling} and \emph{separation} parameters.

The sampling parameter is used to control the density of a point set: if $\Omega$ is a bounded domain, $\mathcal{P}$ is said to be an \emph{$\varepsilon$-sample set} for $\Omega$ with respect to a metric field $g$ if $d_g(x, \mathcal{P})~<~\varepsilon$, for all~${x\in\Omega}$.
The sparsity of a point set is controlled by the separation parameter: the set~$\mathcal{P}$ is said to be \emph{$\mu$-separated} with respect to a metric field $g$ if $d_g(p, q)~\geq~\mu$ for all $p$, $q~\in~\mathcal{P}$.
If $\mathcal{P}$ is \mbox{an $\varepsilon$-sample} that is $\mu$-separated, we say that $\mathcal{P}$ is an \emph{$(\varepsilon,\mu)$-net}.

\section{Riemannian Delaunay triangulations}
Given a metric field $g$, the \emph{Riemannian Voronoi diagram} of a point set $\mathcal{P}$, denoted by~$\RVD_g(\mathcal{P})$, is the Voronoi diagram built using the geodesic distance $d_g$.
Formally, it is a partition of the domain in \emph{Riemannian Voronoi cells} $\{V_g(p_i)\}$, where $\V_g(p_i) = \{ x\in\Omega \mid \, d_g(p_i, x) \leq d_g(p_j, x), \forall p_j\in\mathcal{P}\,\backslash\, p_i \}$.

The Riemannian Delaunay complex of $\mathcal{P}$ is an abstract simplicial complex, defined as~the~nerve of the Riemannian Voronoi diagram, that is the set of simplices $\del_g(\mathcal{P}) = \{ \sigma \mid \ver(\sigma)\in\mathcal{P}, \cap_{p\in\sigma} \V_g(p) \neq 0 \}$.
There is a straightforward duality between the diagram and the complex, and between their respective elements.

In this paper, we will consider both abstract simplices and complexes, as well as their geometric realization in $\R^n$ with vertex set $\mathcal{P}$.
We now introduce two realizations of~a~simplex that will be useful, one curved and the other one straight.

The \emph{straight realization} of a $n$-simplex $\sigma$ with vertices in $\mathcal{P}$ is the convex hull of~its~vertices.
We denote it by $\widebar{\sigma}$.
In other words,
\begin{equation}
\bar{\sigma}=\{ x \in \Omega \subset \R^n \mid x=\sum_{p\in \sigma} \lambda_p(x)\, p, \lambda_p(x) \geq 0, \sum_{p\in \sigma} \lambda_p(x)=1\}. \label{equation-hatbar}
\end{equation}

The \emph{curved realization}, noted $\tilde{\sigma}$ is based on the notion of Riemannian center of mass~\cite{Karcher, dyer2014riemsplx.arxiv}.
Let $y$ be a point of $\bar{\sigma}$ with barycentric coordinate $\lambda_p(y), p\in \sigma$.
We can associate the energy functional $\mathcal{E}_{y} (x) = \frac{1}{2} \sum_{p\in \sigma} \lambda_{p}(y) d_g(x,p)^2$.
We then define the curved realization of~$\sigma$ as
\begin{equation}
\tilde{\sigma}=\{ \tilde{x} \in \Omega \subset \mathbb{R}^n \mid \tilde{x} = \argmin \mathcal{E}_{\bar{x}} (x), \bar{x}\in \bar{\sigma}\}. \label{equation-hattilde}
\end{equation}
The edges of $\tilde{\sigma}$ are geodesic arcs between the vertices.
Such a curved realization is well~defined provided that the vertices of $\sigma$ lie in a sufficiently small ball according to the~following theorem  of Karcher~\cite{Karcher}.

\begin{theorem}[Karcher] \label{theorem-Karcher}
Let the sectional curvatures $K$ of $\Omega$ be bounded, that is $\Lambda_{-} \leq K \leq \Lambda_{+}$.
Let us consider the function $\mathcal{E}_{y}$ on $B_\rho$, a geodesic ball of radius $\rho$ that contains the set~$\{ p_i \}$.
Assume that $\rho\in\R^{+}$ is less than half the injectivity radius and less than $\pi / 4 \sqrt{\Lambda_{+}}$ if $\Lambda_{+} > 0$.
Then $\mathcal{E}_{y}$ has a unique minimum point in $B_\rho$, which is called the \emph{center of mass}.
\end{theorem}

Given an (abstract) simplicial complex $\mathcal{K}$ with vertices in $\mathcal{P}$, we define the straight (resp., curved) realization of $\mathcal{K}$ as the collection of straight (resp., curved) realizations of~its~simplices, and we write $\bar{\mathcal{K}}= \{ \bar{\sigma}, \sigma\in \mathcal{K}\}$ and $\tilde{\mathcal{K}}= \{ \tilde{\sigma}, \sigma\in \mathcal{K}\}$.

We will consider the case where $\mathcal{K}$ is $\del_g(\mathcal{P})$. A simplex of $\widebar{\del}_g(\mathcal{P})$ will simply be called a straight Riemannian Delaunay simplex and a simplex of $\widetilde{\del}_g(\mathcal{P})$ will be called a curved Riemannian Delaunay simplex, omitting ``realization of''.
In the next two sections, we give sufficient conditions for $\widebar{\del}_g(\mathcal{P})$  and $\widetilde{\del}_g(\mathcal{P})$ to be embedded in $\Omega$, in which case we will call them the straight and the curved Riemannian \emph{triangulations} of $\mathcal{P}$.

\subsection{Sufficient conditions for \boldmath{$\widetilde{\del}_g(\mathcal{P})$} to be a triangulation of $\mathcal{P}$ } \label{section-embeddability_curved}
It is known that  $\widetilde{\del}_g(\mathcal{P})$ is embedded in $\Omega$  under sufficient conditions. %, be embedded in the domain using Karcher simplices.
We give a short overview of these results.
As in Dyer et al.~\cite{dyer2014riemsplx.arxiv}, we define the non-degeneracy of a simplex~$\tilde{\sigma}$ of $\widetilde{\del}_g(\mathcal{P})$.
\begin{definition}
The curved realization $\widetilde{\sigma}$ of a Riemannian Delaunay simplex $\sigma$ is said to be non-degenerate if and only if it is homeomorphic to the standard simplex.
\end{definition}

Sufficient conditions for the complex $\widetilde{\del}_g(\mathcal{P})$ to be embedded in $\Omega$ were given in~\cite{dyer2014riemsplx.arxiv}: a curved simplex is known to be non-degenerate if the Euclidean simplex obtained by lifting the vertices to the tangent space at one of the vertices via the exponential map has sufficient quality compared to the bounds on sectional curvature.
Here, good quality means that the simplex is well shaped, which may be expressed either through its fatness (volume compared to longest edge length) or its thickness (smallest height compared to longest edge length).

Let us assume that, for each vertex $p$ of $\del_g (\mathcal{P})$, all the curved Delaunay simplices in a neighborhood of $p$ are non-degenerate and patch together well.
Under these conditions, $\widetilde{\del}_g(\mathcal{P})$ is embedded in $\Omega$.
We call $\widetilde{\del}_g(\mathcal{P})$ the \emph{curved Riemannian Delaunay triangulation} of $\mathcal{P}$.

\subsection{Sufficient conditions for \boldmath{$\protect\widebar{\del}_g(\mathcal{P})$} to be a triangulation of $\mathcal{P}$ } \label{section-embeddability_straight}
Assuming that the conditions for  $\widetilde{\del}_g(\mathcal{P})$ to be embedded in $\Omega$ are satisfied, we now give conditions such that $\widebar{\del}_g(\mathcal{P})$ is also embedded in $\Omega$.
The key ingredient will be a bound on the distance between a point of a simplex $\tilde{\sigma}$ and the corresponding point on the associated  straight simplex $\bar{\sigma}$ (Lemma~\ref{lemma-geo_straight_proximity}).
This bound depends on the properties of the set of sites and on the local distortion of the metric field.
When this bound is sufficiently small, $\widebar{\del}_g(\mathcal{P})$ is embedded in $\Omega$ as stated in Theorem~\ref{theorem-SRDT_embedding_anyD}.

\begin{lemma} \label{lemma-geo_straight_proximity}
Let $\sigma$ be an $n$-simplex of $\del_g(\mathcal{P})$.
Let $\bar{x}$ be a point of $\bar{\sigma}$ and $\widetilde{x}$ the associated point on $\widetilde{\sigma}$ (as defined in Equation~\ref{equation-hatbar}).
If the geodesic distance $d_g$ is close to the Euclidean distance $d_{\E}$, i.e. the distortion $\psi(g, g_{\E})$ is bounded by $\psi_0$, then $\abs{\widetilde{x} - \bar{x}} \leq \sqrt{ 2 \cdot 4^3 (\psi_0 - 1) \varepsilon^2 }$.
\end{lemma}

We now apply Lemma~\ref{lemma-geo_straight_proximity} to the facets of the simplices of $\RDT$.
The altitude of the vertex $p$ in a simplex $\tau$ is noted $D(p, \tau)$.

\begin{theorem} \label{theorem-SRDT_embedding_anyD}
Let $\mathcal{P}$ be a $\delta$-power protected $(\varepsilon, \mu)$-net with respect to $g$ on $\Omega$.
Let $\sigma$ be any $n$-simplex of $\del_g(\mathcal{P})$ and $p$ be any vertex of $\sigma$.
Let $\tau$ be a facet of $\sigma$ opposite of vertex $p$.
If, for all $\widetilde{x} \in \widetilde{\tau}$, we have $\abs{\widetilde{x} - \bar{x}} \leq D(p_i, \sigma)$  ($\bar{x}$ is defined in Equation~\ref{equation-hatbar}), then $\widebar{\del}_d(\mathcal{P})$ is embedded in $\Omega$.
\end{theorem}
The condition $\abs{\widetilde{x} - \bar{x}} \leq D(p_i, \sigma)$ is achieved for a sufficiently dense sampling according to Lemma~\ref{lemma-geo_straight_proximity} and the fact that the distortion $\psi_0 = \psi(g, g_{\E})$ goes to $1$ when the density increases.
The complete proofs of Lemma~\ref{lemma-geo_straight_proximity} and Theorem~\ref{theorem-SRDT_embedding_anyD} can be found in Appendix~\ref{appendix-proofs_straight_embeddability}. 

%\begin{remark}
%The stability of the center of mass in Lemma \ref{lemma-geo_straight_proximity} generalizes trivially to continuous mass distributions and negative weights (see Appendix~\ref{appendix-proofs_straight_embeddability}), which is new in this generality.
%\end{remark}

%%%%%%%%%%%%%%%%%%%%%%%%%%%%%%%%%%%%%%%%%%%%%%%%%%%%%%%%%%%%%%%%%%%%%%%%%%%%%%%%%%%%%%%%%%%%%%%%%%%%%%%%%%%%%%%%%%%%%%%%%%%%%%%%%%%%%%%%%%%%%%%%%%%%%%%%%%
\section{Discrete Riemannian structures}
Although Riemannian Voronoi diagrams and Delaunay triangulations are appealing from a theoretical point of view, they are very difficult to
compute in practice despite many studies~\cite{Peyre:2010:GMC:1923909.1923910}.
To circumvent this difficulty, we  introduce the discrete Riemannian Voronoi diagram.
This discrete structure is easy to compute (see our companion paper~\cite{rouxel2016discretized} for details) and, as will be shown in the following sections, it is a good approximation of the exact Riemannian Voronoi diagram.
In particular, their dual Delaunay structures are identical under appropriate conditions.

We assume that we are given a dense triangulation of the domain $\Omega$ we call the \emph{canvas} and denote by $\canvas$.
The canvas will be used to approximate geodesic distances between points of $\Omega$ and to construct the discrete Riemannian Voronoi diagram of $\mathcal{P}$, which we denote by~$\vor_g^{\textrm{d}}(\mathcal{P})$.
This bears some resemblance to the graph-induced complex of Dey et al.~\cite{dey2015graph}.
Notions related to the canvas will explicitly carry \emph{canvas} in the name (for example, an edge of $\canvas$ is a \emph{canvas edge}).
In our analysis, we shall assume that the canvas is a dense triangulation, although weaker and more efficient structures can be used (see Section~\ref{sec:implementation} and ~\cite{rouxel2016discretized}).

\subsection{The discrete Riemannian Voronoi Diagram}
To define the discrete Riemannian Voronoi diagram of $\mathcal{P}$, we need to give a unique color to each site of $\mathcal{P}$ and to color the vertices of the canvas accordingly.
Specifically, each canvas vertex is colored with the color of its closest site.

\begin{definition}[Discrete Riemannian Voronoi diagram] \label{definition-DRVD}
Given a metric field $g$, we associate to each site $p_i$ its \emph{discrete cell} $\V^{\textrm{d}}_g(p_i)$ defined as the union of all canvas simplices with at least one vertex of the color of $p_i$.
We call the set of these cells the \emph{discrete Riemannian Voronoi diagram} of $\mathcal{P}$, and denote it by $\Vor^{\textrm{d}}_g(\mathcal{P})$.
\end{definition}

Observe that contrary to typical Voronoi diagrams, our discrete Riemannian Voronoi diagram is not a partition of the canvas.
Indeed, there is a one canvas simplex-thick overlapping since each canvas simplex $\sigma_{\canvas}$ belongs to all the Voronoi cells whose sites' colors appear in the vertices of $\sigma_{\canvas}$.
This is intentional and allows for a straightforward definition of the complex induced by this diagram, as shown below.

\subsection{The discrete Riemannian Delaunay complex}
We define the \emph{discrete Riemannian Delaunay complex} as the set of simplices $\del_g^{\textrm{d}}(\mathcal{P}) = \{ \sigma \mid \ver(\sigma)\in\mathcal{P}, \cap_{p\in\sigma} \V_g^{\textrm{d}}(p) \neq 0 \}$.
Using a triangulation as canvas offers a very intuitive way to construct the discrete complex since each canvas \mbox{$k$-simplex} $\simp$ of~$\canvas$ has $k+1$ vertices $\{v_0,\dots,v_k\}$ with respective colors $\{c_0,\dots,c_k\}$ corresponding to the sites $\{p_{c_0},\dots,p_{c_k}\} \in \mathcal{P}$.
Due to the way discrete Voronoi cells overlap, a canvas simplex~$\simp_{\canvas}$ belongs to each discrete Voronoi cell whose color appears in the vertices of $\simp$.
Therefore, the intersection of the discrete Voronoi cells $\{V_g^{\textrm{d}}(p_i)\}$ whose colors appear in the vertices of $\simp$ is non-empty and the simplex $\sigma$ with vertices $\{p_i\}$ thus belongs to the discrete Riemannian Delaunay complex.
In that case, we say that the canvas simplex $\simp_{\canvas}$ \emph{witnesses} (or is a witness of) $\sigma$.
For~example, if the vertices of a canvas $3$-simplex $\tet_{\canvas}$ have colors yellow--blue--blue--yellow, then the intersection of the discrete Voronoi cells of the sites $p_{yellow}$ and $p_{blue}$ is non-empty and the one-simplex $\sigma$ with vertices $p_{yellow}$ and $p_{blue}$ belongs to the discrete Riemannian Delaunay complex.
The canvas simplex $\tet_{\canvas}$ thus witnesses the (abstract, for now) edge between $p_{yellow}$ and $p_{blue}$.

Figure~\ref{fig-canvas_capture} illustrates a canvas painted with discrete Voronoi cells, and the witnesses of the discrete Riemannian Delaunay complex.

\begin{figure*}[!htb]
  \centering
   \includegraphics[width=\linewidth]{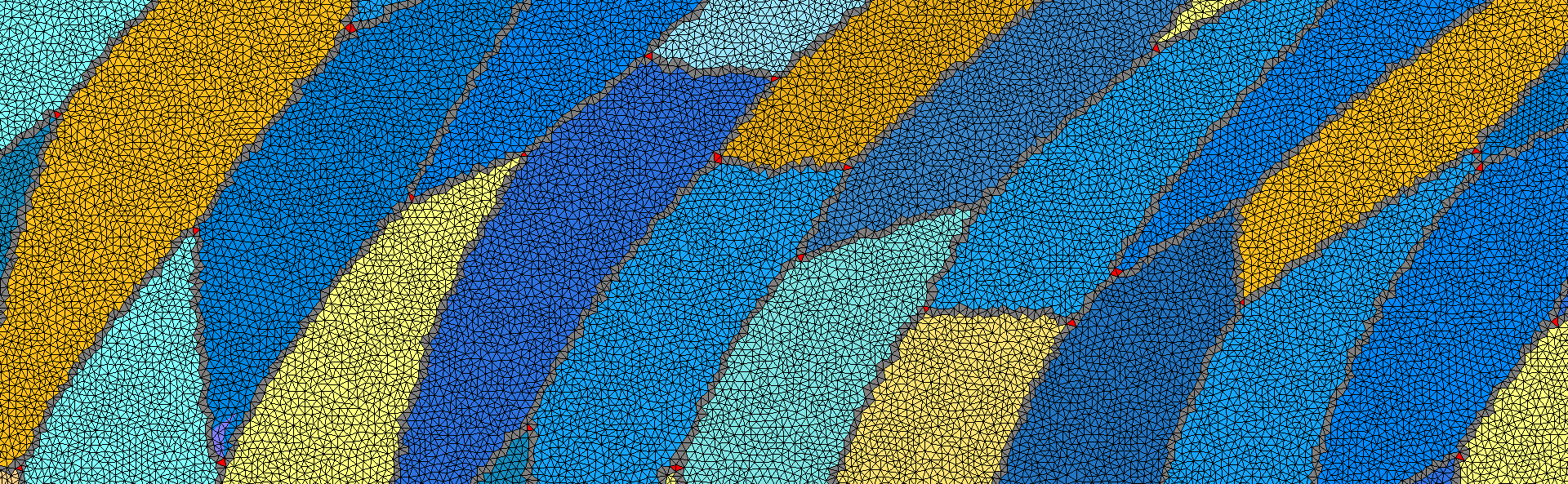}
  \caption{A canvas (black edges) and a discrete Riemannian Voronoi diagram drawn on it.
           The canvas simplices colored in red are witnesses of Voronoi vertices.
           The canvas simplices colored in grey are witnesses of Voronoi edges.
           Canvas simplices whose vertices all have the same color are colored with that color.}
  \label{fig-canvas_capture}
\end{figure*}

\begin{remark} \label{remark-capture_colors}
If the intersection $\bigcap_{i=0\dots k} \V_g^{\textrm{d}}(p_{c_i})$ is non-empty, then the intersection of any subset of $\{ \V_g^{\textrm{d}}(p_{c_i}) \}_{i=0\dots k}$ is non-empty.
In other words, if a canvas simplex $\simp_{\canvas}$ witnesses a simplex $\sigma$, then for each face~$\tau$ of~$\sigma$, there exists a face $\tau_{\canvas}$ of~$\simp_{\canvas}$ that witnesses~$\tau$.
As we assume that there is no boundary, the complex is pure and it is sufficient to only consider canvas $n$-simplices whose vertices have all different colors to build~$\DDC$.
\end{remark}

Similarly to the definition of curved and straight Riemannian Delaunay complexes, we can define their discrete counterparts we respectively denote by $\widetilde{\del_g^{\textrm{d}}}(\mathcal{P})$ and $\widebar{\del}_g^{\textrm{d}}(\mathcal{P})$.
We will now exhibit conditions such that these complexes are well-defined and embedded in $\Omega$.

%%%%%%%%%%%%%%%%%%%%%%%%%%%%%%%%%%%%%%%%%%%%%%%%%%%%%%%%%%%%%%%%%%%%%%%%%%%%%%%%%%%%%%%%%%%%%%%%%%%%%%%%%%%%%%%%%%%%%%%%%%%%%%%%%%%%%%%%%%%%%%%%%%%%%%%%%%
\section{Equivalence between the discrete and the exact structures} \label{section-theory}
We first give conditions such that $\vor^{\textrm{d}}_g (\mathcal{P})$ and $\vor_g (\mathcal{P})$ have the same combinatorial structure, or, equivalently, that the dual Delaunay complexes $\del_g(\mathcal{P})$ and $\del_g^{\textrm{d}}(\mathcal{P})$ are identical.
Under these conditions, the fact that  $\del_g^{\textrm{d}}(\mathcal{P})$ is
embedded in $\Omega$ will immediately follow from the fact that the
exact Riemannian Delaunay complex  $\del_g(\mathcal{P})$ is embedded (see Sections~\ref{section-embeddability_curved} and \ref{section-embeddability_straight}).
It thus remains to exhibit conditions under which $\DDC$ and $\DC$ are identical.

Requirements will be needed on both the set of sites in terms of density, sparsity and protection, and on the density of the canvas.
The central idea in our analysis is that power protection of $\mathcal{P}$ will imply a lower bound on the distance separating two non-adjacent Voronoi objects (and in particular two Voronoi vertices). 
From this lower bound, we will obtain an upper bound on the size on the cells of the canvas so that the combinatorial structure of the discrete diagram is the same as that of the exact one.
The density of the canvas is expressed by $e_{\canvas}$, the length of its longest edge.

The main result of this paper is the following theorem.
\begin{theorem} \label{theorem-generic_anyD}
Assume that $\mathcal{P}$ is a $\delta$-power protected $(\varepsilon,\mu)$-net in $\Omega$ with respect to $g$.
Assume further that $\varepsilon$ is sufficiently small and $\delta$ is
sufficiently large  {compared to the distortion between $g(p)$ and $g$ in an $\varepsilon$-neighborhood of $p$}.
Let $\{\lambda_i \}$ be the eigenvalues of $g(p)$ and $\ell_0$ a value that depends on $\varepsilon$ and $\delta$ (Precise bounds for $\varepsilon, \delta$ and $l_0$ are given in the proof).
Then, if $e_{\canvas} < \min\limits_{p\in\mathcal{P}} \left[ \min\limits_{i} \left(\sqrt{\lambda_i} \,\right) \min \left\{ \mu / 3 , \ell_0 / 2 \right\} \right]$, $\DDC = \DC$.
\end{theorem}

The rest of the paper will be devoted to the proof of this theorem. Our analysis is divided into two parts.
We first consider in Section~\ref{section-complex_equality} the most basic case of a  domain of $\R^n$ endowed with the Euclidean metric field.
The result is given by Theorem~\ref{theorem-basic_anyD}.
The assumptions are then relaxed and we consider the case of an arbitrary metric field over $\Omega$ in Section~\ref{section-extension}.
As we shall see, the Euclidean case already contains most of the difficulties that arise during the proof and the extension to more complex settings will be deduced from the Euclidean case by bounding the distortion.

\section{Equality of the Riemannian Delaunay complexes in the Euclidean setting} \label{section-complex_equality}
% As explained in the previous section, we shall prove the equality between the discrete Riemannian Delaunay complex and the Riemannian Delaunay complex by incrementally relaxing these assumptions.
In this section, we restrict ourselves to the case where the metric field is the Euclidean metric $g_{\E}$.
To simplify matters, we initially assume that geodesic distances are computed exactly on the canvas.
The following theorem gives sufficient conditions to have equality of the complexes.
\begin{theorem} \label{theorem-basic_anyD}
Assume that $\mathcal{P}$ is a $\delta$-power protected $(\varepsilon,\mu)$-net of $\Omega$ with respect to the Euclidean metric field $g_{\E}$.
Denote by $\canvas$ the canvas, a triangulation with maximal edge length $e_{\canvas}$.
If $e_{\canvas} < \min \left\{ \mu / 16, \delta^2 / 64\varepsilon \right\}$, then $\del_{\E}^{\textrm{d}}(\mathcal{P}) = \del_{\E}(\mathcal{P})$.
\end{theorem}

We shall now prove Theorem~\ref{theorem-basic_anyD} by enforcing the two following conditions which, combined, give the equality between the discrete Riemannian Delaunay complex and the Riemannian Delaunay complex:
\begin{enumerate}[label={(\arabic*)}]
\item for every Voronoi vertex in the Riemannian Voronoi diagram $v = \cap_{\{p_i\}} V_g(p_i)$, there exists at least one canvas simplex with the corresponding colors $\{ c_{p_i} \}$; \label{enum-first_cond_anyD}
\item no canvas simplex witnesses a simplex that does not belong to the Riemannian Delaunay complex (equivalently, no canvas simplex has vertices whose colors are those of non-adjacent Riemannian Voronoi cells). \label{enum-second_cond_anyD}
\end{enumerate}

Condition~\ref{enum-second_cond_anyD} is a consequence of the separation of Voronoi objects, which in turn follows from power protection.
The separation of Voronoi objects has previously been studied, for example by Boissonnat et al.~\cite{DBLP:journals/corr/BoissonnatDG13a}.
Although the philosophy is the same, our setting is slightly more difficult and the results using power protection are new and use a more geometrical approach (see Appendix~\ref{appendix-separation}).

\subsection{Sperner's lemma}
Rephrasing Condition~\ref{enum-first_cond_anyD}, we seek requirements on the density of the canvas $\canvas$ and on the nature of the point set $\mathcal{P}$ such that there exists at least one canvas $n$-simplex of $\canvas$ that has exactly the colors $c_0, \dots, c_d$ of the vertices $p_0, \dots, p_d$ of a simplex $\sigma$, for all $\sigma \in \DC$.
To prove the existence of such a canvas simplex, we employ Sperner's lemma~\cite{sperner1980fifty}, which is a discrete analog of Brouwer's fixed point theorem.
We recall this result in Theorem~\ref{theorem-Sperner} and illustrate it in a two-dimensional setting (inset). $\phantom{555555555555555}$
%Figure~\ref{fig-Sperner_illustration}
\begin{wrapfigure}[7]{r}{0.30\textwidth}
  \centering
  \includegraphics[width= \linewidth]{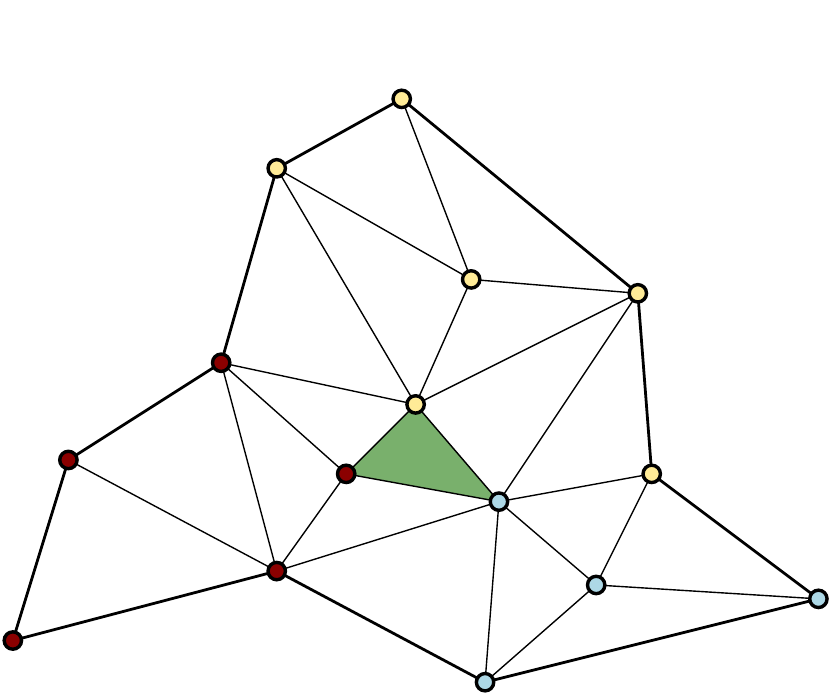}
  %\caption{Sperner's lemma in $2$D: the green triangle is colored with all $3$ colors.}
  \label{fig-Sperner_illustration}
\end{wrapfigure}

\begin{theorem}[Sperner's lemma] \label{theorem-Sperner}
~\\
Let $\sigma = (p_0 ,\ldots, p_n)$ be an $n$-simplex and let $T_\sigma$ denote a triangulation of the simplex.
Let each vertex $v' \in T_\sigma$ be colored such that the following conditions are satisfied:
\begin{itemize}
\item The vertices $p_i$ of $\sigma$ all have different colors.
\item If a vertex $p'$ lies on a $k$-face $(p_{i_0}, \ldots p_{i_k})$ of~$\sigma$, then $p' $ has the same color as one of the vertices of the face, that is $p_{i_j}$.
\end{itemize}
Then, there exists an odd number of simplices in $T_\sigma$ whose vertices are colored with all $n+1$ colors.
In particular, there must be at least one.
\end{theorem}

We shall apply Sperner's lemma to the canvas $\canvas$ and show that for every Voronoi vertex~$v$ in the Riemannian Voronoi diagram, we can find a subset $\canvas_v$ of the canvas that fulfills the assumptions of Sperner's lemma, hence obtaining the existence of a canvas simplex in $\canvas_v$ (and therefore in $\canvas$) that witnesses $\sigma_v$.
Concretely, the subset $\canvas_v$ is obtained in two steps:
\begin{itemize}
  \item[--] We first apply a barycentric subdivision of the Riemannian Voronoi cells incident to~$v$.
  From the resulting set of simplices, we extract a triangulation $\mathcal{T}_{v}$ composed of the simplices incident to $v$ (Section~\ref{section-building_pi}).
  \item[--] We then construct the subset $\canvas_v$ by overlaying the border of $\mathcal{T}_v$ and the canvas (Section~\ref{section-canvas_v}).
\end{itemize}
We then show that if the canvas simplices are small enough -- in terms of edge length -- then~$\canvas_v$ is the triangulation of a simplex that satisfies the assumptions of Sperner's lemma.

The construction of $\canvas_v$ is detailed in the following sections and illustrated in Figure~\ref{fig-Sperner_canvas_example}: starting from a colored canvas (left), we subdivide the incident Voronoi cells of $v$ to obtain $\mathcal{T}_v$ (middle), and deduce the set of canvas simplices~$\canvas_v$ which forms a triangulation that satisfies the hypotheses of Sperner's lemma, thus giving the existence of a canvas simplex (in green, right) that witnesses the Voronoi vertex within the union of the simplices, and therefore in the canvas.

\begin{figure}[!htb]
  \centering
  \includegraphics[width=\textwidth]{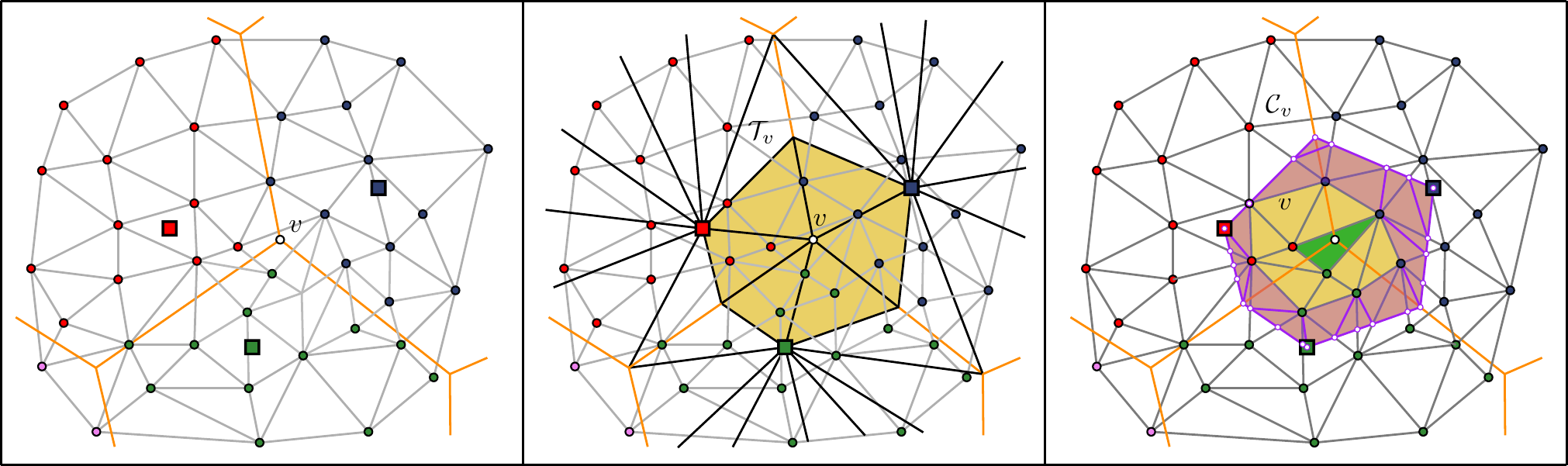}
  \caption{Illustration of the construction of $\canvas_v$.
           The Riemannian Voronoi diagram is drawn with thick orange edges and the sites are colored squares.
           The canvas is drawn with thin gray edges and colored circular vertices.
           The middle frame shows the subdivision of the incident Voronoi cells with think black edges and the triangulation $\mathcal{T}_v$ is drawn in yellow.
           On the right frame, the set of simplices $\canvas_v$ is colored in purple (simplices that do not belong to $\canvas$) and in dark yellow (simplices that belong to $\canvas$).}
  \label{fig-Sperner_canvas_example}
\end{figure}

\subsection{The triangulation \boldmath{$\mathcal{T}_v$}} \label{section-building_pi}
For a given Voronoi vertex $v$ in the Euclidean Voronoi diagram $\RVD_{\E}(\mathcal{P})$ of the domain $\Omega$, the initial triangulation $\mathcal{T}_{v}$ is obtained by applying a combinatorial barycentric subdivision of the Voronoi cells of $\RVD_{\E}(\mathcal{P})$ that are incident to $v$: to each Voronoi cell $V$ incident to $v$, we associate to each face $F$ of $V$ a point $c_{F}$ in $F$ which is not necessarily the geometric barycenter.
We randomly associate to $c_F$ the color of any of the sites whose Voronoi cells intersect to give $F$.
For example, in a two-dimensional setting, if the face $F$ is a Voronoi edge that is the intersection of $V_{red}$ and $V_{blue}$, then $c_{F}$ is colored either red or blue.
Then, the subdivision of $V$ is computed by associating to all possible sequences of faces $\{F_0, F_1, \dots F_{n-1}, F_n\}$ such that $F_0 \subset F_1 \dots \subset F_n = V$ and $\dim (F_{i+1}) = \dim(F_i) + 1$ the simplex with vertices $\{c_{F_0}, c_{F_1}, \dots, c_{F_{n-1}}, c_{F_n} \}$.
These barycentric subdivisions are allowed since Voronoi cells are convex polytopes.

Denote by $\Sigma_V$ the set of simplices obtained by barycentric subdivision of $V$ and $\Sigma_v = \{ \cup \Sigma_V \mid v\in V \}$.
The triangulation~$\mathcal{T}_v$ is defined as the star of $v$ in $\Sigma_v$, that is the set of simplices in $\Sigma_v$ that are incident to $v$.
$\mathcal{T}_v$ is illustrated in Figure~\ref{fig-Sperner_construction} in dimension $3$.
As shall be proven in Lemma~\ref{lemma-Tv_sperner}, $\mathcal{T}_{v}$ can be used to define a combinatorial simplex that satisfies the assumptions of Sperner's lemma.

\paragraph*{\boldmath{$\mathcal{T}_v$} as a triangulation of an $n$-simplex}
\begin{wrapfigure}[15]{r}{0.5\textwidth}
  \centering
  \includegraphics[height=.7\linewidth]{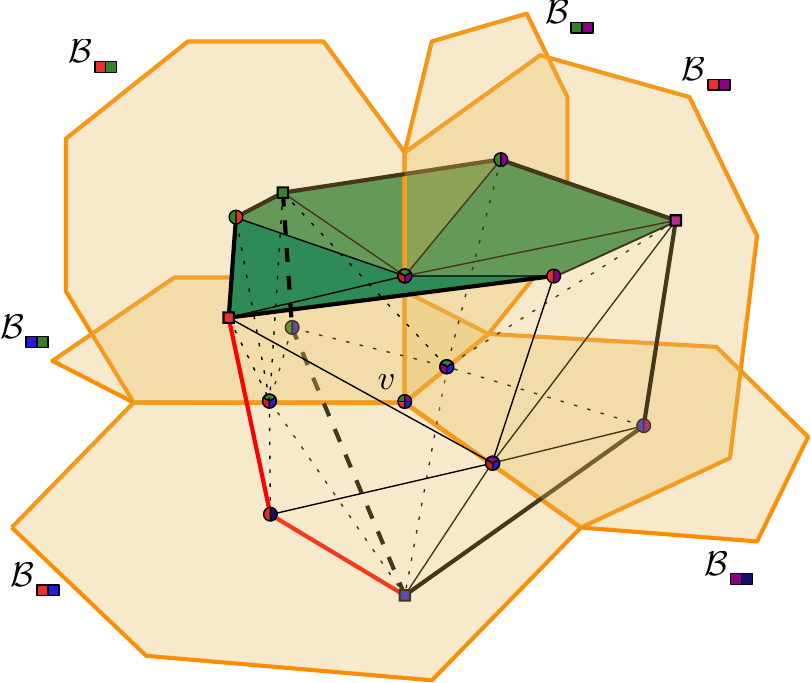}
  \caption{The triangulation $\mathcal{T}_v$ in $3$D.
          A face (in green) and an edge (in red) of $\sigma_{\mathcal{S}}$.}
  \label{fig-Sperner_construction}
\end{wrapfigure}

By construction, the triangulation $\mathcal{T}_v$ is a triangulation of the (Euclidean) Delaunay simplex $\sigma_v$ dual of $v$ as follows.
We first perform the standard barycentric subdivision on this Delaunay simplex $\sigma_v$.
We then map the barycenter of a $k$-face $\tau$ of $\sigma_v$ to the point $c_{F_i}$ on the Voronoi face $F_i$, where $F_i$ is the Voronoi dual of the $k$-face~$\tau$.
This gives a piecewise linear homeomorphism from the Delaunay simplex $\sigma_v$ to the triangulation~$\mathcal{T}_v$.
We call the image of this map the simplex $\sigma_{\mathcal{S}}$ and refer to the images of the faces of the Delaunay simplex as the faces of $\sigma_{\mathcal{S}}$.
We can now apply Sperner's lemma.

\begin{lemma} \label{lemma-Tv_sperner}
Let $\mathcal{P}$ be a $\delta$-power protected $(\varepsilon, \mu)$-net.
Let $v$ be a Voronoi vertex in the Euclidean Voronoi diagram, $\RVD_{\E}(\mathcal{P})$, and let $\Sigma_v$ be defined as above.
The simplex~$\sigma_{\mathcal{S}}$ and the triangulation~$\mathcal{T}_v$ satisfy the assumptions of Sperner's lemma in dimension $n$.
\end{lemma}
\begin{proof}
By the piecewise linear map that we have described above, $\mathcal{T}_v$ is a triangulation of the simplex $\sigma_{\mathcal{S}}$.
Because by construction the vertices $c_{F_i}$ lie on the Voronoi duals $F_i$ of the corresponding Delaunay face $\tau$, $c_{F_i}$ has the one of the colors of of the Delaunay vertices of~$\tau$.
Therefore, $\sigma_{\mathcal{S}}$ satisfies the assumptions of Sperner's lemma and there exists an \mbox{$n$-simplex} in~$\mathcal{T}_v$ that witnesses $v$ and its corresponding simplex $\sigma_v$ in~$\DC$.
\end{proof}

\subsection{Building the triangulation \boldmath{$\canvas_v$}} \label{section-canvas_v}
Let $p_i$ be the vertices of the $k$-face~$\tau_{\mathcal{S}}$ of $\sigma_{\mathcal{S}}$. 
In this section we shall assume not only that~$\tau_{\mathcal{S}}$ is contained in the union of the Voronoi cells of $V(p_i)$, but in fact that $\tau_{\mathcal{S}}$ is a distance~$8 e_{\canvas}$ removed from the boundary of $\cup V(p_i)$, where $e_{\canvas}$ is the longest edge length of a simplex in the canvas.
We will now construct a triangulation $\canvas_v$ of $\sigma_{\mathcal{S}}$ such that:
\begin{itemize}
\item $\sigma_{\mathcal{S}}$ and its triangulation $\canvas_v$ satisfy the conditions of Sperner's lemma,
\item the simplices of $\canvas_v$ that have no vertex that lies on the boundary $\partial \sigma_{\mathcal{S}}$ are simplices of the canvas $\canvas$.
\end{itemize}

The construction goes as follows.
We first intersect the canvas $\canvas$ with $\sigma_{\mathcal{S}}$ and consider the canvas simplices $\sigma_{\canvas,i}$ such that the intersection of $\sigma_{\mathcal{S}}$ and $\sigma_{\canvas,i}$ is non-empty.
These simplices $\sigma_{\canvas,i}$ can be subdivided into two sets, namely those that lie entirely in the interior of $\sigma_{\mathcal{S}}$, which we denote by $\sigma_{\mathcal{C},i}^{\textrm{int}}$, and those that intersect the boundary, denoted by $\sigma_{\mathcal{C},i}^{\partial}$.

The simplices $\sigma_{\mathcal{C},i}^{\textrm{int}}$ are added to the set $\canvas_v$.
We intersect the simplices $\sigma_{\canvas, i}^{\partial}$ with $\sigma_{\mathcal{S}}$ and triangulate the intersection.
Note that $\sigma_{\canvas, i}^{\partial} \cap \sigma_{\mathcal{S}}$ is a convex polyhedron and thus triangulating it is not a difficult task.
The vertices of the simplices in the triangulation of $\sigma_{\mathcal{C},i}^{\partial} \cap \sigma_{\mathcal{S}}$ are colored according to which Voronoi cell they belong to.
Finally, the simplices in the triangulation of $\sigma_{\canvas, i}^{\partial} \cap \sigma_{\mathcal{S}}$ are added to the set $\canvas_v$.

Since $\mathcal{T}_v$ is a triangulation of $\sigma_{\mathcal{S}}$, the set $\canvas_v$ is by construction also a triangulation of $\sigma_{\mathcal{S}}$. This triangulation trivially gives a triangulation of the faces $\tau_{\mathcal{S}}$. Because we assume that~$\tau_{\mathcal{S}}$ is contained in the union of its Voronoi cells, with a margin of $8 e_\canvas$ we now can draw two important conclusions: 
\begin{itemize}
\item The vertices of the triangulation of each face $\tau_{\mathcal{S}}$ have the colors of the vertices $p_i$ of $\tau_{\mathcal{S}}$. 
\item None of the simplices in the triangulation of $\sigma_{\mathcal{C},i}^{\partial} \cap \sigma_{\mathcal{S}}$ can have $n+1$ colors, because every such simplex must be close to one face $\tau_{\mathcal{S}}$, which means that it must be contained in the union of the Voronoi cells $V(p_i)$ of the vertices of $\tau_{\mathcal{S}}$. 
\end{itemize}

We can now invoke Sperner's lemma; $\canvas_v$ is a triangulation of the simplex $\sigma_{\mathcal{S}}$ whose every face has been colored with the appropriate colors (since $\sigma_{\mathcal{S}}$ triangulated by $\mathcal{T}_v$ satisfies the assumptions of Sperner's lemma, see Lemma~\ref{lemma-Tv_sperner}).
This means that there is a simplex $\canvas_v$ that is colored with $n+1$ colors.
Because of our second observation above, the simplex with these $n+1$ colors must lie in the interior of $\sigma_{\mathcal{S}}$ and is thus a canvas simplex.

We summarize by the following lemma:
\begin{lemma} \label{lemma-Cv_sperner}
If every face $\tau_{\mathcal{S}}$ of  $\sigma_{\mathcal{S}}$ with vertices $p_i$ is at distance $8 e_{\canvas}$ from the boundary of the union of its Voronoi cells $\partial ( \cup V(p_i))$, then there exists a canvas simplex in $\canvas_v$ such that it is colored with the same vertices as the vertices of $\sigma_{\mathcal{S}}$.
\end{lemma}

The key task that we now face is to guarantee that faces $\tau_{\mathcal{S}}$ indeed lie well inside of the union of the appropriate Voronoi regions.
This requires first and foremost power protection.
Indeed, if a point set is power protected, the distance between a Voronoi vertex $c$ and the Voronoi faces that are not incident to $c$, which we will refer to from now on as \emph{foreign} Voronoi faces, can be bounded, as shown in the following Lemma:
\begin{lemma} \label{lemma-Protection_Of_Circumcenters}
Suppose that $c$ is the circumcenter of a $\delta$-power protected simplex $\sigma$ of a Delaunay triangulation built from an $\varepsilon$-sample, then all foreign Voronoi faces are at least $\delta^2 / 8 \varepsilon$ far from $c$.
\end{lemma}
The proof of this Lemma is given in the appendix (Section~\ref{section-vor_faces_separation}).

In almost all cases, this result gives us the distance bound we require: we can assume that vertices $\{c_{F_0}, c_{F_1}, \dots, c_{F_{n-1}}, c_{F_n} \}$ which we used to construct $\mathcal{T}_v$, are well placed, meaning that there is a minimum distance between these vertices and foreign Voronoi objects.
However it can still occur that foreign Voronoi objects are close to a face~$\tau_{\mathcal{S}}$ of $\sigma_{\mathcal{S}}$.
This occurs even in two dimensions, where a Voronoi vertex $v'$ can be very close to a face $\tau_{\mathcal{S}}$ because of obtuse angles, as illustrated in Figure~\ref{fig-ctau_problem}.
\begin{figure}[!htb]
  \centering
  \includegraphics[width=\linewidth]{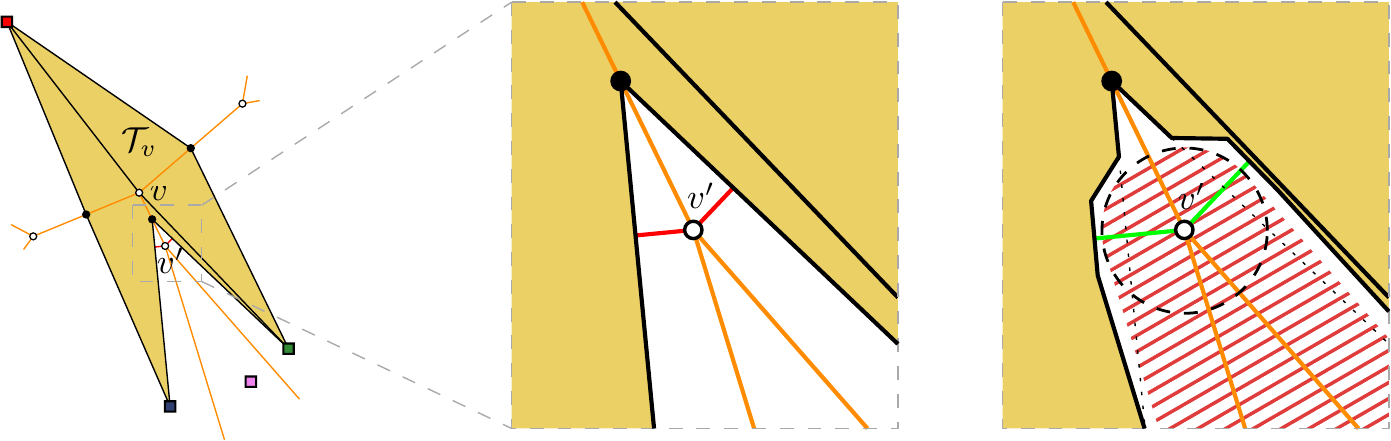}
  \caption{The point $v'$ can be arbitrarily close to $\mathcal{T}_v$, as shown by the red segments (left and center).
           After piecewise linear deformation, this issue is resolved, as seen by the green segments (right). } \label{fig-ctau_problem}
\end{figure}

Thanks to power protection, we know that $v'$ is removed from foreign Voronoi objects.
This means that we can deform $\sigma_{\mathcal{S}}$ (in a piecewise linear manner) in a neighborhood of $v'$ such that the distance between $v'$ and all the faces of the deformed $\sigma_{\mathcal{S}}$ is lower bounded.

In general the deformation of $\sigma_{\mathcal{S}}$ is performed by ``radially pushing'' simplices away from the foreign Voronoi faces of $v$ with a ball of radius $r = \min \left\{\mu / 16, \delta^2 / 64\varepsilon \right\}$.
The value~$\mu/16$ is chosen so that we do not move any vertex of $\sigma_v$ (the dual of $v$): indeed, $\mathcal{P}$ is $\mu$-separated and thus $d_{\E}(p_i, p_j) > \mu$.
The value $\delta^2 / 64\varepsilon$ is chosen so that $\sigma_{\mathcal{S}}$ and its deformation stay isotopic (no ``pinching'' can happen), using Lemma~\ref{lemma-Protection_Of_Circumcenters}.
In fact it is advisable to use a piecewise linear version of ``radial pushing'', to ensure that the deformation of $\sigma_{\mathcal{S}}$ is a polyhedron. This guarantees that we can triangulate the intersection, see Chapter~$2$ of Rourke and Sanderson~\cite{rourke2012introduction}.
After this deformation we can follow the steps we have given above to arrive at a well-colored simplex.

\begin{lemma} \label{theorem-Euclidean_capture}
Let $\mathcal{P}$ be a $\delta$-power protected $(\varepsilon, \mu)$-net.
Let $v$ be a Voronoi vertex of the Euclidean Voronoi diagram $\RVD_{\E}(\mathcal{P})$, and $\mathcal{T}_v$ as defined above.
If the length $e_{\canvas}$ of the longest canvas edge is bounded as follows: $e_{\canvas} < r = \min \left\{ \mu / 16, \delta^2 / 64\varepsilon \right\}$, then there exists a canvas simplex that witnesses $v$ and the corresponding simplex $\sigma_v$ in $\del_{\E}(\mathcal{P})$.
\end{lemma}

\paragraph*{Conclusion}
So far, we have only proven that $\DC \subseteq \DDC$.
The other inclusion, which corresponds to Condition~\ref{enum-second_cond_anyD} mentioned above, is much simpler: as long as a canvas edge is shorter than the smallest distance between a~Voronoi vertex and a foreign face of the Riemannian Voronoi diagram, then no canvas simplex can witness a simplex that is not in $\DC$.
Such a bound is already given by Lemma~\ref{lemma-Protection_Of_Circumcenters} and thus, if $e_{\canvas} < \delta^2 / 8 \varepsilon$ then $\DDC \subseteq \DC$.
Observe that this requirement is weaker than the condition imposed in Lemma~\ref{theorem-Euclidean_capture} and it was thus already satisfied.
It follows that $\DDC = \DC$ if $e_{\canvas} < \min \left\{ \mu / 16, \delta^2 / 64\varepsilon \right\}$, which concludes the proof of Theorem~\ref{theorem-basic_anyD}.

\begin{remark}
Assuming that the point set is a $\delta$-power protected $(\varepsilon, \mu)$-net might seem like a strong assumption.
However, it should be observed that any non-degenerate point set can be seen as a $\delta$-power protected $(\varepsilon, \mu)$-net, for a sufficiently large value of $\varepsilon$ and sufficiently small values of $\delta$ and $\mu$.
Our results are therefore always applicable but the necessary canvas density increases as the quality of the point set worsens (Lemma~\ref{theorem-Euclidean_capture}).
In our companion practical paper~\cite[Section~$7$]{rouxel2016discretized}, we showed how to generate $\delta$-power protected $(\varepsilon, \mu)$-nets for given values of $\varepsilon$, $\mu$ and $\delta$.
\end{remark}

\section{Extension to more complex settings} \label{section-extension}
In the previous section, we have placed ourselves in the setting of an
(open) domain endowed with the Euclidean metric field. % and assumed exact geodesic computations.
To prove Theorem~\ref{theorem-generic_anyD}, we need to  generalize
Theorem~\ref{theorem-basic_anyD} to more general metrics, which will
be done in the two following subsections.

The common path to prove $\DDC = \DC$ in all settings is to assume that $\mathcal{P}$ is a power protected net with respect to the metric field.
We then use the stability of entities under small metric perturbations to take us back to the now solved case of the domain $\Omega$ endowed with an Euclidean metric field.
Separation and stability of Delaunay and Voronoi objects has previously been studied by Boissonnat et al.~\cite{DBLP:journals/corr/BoissonnatDG13a, DBLP:journals/corr/BoissonnatDGO14}, but our work lives in a slightly more complicated setting.
Moreover, our proofs are generally more geometrical and sometimes simpler.
For completeness, the extensions of these results to our context are detailed in Appendix~\ref{appendix-separation} for separation, and in Appendix~\ref{appendix-stability} for stability.

We now detail the different intermediary settings.
For completeness, the full proofs are included in the appendices.

\subsection{Uniform metric field}
We first consider the rather easy case of a non-Euclidean but uniform (constant) metric field over an (open) domain.
The square root of a metric gives a linear transformation between the base space where distances are considered in the metric and a \emph{metric space} where the Euclidean distance is used (see Appendix~\ref{section-metric_transformation}).
Additionally, we show that a $\delta$-power protected $(\varepsilon, \mu)$-net with respect to the uniform metric is, after transformation, still a $\delta$-power protected $(\varepsilon, \mu)$-net but with respect to the Euclidean setting (Lemma~\ref{lemma-protected_net_metric_transformation} in Appendix~\ref{appendix-stability}), bringing us back to the setting we have solved in Section~\ref{section-complex_equality}.
Bounds on the power protection, sampling and separation coefficients, and on the canvas edge length can then be obtained from the result for the Euclidean setting, using Theorem~\ref{theorem-Euclidean_capture}.
These bounds can be transported back to the case of uniform metric fields by scaling these values according to the smallest eigenvalue of the metric (Theorem~\ref{theorem-uniform_metric_sizing_field_anyD} in Appendix~\ref{appendix-extension}).

\subsection{Arbitrary metric field}
The case of an arbitrary metric field over $\Omega$ is handled by observing that an arbitrary metric field is locally well-approximated by a uniform metric field.
It is then a matter of controlling the distortion.

We first show that, for any  point $p\in\Omega$, density separation and power protection are locally preserved in a neighborhood $U_p$ around $p$ when the metric field $g$ is approximated by the constant metric field $g'=g(p)$ (Lemmas~\ref{lemma-net_metric_perturbation} and~\ref{lemma-protection_metric_perturbation} in Appendix~\ref{appendix-stability}): if $\mathcal{P}$ is a $\delta$-power protected $(\varepsilon, \mu)$-net with respect to $g$, then $\mathcal{P}$ is a $\delta'$-power protected $(\varepsilon', \mu')$-net with respect to $g'$.
Previous results can now be applied to obtain conditions on $\delta'$, $\varepsilon'$, $\mu'$ and on the (local) maximal length of the canvas such that $\DDC = \DC$ (Lemma~\ref{lemma-generic_single_cell_anyD} in Appendix~\ref{appendix-extension}).

These local triangulations can then be stitched together to form a triangulation embedded in $\Omega$.
The (global) bound on the maximal canvas edge length is given by the
minimum of the local bounds, each computed through the results of the
previous sections. This ends the proof of Theorem~\ref{theorem-generic_anyD}.

Once the equality between the complexes is obtained, conditions giving the embeddability of the discrete Karcher Delaunay triangulation and the discrete straight Delaunay triangulation are given by previous results that we have established in Sections~\ref{section-embeddability_curved} and~\ref{section-embeddability_straight} respectively.

\section{Extensions of the main result}
{\bf Approximate geodesic computations}
Approximate geodesic distance computations can be incorporated in the analysis of the previous section by observing that computing inaccurately geodesic distances in a domain $\Omega$ endowed with a metric field $g$ can be seen as computing exactly geodesic distances in $\Omega$ with respect to a metric field $g'$ that is close to~$g$ (Section~\ref{section-approx_geo} in Appendix~\ref{appendix-extension}). \newline
~\\
\textbf{General manifolds}
The previous section may also be generalized to an arbitrary smooth $n$-manifold $\mathcal{M}$ embedded in $\R^m$.
We shall assume that, apart from the metric induced by the embedding of the domain in Euclidean space, there is a second metric $g$ defined on $\mathcal{M}$.
Let $\pi_p: \mathcal{M} \to T_p \mathcal{M}$ be the orthogonal projection of points of $\mathcal{M}$ on the tangent space $T_p \mathcal{M}$ at $p$.
For a sufficiently small neighborhood $U_p\subset T_p \mathcal{M}$, $\pi_p$ is a local diffeomorphism (see Niyogi~\cite{niyogi2008}).

Denote by $\mathcal{P}_{T_p}$ the point set $\{ \pi_p(p_i), p_i\in\mathcal{P} \}$ and $\mathcal{P}_{U_p}$ the restriction of $\mathcal{P}_{T_p}$ to $U_p$.
Assuming that the conditions of Niyogi et al.~\cite{niyogi2008} are satisfied (which are simple density constraints on $\varepsilon$ compared to the reach of the manifold), the pullback of the metric with the inverse projection $(\pi_p^{-1})^{*} g$ defines a metric $g_p$ on $U_p$ such that for all $q, r\in U_p$, $d_{g_p}(q,r) = d_g( \pi_p^{-1}(q), \pi_p^{-1}(r))$.
This implies immediately that if $\mathcal{P}$ is a $\delta$-power protected $(\varepsilon,\mu)$-net on $\mathcal{M}$ with respect to $g$ then $\mathcal{P}_{U_p}$ is a $\delta$-power protected $(\varepsilon,\mu)$-net on $U_p$.
We have thus a metric on a subset of a $n$-dimensional space, in this case the tangent space, giving us a setting that we have already solved.
It is left to translate the sizing field requirement from the tangent plane to the manifold $\mathcal{M}$ itself.
Note that the transformation $\pi_p$ is completely independent of $g$.
Boissonnat et al.~\cite[Lemma 3.7]{DBLP:journals/corr/BoissonnatDG13a} give bounds on the metric distortion of the projection on the tangent space.
This result allows to carry the canvas sizing field requirement from the tangent space to $\mathcal{M}$.

\section{Implementation} \label{sec:implementation}
The construction of the discrete Riemannian Voronoi diagram and of the discrete Riemannian Delaunay complex has been implemented for $n=2, 3$ and for surfaces of $\R^3$. An in-depth description of our structure and its construction as well as an empirical study can be found in our practical paper~\cite{rouxel2016discretized}.
We simply make a few observations here.

%\paragraph*{Generation of the canvas}
The theoretical bounds on the canvas edge length provided by Theorems~\ref{theorem-generic_anyD} and~\ref{theorem-basic_anyD} are far from tight and thankfully do not need to be honored in practice. A canvas whose edge length are about a tenth of the distance between two seeds suffices. This creates nevertheless unnecessarily dense canvasses since the density does not in fact need to be equal everywhere at all points and even in all directions.
This issue is resolved by the use of anisotropic canvasses. 

%We instead generate the canvas through a classical Delaunay refinement algorithm driven by a sizing field: the circumradius of the (canvas) Delaunay simplices is bounded by a given value $r_0$ whose value is usually taken as a fraction of the smallest distance between any two sites.
%This creates nevertheless unnecessarily dense canvasses since the density does not in fact need to be equal everywhere at all points and even in all directions.
%This issue is resolved with the use of anisotropic canvasses.
%
%\paragraph*{Generation of the sites}
%The generation of the sites is generally part of the problem when building anisotropic triangulations.
%The discrete Riemannian Voronoi diagram can be employed to generate point sets, for example through a farthest point refinement algorithm: given a bound $r_0$, as long there exists canvas vertices whose distance to a site is greater than $r_0$, the canvas vertex with farthest distance to a site is inserted in the set of sites and the diagram is updated.
%For the canvas to be sufficiently dense to ensure the equality of the Delaunay complexes once the refinement process ends, we evaluate the necessary density of the canvas beforehand, from the knowledge of the metric field and of $r_0$.

%\paragraph*{Discrete structures}
Our analysis was based on the assumption that all canvas vertices are painted with the color of the closest site.
In our implementation, we color the canvas using a multiple-front vector Dijkstra algorithm~\cite{Campen:2013:PAG:2600289.2600299}, which empirically does not suffer from the same convergence issues as the traditional Dijkstra algorithm, starting from all the sites.
It should be noted that any geodesic distance computation method can be used, as long as it converges to the exact geodesic distance when the canvas becomes denser.
The Riemannian Delaunay complex is built on the fly during the construction of the discrete Riemannian Voronoi diagram: when a canvas simplex is first fully colored, its combinatorial information is extracted and the corresponding simplex is added to $\del_g(\mathcal{P})$.

\section*{Acknowledgments}
We thank Ramsay Dyer for enlightening discussions. The first and third authors have received funding from the European Research Council under the European Union's ERC Grant Agreement number 339025 GUDHI (Algorithmic Foundations of Geometric Understanding in Higher Dimensions).

\bibliographystyle{acm}
\bibliography{socg}

\newpage
\appendix

\section*{Overview of the appendices}
The topics of the appendices are as follows:
\begin{itemize}
\item[\ref{appendix-simplices_and_complexes}] We review basic definitions related to simplices and complexes.
\item[\ref{appendix-distortion_properties}] We describe our generalization of the notion of distortion between metrics due to Labelle and Shewchuk.
\item[\ref{appendix-separation}] We discuss the separation of Voronoi objects.
The main differences between this appendix and \cite{DBLP:journals/corr/BoissonnatDG13a} are in our definition of metric distortion, the use of power protection, and the more geometrical nature of the proofs.
\item[\ref{appendix-dihedral_angles}] This appendix is related to the previous one and focuses on dihedral angles of Delaunay simplices.
These results are intermediary steps used in Appendices~\ref{appendix-stability} and~\ref{appendix-extension}.
\item[\ref{appendix-stability}] Here we built upon the previous two sections and discuss the stability of nets and Voronoi cells.
This section distinguishes itself by the elementary and geometrical nature of the proofs.
\item[\ref{appendix-proofs_straight_embeddability}] We prove that the Delaunay simplices can be straightened, under sufficient conditions.
The proof of the stability of the center of mass, which forms the core of this appendix, is also remarkable in the sense that it generalizes trivially to a far more general setting.
\item[\ref{appendix-deforming}] We illustrate a degenerate case of Section~\ref{section-canvas_v}.
\item[\ref{appendix-extension}] This appendix gives the proofs for our main result, Theorem \ref{theorem-generic_anyD}, in the general setting of an arbitrary metric.
We naturally rely heavily on Appendices~\ref{appendix-separation},~\ref{appendix-dihedral_angles},~\ref{appendix-stability}, and~\ref{appendix-extension}.
\end{itemize}
The references for the Appendix can be found at the end of the Appendix.

\section{Simplices and complexes} \label{appendix-simplices_and_complexes}
The purpose of this section is to offer the precise definitions of concepts and notions related to simplicial complexes.
The following definitions live within the context of abstract simplices and complexes.

A \emph{simplex} $\sigma$ is a non-empty finite set.
The \emph{dimension} of $\sigma$ is given by $\operatorname{dim} \sigma = \# (\sigma) - 1$, and a $j$-simplex refers to a simplex of dimension $j$.
The elements of $\sigma$ are called the \emph{vertices} of $\sigma$.
The set of vertices of $\sigma$ is noted $\Vertices(\sigma)$.

If a simplex $\tau$ is a subset of $\sigma$ , we say it is a \emph{face} of $\sigma$ , and we write $\tau \leq \sigma$.
A $1$-dimensional face is called an \emph{edge}.
If $\tau$ is a proper subset of $\sigma$ , we say it is a \emph{proper} face and we write $\tau < \sigma$.
A \emph{facet} of $\sigma$ is a face $\tau$ with dim $\tau = \operatorname{dim} \sigma - 1$.

For any vertex $p \in \sigma$, the face opposite to $p$ is the face determined by the other vertices of $\sigma$, and is denoted by~$\sigma_p$.
If $\tau$ is a $j$-simplex, and~$p$ is not a vertex of~$\sigma$, we may construct a~$(j + 1)$-simplex $\sigma = p \ast \tau$, called the \emph{join} of~$p$ and~$\tau$ as the simplex defined by $p$ and the vertices of~$\tau$.

The \emph{length} of an edge is the distance between its vertices.
The \emph{height} of~$p$ in~$\sigma$ is~$D(p, \sigma) = d(p, \aff{\sigma_p})$.

A \emph{circumscribing ball} for a simplex $\sigma$ is any $n$-dimensional ball that contains the vertices of $\sigma$ on its boundary.
If $\sigma$ admits a circumscribing ball, then it has a circumcenter, $C(\sigma)$, which is the center of the unique smallest circumscribing ball for $\sigma$.
The radius of this ball is the \emph{circumradius} of $\sigma$, denoted by $R(\sigma)$.

\subsection{Complexes}
Before defining Delaunay triangulations, we introduce the more general concept of simplicial complexes.
Since the standard definition of a simplex as the convex hull of a set of points does not extend well to the Riemannian setting (see Dyer~\cite{dyer2014riemsplx.arxiv}), we approach these definitions from a more abstract point of view.

The \emph{length} of an edge is the distance between its vertices.
A \emph{circumscribing ball} for a simplex $\sigma$ is any $n$-dimensional ball that contains the vertices of $\sigma$ on its boundary.
If $\sigma$ admits a circumscribing ball, then it has a circumcenter, $C(\sigma)$, which is the center of the unique smallest circumscribing ball for $\sigma$.
The radius of this ball is the \emph{circumradius} of $\sigma$, denoted by $R(\sigma)$.
The \emph{height} of~$p$ in~$\sigma$ is~$D(p, \sigma) = d(p, \aff(\sigma_p))$.
The \emph{dihedral angle} between two facets is the angle between their two supporting planes.
If~$\sigma$ is a $j$-simplex with~$j \geq 2$, then for any two vertices~$p$,~$q\in\sigma$, the dihedral angle between~$\sigma_p$ and~$\sigma_q$ defines an equality between ratios of heights (see Figure~\ref{fig-dihedral_angles}).
\[ \sin\angle(\aff(\sigma_p), \aff(\sigma_q)) = \frac{D(p,\sigma)}{D(p,\sigma_q)} = \frac{D(q,\sigma)}{D(q,\sigma_p)}. \]

\begin{figure}[!htb]
\centering
\includegraphics[width=0.9\linewidth]{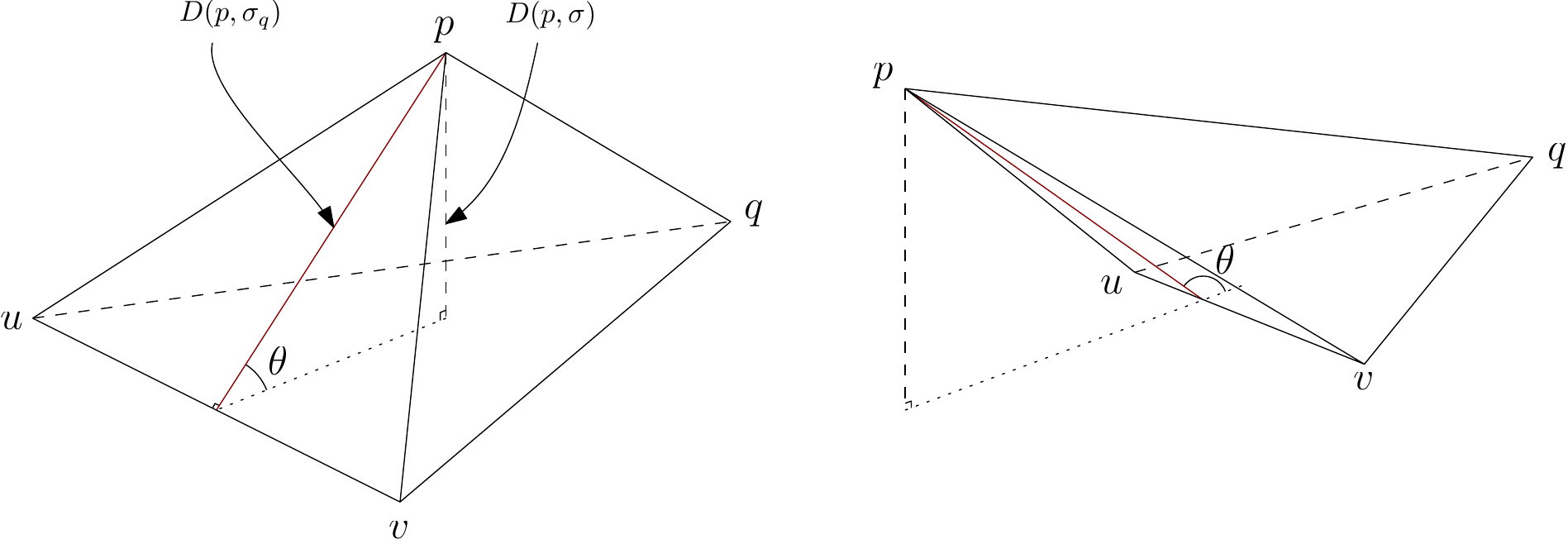}
\caption{Acute and obtuse dihedral angles}
\label{fig-dihedral_angles}
\end{figure}

\subsection{Simplicial complexes}
Simplicial complexes form the underlying framework of Delaunay triangulations.
An \emph{abstract simplicial complex} is a set $\mathcal{K}$ of simplices such that if $\sigma\in \mathcal{K}$, then all the faces of $\sigma$ also belong to $\mathcal{K}$.
The union of the vertices of all the simplices of $\mathcal{K}$ is called the \emph{vertex set} of $K$.
The dimension of a complex is the largest dimension of any of its simplices.
A subset $L \subseteq \mathcal{K}$ is a \emph{subcomplex} of $\mathcal{K}$ if it is also a complex.
Two simplices are \emph{adjacent} if they share a face and \emph{incident} if one is a face of the other.
If a simplex in $\mathcal{K}$ is not a face of a larger simplex, we say that it is \emph{maximal}.
If all the maximal simplices in a complex $\mathcal{K}$ of dimension $n$ have dimension $n$, then the simplicial complex is \emph{pure}.
The \emph{star} of a vertex $p$ in a complex $\mathcal{K}$ is the subcomplex $\str$ formed by set of simplices that are incident to $p$.
The \emph{link} of a vertex $p$ is the union of the simplices opposite of $p$ in $\str_p$.

A \emph{geometric simplicial complex} is an abstract simplicial complex whose faces are materialized, and to which the following condition is added: the intersection of any two faces of the complex is either empty or a face of the complex.

%%%%%%%%%%%%%%%%%%%%%%%%%%%%%%%%%%%%%%%%%%%%%%%%%%%%%%%%%%%%%%%%%%%%%%%%%%%%%%%%%%%%%%%%%%%%%%%%%%%%%%%%%%%%%%%%%%%%%%%%%%%%%%%%%%%%%%%%%%%%%%%%%%%%%%%%%%
\section{Geodesic distortion} \label{appendix-distortion_properties}
The concept of distortion was originally introduced by Labelle and Shewchuk~\cite{labelle2003} to relate distances with respect to two metrics, but this result can be (locally) extended to geodesic distances.

\subsection{Original distortion}
We first recall their definition and then show how to extend it to metric fields.
The notion of metric transformation is required to define this original distortion, and we thus recall it now.

\subsubsection{Metric transformation} \label{section-metric_transformation}
Given a symmetric positive definite matrix~$G$, we denote by~$F$ any matrix such that~$\operatorname{det}(F) > 0$ and $F^t F^{\phantom{t}} = G$.
The matrix~$F$ is called a \emph{square root} of~$G$.
The square root of a matrix is not uniquely defined: the Cholesky decomposition provides, for example, an upper triangular~$F$, but it is also possible to consider the diagonalization of $G$ as $O^T D O$, where~$O$ is an orthonormal matrix and~$D$ is a diagonal matrix; the square root is then taken as $F = O^T \sqrt{D} O$.
The latter definition is more natural than other decompositions since $\sqrt{D}$ is canonically defined and~$F$ is then also a symmetric definite positive matrix, with the same eigenvectors as~$G$.
Regardless of the method chosen to compute the square root of a metric, we assume that it is fixed and identical for all the metrics considered.

The square root~$F$ offers a linear bijective transformation between the Euclidean space and a metric space, noting that
\[ d_G(x,y) = \sqrt{(x-y)^t F^t F^\phant (x-y)} = \norm{F(x-y)} = \norm{F x - F y}, \]
where $\norm{\cdot}$ stands for the Euclidean norm.
Thus, the metric distance between $x$ and $y$ in Euclidean space can be seen as the Euclidean distance between two \emph{transformed} points $F x$ and $F y$ living in the metric space of $G$.

\subsubsection{Distortion}
The \emph{distortion} between two points $p$ and $q$ of $\Omega$ is defined by Labelle and Shewchuk~\cite{labelle2003} as $\psi(p, q) = \psi(G_p, G_q) = \max\left\{ \norm {F_p^\phant F_q ^{-1}}, \norm{F_q^\phant F_p^{-1} } \right\}$, where $\norm{\cdot}$ is the Euclidean matrix norm, that is $\norm{M} = \sup_{x\in\R^n} \frac{\norm{M x}}{\norm{x}}$.
Observe that $\psi(G_p, G_q) \geq 1$ and $\psi(G_p, G_q) = 1$ when $G_p = G_q$.

A fundamental property of the distortion is to relate the two distances $d_{G_p}$ and $d_{G_q}$.
Specifically, for any pair $x,y$ of points, we have
\begin{equation}
\frac{1}{\psi(p,q)}\, d_{G_p}(x,y) \leq d_{G_q}(x,y) \leq \psi(p,q)\, d_{G_p}(x,y). \label{equation-distortion}
\end{equation}
Indeed,
\[ d_{p}(x,y) = \Big\lVert F_p^\phant (x-y) \Big\rVert = \norm{F_p^\phant F_q^{-1} F_q^\phant (x-y)} \leq \norm{F_p^\phant F_q^{-1}} \Big\lVert F_q^\phant (x-y) \Big\rVert \leq \psi(G_p, G_q)\, d_q(x, y), \]
and the other inequality is obtained similarly.

\subsection{Geodesic distortion}
The previous definition can be defined to hold (locally) for metric fields instead of metric, as we show now.

\begin{lemma} \label{lemma-geodesic_distortion}
Let $U \subset \Omega$ be open, and $g$ and $g'$ be two Riemannian metric fields on $U$.
Let $\psi_0\geq 1$ be a bound on the metric distortion, in the sense of Labelle and Shewchuk.
Suppose that $U$ is included in a ball $B_g(p_0, r_0)$, with $p_0 \in U$ and $r_0 \in \R^{+}$, such that ${\forall p\in B(p_0,r_0)}, {\psi(g(p), g'(p)) \leq \psi_0}$.
Then, for all $x,y \in U$,
\[ \frac{1}{\psi_0} d_g (x,y) \leq d_{g^\prime} (x,y) \leq \psi_0\,d_g (x,y), \]
where $d_g$ and $d_{g'}$ indicate the geodesic distances with respect to~$g$ and~$g'$ respectively.
%So $\psi_0$ is also the distortion according our new definition
\end{lemma}
\begin{proof}
Recall that for $p\in B_g(p_0, r_0)$ and, for any pair $x,y$ of points, we have
\[ \frac{1}{\psi_0}\, d_{g(p)}(x,y) \leq d_{g'(p)}(x,y) \leq \psi_0\, d_{g(p)}(x,y). \]
Therefore, for any curve $\gamma(t)$ in $U$, we have that 
\[
\frac{1}{\psi_0} \int  \sqrt{\langle \dot{\gamma} ,\dot{\gamma} \rangle}_{g(\gamma(t))} \ud t
\leq \int \sqrt{ \langle \dot{\gamma} ,\dot{\gamma} \rangle}_{g'(\gamma(t))}  \ud t
\leq \psi_0 \int \sqrt{\langle \dot{\gamma} ,\dot{\gamma} \rangle}_{g(\gamma(t))}  \ud t.
\]
Considering the infimum over all paths $\gamma$ that begin at $x$ and end at $y$, we obtain the result.
\end{proof}

Note that this result is independent from the definition of the distortion and is entirely based on the inequality comparing distances in two metrics (Equation~\ref{equation-distortion}).

\section{Separation of Voronoi objects} \label{appendix-separation}
Power protected point sets were introduced to create quality bounds for the simplices of Delaunay triangulations built using such point sets~\cite{DBLP:journals/corr/BoissonnatDG13a}.
We will show that power protection allows to deduce additional useful results for Voronoi diagrams.
In this section we show that when a Voronoi diagram is built using a power protected sample set, its non-adjacent Voronoi faces, and specifically its Voronoi vertices are separated.
This result is essential to our proofs in Sections~\ref{section-complex_equality} and~\ref{section-extension} where we approximate complicated Voronoi cells with simpler Voronoi cells without creating inversions in the dual, which is only possible because we know that Voronoi vertices are far from one another.

We assume that the protection parameter~$\delta$ is proportional to the sampling parameter $\varepsilon$, thus there exists a positive $\iota$, with $\iota\leq 1$, such that $\delta = \iota \varepsilon$.
We assume that the separation parameter~$\mu$ is proportional to the sampling parameter~$\varepsilon$ and thus there exists a positive $\lambda$, with $\lambda\leq 1$, such that $\mu = \lambda \varepsilon$.

\subsection{Separation of Voronoi vertices}
The main result provides a bound on the Euclidean distance between Voronoi vertices of the Euclidean Voronoi diagram of a point set and is given by Lemma~\ref{lemma-global_separation_bound}.
The following three lemmas are intermediary results needed to prove Lemma~\ref{lemma-global_separation_bound}.

\begin{lemma} \label{lemma-sep_1}
Let $B = B(c,R)$ and $B^\prime = B(c^\prime, R^\prime)$ be two $n$-balls whose bounding spheres~$\partial B$ and~$\partial B^\prime$ intersect, and let~$H$ be the bisecting hyperplane of~$B$ and~$B^\prime$, \ie the hyperplane that contains the \mbox{$(n-2)$-sphere} $S = \partial B~\cap~\partial B^\prime$.
Let~$\theta$ be the angle of the cone~$(c, S)$.
Writing $\rho = \frac{R^\prime}{R}$ and $\norm{c-c^\prime} = \lambda R$, we have
\begin{align}
\cos(\theta) = \frac{1+\lambda^2-\rho^2}{2\lambda}.
\label{Form:lemma-sep_1}
\end{align} 
If $R \geq R^\prime$, we have $\cos(\theta) \geq \frac{\lambda}{2}.$
\end{lemma}
\begin{proof}
Let $q \in S$; applying the cosine rule to the triangle $\triangle c c' q$ gives 
\begin{align} 
\lambda^2 R^2 + R^2 - 2 \lambda R^2 \cos(\theta)= {R^\prime}^2,
\label{Form2:lemma-sep_1}
\end{align}
which proves Equation~\eqref{Form:lemma-sep_1}.
If $R \geq R^\prime$, then $\rho\leq 1$, and $\cos(\theta) \geq \lambda / 2$ immediately follows from Equation~\eqref{Form:lemma-sep_1}.
\end{proof}

\begin{figure}[!htb]
\centering
\includegraphics[width=0.6\linewidth]{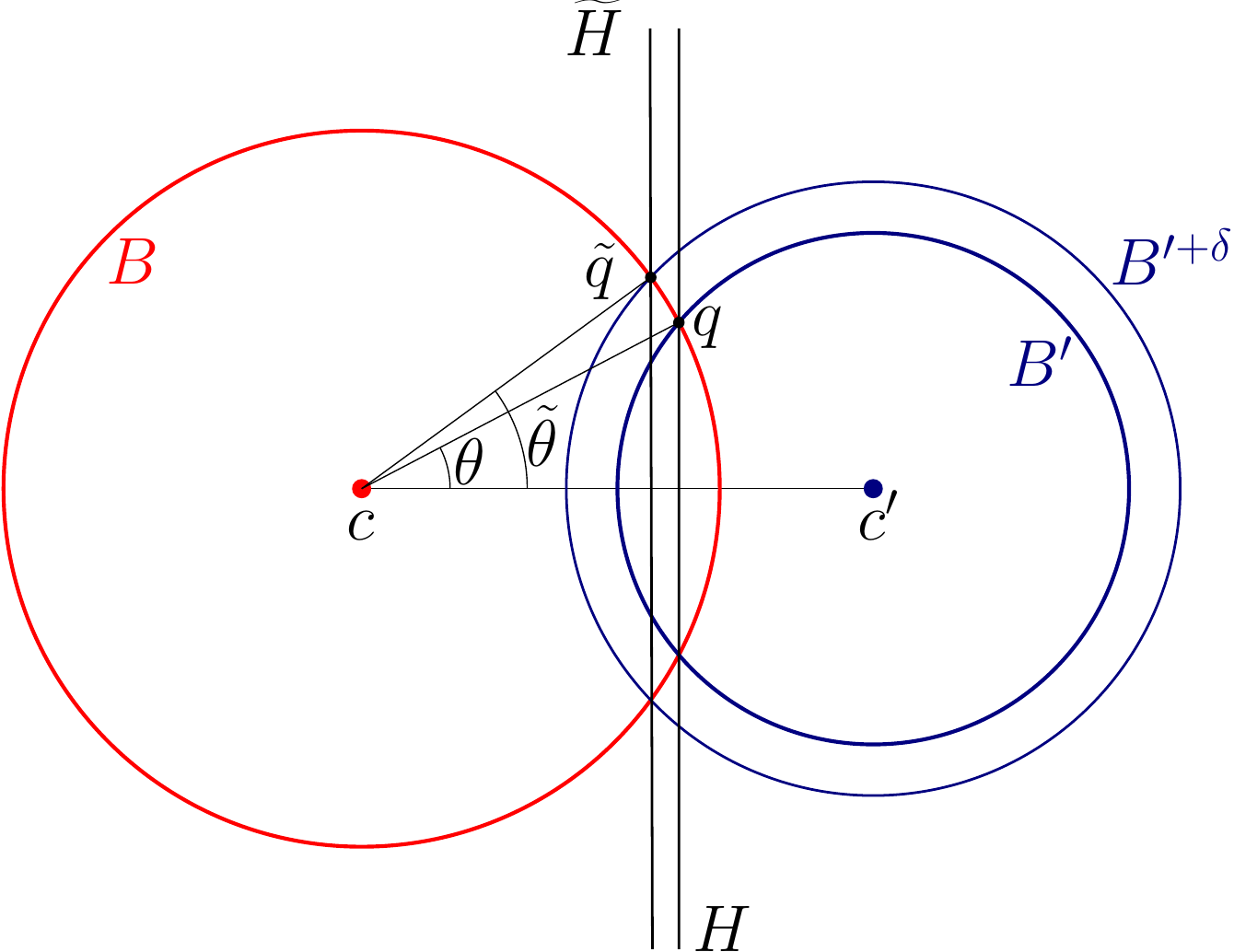}
\caption{Construction used in Lemmas~\ref{lemma-sep_1} and~\ref{lemma-sep_2}.}
\label{fig-separation}
\end{figure}

\begin{lemma} \label{lemma-sep_2}
Let $B = B(c,R)$ and $B^\prime = B(c^\prime, R^\prime)$ be two $n$-balls whose bounding spheres $\partial B$ and $\partial B^\prime$ intersect, and let $\tilde{\theta}$ be the angle of the cone $(c, \widetilde{S})$ where $\widetilde{S}= \partial B \cap \partial B^{\prime+\delta}$.
Writing $\norm{c-c^\prime} = \lambda R$, we have
\begin{align}
 \cos(\tilde{\theta}) = \cos(\theta) - \frac{\delta^2}{2 R^2 \lambda} 
\nonumber
\end{align}
\end{lemma}
\begin{proof}
Let $\tilde{q} \in \tilde{S}$, applying the cosine rule to the triangle $\triangle c c' \tilde{q}$ gives 
\begin{align} 
\lambda^2 R^2 + R^2 - 2 \lambda R^2 \cos(\tilde{\theta}) = {R^\prime}^2 + \delta^2.
\nonumber
\end{align}
Subtracting \eqref{Form2:lemma-sep_1} from the previous equality yields $\delta^2 = 2\lambda R^2(\cos(\theta) - \cos(\tilde{\theta}))$, which proves the lemma.
\end{proof}

The altitude of the vertex $p_i$ in the simplex $\sigma$ is denoted by $D(p_i, \sigma)$.

\begin{lemma} \label{lemma-separation}
Let $\sigma = p\ast\tau$ and $\sigma^\prime = p^\prime\ast\tau$ be two Delaunay simplices sharing a common facet $\tau$.
(Here~$\ast$ denotes the join operator.)
Let $B(\sigma) = B(c,R)$ and $B(\sigma^\prime) = B(c^\prime, R^\prime)$ be the circumscribing balls of~$\sigma$ and~$\sigma^\prime$ respectively.
Then $\sigma^\prime$ is $\delta$-power protected with respect to $p$, that is $p\not\in B(\sigma^\prime)^{+\delta}$ if and only if $\norm{c-c^\prime} \geq \frac{\delta^2}{2D(p,\sigma)}$.
\end{lemma}
\begin{proof}
The spheres $\partial B$ and~$\partial B^{\prime+\delta}$ intersect in a $(n-2)$-sphere $\widetilde{S}$ which is contained in a hyperplane $\widetilde{H}$ parallel to the hyperplane $H = \aff(\tau)$.
For any $\tilde{q}\in\widetilde{S}$ we have
\[ d(\widetilde{H},H)=  d(\tilde{q},H) = R(\cos(\theta) - \cos(\tilde{\theta})) = \frac{\delta^2}{2\norm{c-c^\prime}}, \]
where the last equality follows from Lemma~\ref{lemma-sep_2} and $d(\widetilde{H},H)$ denotes the distance between the two parallel hyperplanes. See Figure~\ref{fig-separation} for a sketch. 
Since $p\in \partial B$, $p$ belongs to~$B(\sigma^\prime)^{+\delta}$ if and only if~$p$ lies in the strip bounded by~$H$ and~$\widetilde{H}$, which is equivalent to
\[ d(p, H) = D(p, \sigma) < \frac{\delta^2}{2\norm{c-c^\prime}}.\]
The result now follows.
\end{proof}

We can make this bound independent of the simplices considered, as shown in Lemma~\ref{lemma-global_separation_bound}.
\begin{lemma} \label{lemma-global_separation_bound}
Let $\mathcal{P}$ be a $\delta$-power protected $(\varepsilon, \mu)$-sample.
The protection parameter $\iota$ is given by $\delta = \iota \varepsilon$.
For any two adjacent Voronoi vertices $c$ and $c'$ of $\VD(\mathcal{P})$, we have
\[ \norm{c-c^\prime} \geq \frac{\delta^2}{4\varepsilon} = \frac{\iota^2\varepsilon}{4}.\]
\end{lemma}
\begin{proof}
For any simplex $\sigma$, we have $D(p,\sigma) \leq 2R_\sigma$ for all $p\in\sigma$, where $R_\sigma$ denotes the radius of the circumsphere of $\sigma$. 
For any $\sigma$ in the triangulation of an $\varepsilon$-net, we have $R_\sigma \leq \varepsilon$.
Thus $D(p,\sigma) \leq 2\varepsilon$, and Lemma~\ref{lemma-separation} yields $\norm{c-c^\prime} \geq \delta^2 / 4\varepsilon$.
\end{proof}

\begin{remark}
In this section, we have computed a lower bound on the distance between two (adjacent) Voronoi vertices.
In Appendix~\ref{appendix-stability}, we shall show that Voronoi vertices are stable under metric perturbations, meaning that when a metric field is slightly modified, the position of a Voronoi vertex does not move too much.
The combination of this separation and stability will then be the basis of many proofs in this paper.
\end{remark}

\subsection{Separation of Voronoi faces (Proof of Lemma~\ref{lemma-Protection_Of_Circumcenters})} \label{section-vor_faces_separation}
Similar results can be obtained on the distance between a Voronoi vertex $c$ and faces that are not incident to $c$, also referred to as \emph{foreign} faces.
Note that we are still in the context of an Euclidean metric.

The following lemma can be found in~\cite[Lemma 3.3]{boissonnat2011cgl} for ordinary protection instead of power protection.
\begin{lemma}
Suppose that $c$ is the circumcenter of a $\delta$-power protected simplex $\sigma$ of a Delaunay triangulation built from an $\varepsilon$-sample, then all foreign Voronoi faces are at least $\delta^2 / 8 \varepsilon$ removed from $c$.
\end{lemma}
\begin{proof}
We denote by $p_0$ an (arbitrary) vertex of $\sigma$, and by $q$ a vertex that is not in $\sigma$ but is adjacent to $p_0$ in $\del(\mathcal{P})$.
Let $x$ be a point in $B(c,r)~\cap~\V_{\E}(p_0)$, with $0 < r < \delta^2 / 4\varepsilon$.
The upper bound for $r$ is chosen with Lemma~\ref{lemma-separation} in mind: we are trying to find a condition such that $x\in\V_{\E}(p_0)$.
The notations are illustrated in Figure~\ref{fig-Voronoi_separation}.

\begin{figure}[!htb]
\centering
\includegraphics[width=.8\linewidth]{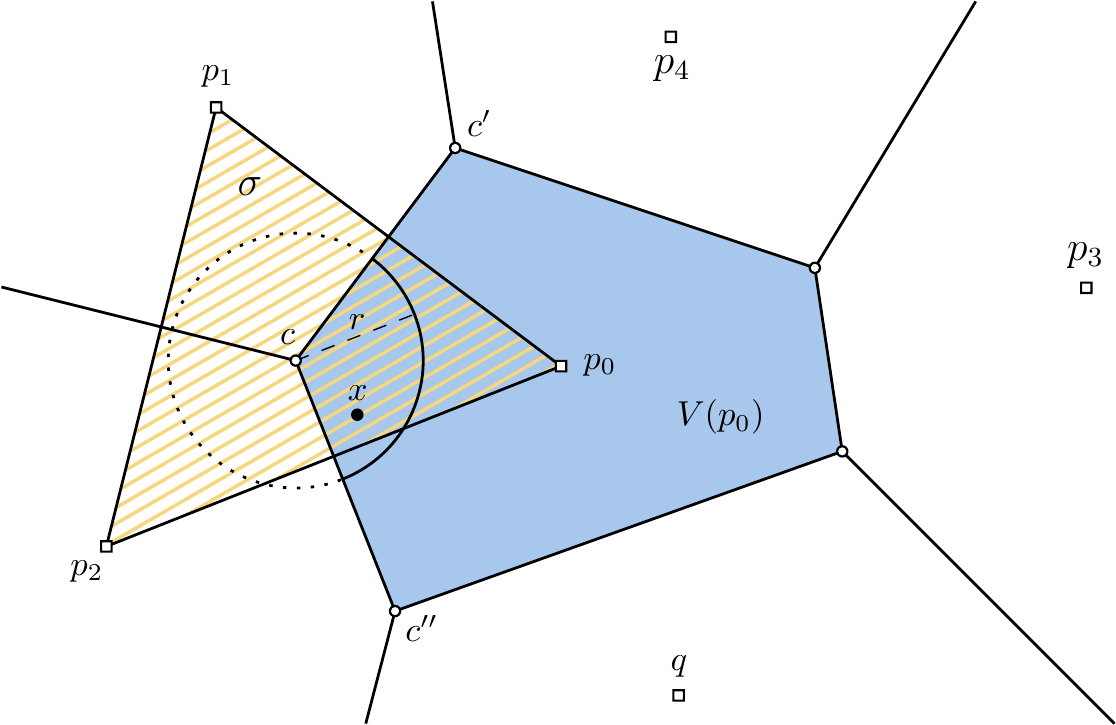}
\caption{Illustration of the notations for the proof of Lemma~\ref{lemma-Protection_Of_Circumcenters}.
         The simplex $\sigma$ is dashed in yellow and has vertices $p_0 p_1 p_2$.
         The distances $cc'$ and $cc''$ are lower bounded by $\delta^2 / 4\varepsilon$.}
\label{fig-Voronoi_separation}
\end{figure}

Because of the triangle inequality, we have that
\begin{align*}
|d(c,q )- d(x,q)|\leq d(x,c) \\
|d(c,p_i )- d(x,p_i)|\leq d(x,c).
\end{align*}
By power protection, we have that $d(c,q)^2\geq d(c,p_i)^2+ \delta^2$.
Therefore,
\begin{align*}
(d(x,q) +r)^2 &\geq (d(x,p_i) -r)^2 + \delta^2 \\
d(x,q)^2 + 2 r d(x,q) &\geq d(x,p_i)^2 - 2 r d(x,p_i) + \delta^2 \\
d(x,q)^2 &\geq d(x,p_i)^2 - 2 r (d(x,p_i) + d(x,q) ) + \delta^2.
\end{align*}

Without loss of generality, we can assume that $q$ is the site closest to $c$ and thus $d(x, q) < d(x, c)$.
If $\mathcal{P}$ is an $\varepsilon$-net, we have 
\[ d(x,p_i) + d(x,q) \leq d(x,p_i) + d(c, q) \leq \varepsilon + 3 \varepsilon = 4\varepsilon, \]
so
\begin{align*}
d(x,q)^2 &\geq d(x,p_i)^2 - 8 r \varepsilon +\delta^2.
\end{align*}
This implies that as long $r < \delta^2 / 8 \varepsilon$, $x$ lies in a Voronoi object associated to the vertices $p_i$ of $\sigma$. 
\end{proof}

Further progressing, we can show that Voronoi faces are thick, with Lemma~\ref{lemma-Voronoi_face_thickness}.
This property is useful to construct a triangulation that satisfies the hypotheses of Sperner's lemma.

\begin{lemma} \label{lemma-Voronoi_face_thickness}
Let $\mathcal{P}$ be a $\delta$-power protected $(\varepsilon, \mu)$-net.
Let $\V_0$ be the Voronoi cell of the site $p_0\in\mathcal{P}$ in the Euclidean Voronoi diagram $\VD_{\E}(\mathcal{P})$.
Then for any $k$-face $F_0$ of $\V_0$ with $k\in [1, \dots, n]$, there exists $x\in F_0$ such that
\[ d(x, \partial F_0) > \frac{\delta^2}{16 \varepsilon}, \]
where $\partial F_0$ denotes the boundary of the face $F_0$.
\end{lemma}
\begin{proof}
All the vertices of $F_0$ are circumcenters of $\VD_{\E}(\mathcal{P})$.
Consider the erosion of the face $F_0$ by a ball of radius $\frac{\delta^2}{16 \varepsilon}$ and denote it $F_0^{-}$.
If $F_0^{-}$ is empty, we contradict the separation Lemma~\ref{lemma-Protection_Of_Circumcenters}.
\end{proof}

\section{Bounds on dihedral angles} \label{appendix-dihedral_angles}
The use of nets allows us to deduce bounds on the dihedral angles of faces of the Delaunay triangulation, as well as on the dihedral angles between adjacent faces of a Voronoi diagram.
Those bounds are frequently used throughout this paper, and specifically to prove stability of Voronoi vertices (see Appendix~\ref{appendix-stability}).
Since we are interested in dihedral bounds in the Euclidean setting, the point set is first assumed to be a net with respect to the Euclidean metric field.
We complicate matters slightly with Lemma~\ref{lemma-dihedral_angle_metric} by assuming that the point set is a power protected net with respect to an arbitrary metric field that is not too far from the Euclidean metric field (in terms of distortion), and still manage to expose bounds with respect to the Euclidean distance.

\subsection{Bounds on the dihedral angles of Euclidean Voronoi cells}
Assuming that a point set is an $(\varepsilon, \mu)$-net allows us to deduce lower and upper bounds on the dihedral angles between adjacent Voronoi faces when the metric field is Euclidean.

\begin{lemma} \label{lemma-Voronoi_angle_bounds}
Let $\Omega = \R^n$ and $\mathcal{P}$ be an $(\varepsilon,\mu)$-net with respect to the Euclidean distance on $\Omega$.
Let $p\in\mathcal{P}$ and $\V(p)$ be the Voronoi cell of $p\in\mathcal{P}$.
Let $q$, $r\in\mathcal{P}$ be two sites such that $\V(q)$ and $\V(r)$ are adjacent to $\V(p)$ in the Euclidean Voronoi diagram of $\mathcal{P}$.
Let $\theta$ be the dihedral angle between $\BS(p,q)$ and $\BS(p,r)$.
Then
\[ 2\arcsin\left(\frac{\mu}{2\varepsilon}\right) \leq \theta \leq \pi-\arcsin\left(\frac{\mu}{2\varepsilon}\right). \]
\end{lemma}

\begin{figure}[!htb]
\centering
\includegraphics[width=.85\linewidth]{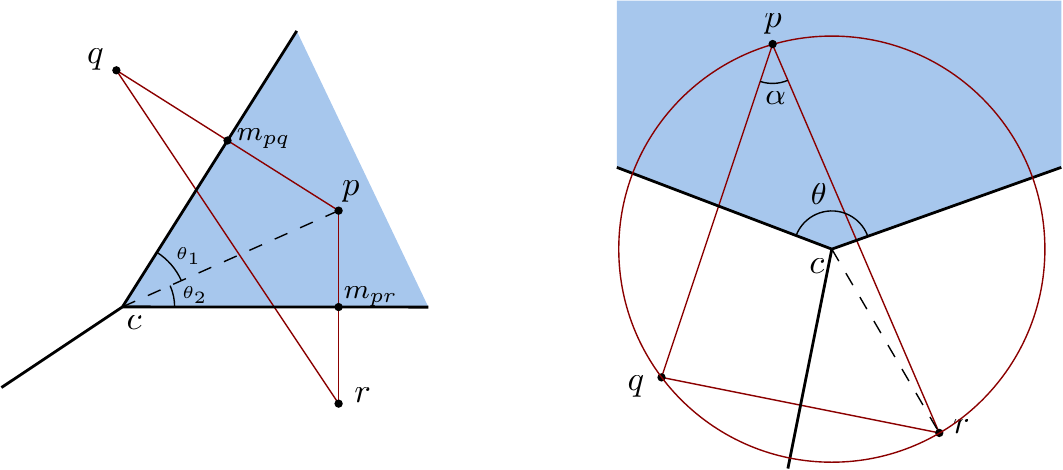}
\caption{Construction and notations used in Lemma~\ref{lemma-Voronoi_angle_bounds}.}
\label{fig-Voronoi_angle_bounds}
\end{figure}

\begin{proof}
We consider the plane $\mathcal{H}$ that passes through the sites $p$, $q$ and $r$.
Notations used below are illustrated in Figure~\ref{fig-Voronoi_angle_bounds}.

\noindent\textbf{Lower bound.}
Let $m_{pq}$ and $m_{pr}$ be the projections of the site $p$ on respectively the bisectors $\BS(p,q)$ and $\BS(p,r)$.
Since $\mathcal{P}$ is an $(\varepsilon,\mu)$-net, we have that $l_q = \abs{p - m_{pq}} \geq \mu / 2$, $l_r = \abs{p - m_{pr}} \geq \mu/2$ and $L = \abs{p - c} \leq \varepsilon$.
Thus
\[ \theta = \arcsin\left(\frac{l_q}{L}\right) + \arcsin\left(\frac{l_r}{L}\right)
          \geq 2\arcsin\left(\frac{\frac{\mu}{2}}{\varepsilon}\right)
          = 2\arcsin\left( \frac{\mu}{2\varepsilon} \right) . \]
Note that since $0 < \mu < 2\varepsilon$, we have $0 < \mu / 2\varepsilon < 1$.

\noindent\textbf{Upper bound.}
To obtain an upper bound on $\theta$, we compute a lower bound on the angle $\alpha = \widehat{q p r}$ at $p$, noting that $\theta = \pi - \alpha$.
Let $l_{qr} = \abs{q - r}$, and $R = \abs{c-r}$.
By the law of sines, we have
\[ \frac{l_{qr}}{\sin(\alpha)} = 2R. \]
Since $\mathcal{P}$ is an $(\varepsilon,\mu)$-net, we have $l_{qr} \geq \mu$ and $R \leq \varepsilon$.
Finally,
\begin{align*}
\alpha \geq \arcsin\left(\frac{\mu}{2\varepsilon}\right) \Longrightarrow \theta \leq \pi - \arcsin\left(\frac{\mu}{2\varepsilon}\right).
\end{align*}
\end{proof}

\subsection{Bounds on the dihedral angles of Euclidean Delaunay simplices}
Bounds on the dihedral angles of simplices guarantee the thickness -- the smallest height of any vertex -- of simplices, and thus their quality.
Additionally, they can be used to show that circumcenters of adjacent simplices are far from one another, thus proving the stability of circumcenters and of Delaunay simplices.

\subsubsection{Using power protection with respect to the Euclidean metric field}
We first assume that the metric field $g$ is the Euclidean metric field $g_{\E}$ and show that the simplices of an Euclidean Delaunay triangulation constructed from a power protected net are thick.

Recall that the dihedral angle can be expressed through heights as
\[ \sin\angle(\aff(\sigma_p), \aff(\sigma_q)) = \frac{D(p,\sigma)}{D(p,\sigma_q)} = \frac{D(q,\sigma)}{D(q,\sigma_p)}. \]
The bound on dihedral angles is obtained by bounding the height of vertices in a simplex.
An obvious upper bound on the height of a vertex $p$ in $\sigma$ is $D(p,\sigma) < 2\varepsilon$.
A lower bound is already obtained in Lemma~\ref{lemma-separation}: we have that $D(p,\sigma) \geq \delta^2 / 2 \norm{c-c'} = \delta^2 / 4\varepsilon$.
We can thus bound the dihedral angles as follows:
\begin{lemma} \label{lemma-dihedral_angle_phi}
Let $\mathcal{P}$ be a $\delta$-power protected $(\varepsilon,\mu)$-net with respect to the Euclidean metric field~$g_0$.
Let $\varphi$ be the dihedral angle between two facets~$\tau_1$ and~$\tau_2$ of a simplex~$\sigma$ of~$\del_{g_0}(\mathcal{P})$.
Then
\[ \arcsin(s_0) \leq \varphi \leq \pi - \arcsin(s_0), \]
with $s_0 = \frac{A_{\lambda, \iota}}{2}$ and $A_{\lambda, \iota}$ defined as in the previous Lemma.
\end{lemma}
\begin{proof}
Recall that
\[ \sin(\varphi) = \sin\angle(\aff(\sigma_p), \aff(\sigma_q)) = \frac{D(p,\sigma)}{D(p,\sigma_q)} = \frac{D(q,\sigma)}{D(q,\sigma_p)}. \]
From previous remarks, we have that
\[ D(q,\sigma_p) \geq D(q, \sigma) \geq \frac{\delta^2}{4\varepsilon}, \]
and $D(q, \sigma) \leq 2\varepsilon$.
Thus, if $\varphi = \angle(\aff(\sigma_p), \aff(\sigma_q))$, then
\begin{align*}
\sin\angle(\aff(\sigma_p), \aff(\sigma_q)) &\geq \frac{\delta^2}{4\varepsilon}\frac{1}{2\varepsilon} = \frac{\iota^2}{2} =: s_0
\end{align*}
Note that $0 < s_0 < 1$ and thus
\[ \arcsin(s_0) \leq \varphi \leq \pi - \arcsin(s_0). \]
\end{proof}

\subsubsection{Using power protection with respect to an arbitrary metric field}
When considering a Voronoi diagram built using the geodesic distance induced by an arbitrary metric field $g$, the assumption of a power protected net is made with respect to this geodesic distance.
To prove the stability of the power protected assumption under metric perturbation, we will however need to deduce lower and upper bounds on the dihedral angles between faces of the simplices of the Riemannian Delaunay complex with respect to the Euclidean metric field .
We prove here that if the point set $\mathcal{P}$ is a $\delta$-power protected $(\varepsilon,\mu)$-net with respect to an arbitrary metric field $g$ and if the distortion between $g$ and the Euclidean metric field $g_{\E}$ is small, then the dihedral angles of the simplices of the Euclidean Delaunay triangulation of $\mathcal{P}$ can be bounded.

We first give a result on the stability of Delaunay balls which expresses that if two metric fields have low distortion, the Delaunay balls of a simplex with respect to each metric field are close.
One of these metric fields is assumed to be the Euclidean metric field.
A similar result can be found in the proof of Lemma $5$ in the theoretical analysis of locally uniform anisotropic meshes of Boissonnat et al.~\cite{boissonnat2015anisotropic}.
\begin{lemma} \label{lemma-euclidean_enclosing_balls}
Let $U \subset \Omega$ be open, and $g$ and $g'$ be two Riemannian metric fields on $U$.
Let $\psi_0\geq 1$ be a bound on the metric distortion.
Suppose that $U$ is included in a ball $B_g(p_0, r_0)$, with $p_0 \in U$ and $r_0 \in \R^{+}$, such that ${\forall p\in B(p_0,r_0)}, {\psi(g(p), g'(p)) \leq \psi_0}$.
Assume furthermore that $g'$ is the Euclidean distance (thus $d_{g'} = d_{\E}$).
Let $B = B_g(c, r)$ be the geodesic ball with respect to the metric field $g$, centered on~$c\in U$ and of radius $r$.
Assume that $B_{\E}(c, \psi_0 r) \subset U$.
Then $B$ can be encompassed by two Euclidean balls $B_{\E}(c, r_{-\psi_0})$ and $B_{\E}(c, r_{+\psi_0})$ with $r_{-\psi_0} = r/\psi_0$ and $r_{+\psi_0} = \psi_0 r$.
\end{lemma}
\begin{proof}
This is a straight consequence from Lemma~\ref{lemma-geodesic_distortion}.
Indeed, we have for all $x\in U$ that
\[ \frac{1}{\psi_0} d_{\E}(c, x) \leq d_g(c, x) \leq \psi_0 d_{\E} (c, x), \]
and similarly
\[ \frac{1}{\psi_0} d_g(c, x) \leq d_{\E}(c, x) \leq \psi_0 d_g (c, x). \]
Thus,
\begin{align*}
x\in B_{\E}\left(c, \frac{r}{\psi_0}\right) \iff d_{\E}(c, x) \leq \frac{r}{\psi_0} \Longrightarrow \frac{1}{\psi_0} d_g(c, x) \leq \frac{r}{\psi_0},
\end{align*}
giving us $B_{\E}(c, r_{-\psi_0}) \subset B$.
On the other hand, we have
\begin{align*}
x\in B_g(c, r) \iff d_g(c, x) \leq r \Longrightarrow \frac{1}{\psi_0} d_{\E}(c, x) \leq r,
\end{align*}
giving us $B \subset B_{\E}(c, r_{+\psi_0})$.
\end{proof}
Note that $r_{-\psi_0}$ and $r_{+\psi_0}$ go to $r$ as $\psi_0$ goes to $1$.

We now use this stability result to provide Euclidean dihedral angle bounds assuming power protection with an arbitrary metric field that is close to $g_\E$.
We first require the intermediary result given by Lemma~\ref{lemma-dih_angle_intermediary}.

\begin{lemma}[Whitney's lemma] \label{lemma-dih_angle_intermediary}
Let $H$ be a hyperplane in Euclidean $n$-space and $\tau$ an $n-1$-simplex whose vertices lie at most $\eta$ from the $H$ and whose minimum height is $h_{\textrm{min}}$.
Then the angle $\xi$ between $\aff(\tau)$ and $H$ is bounded from above by
\[ \sin(\xi) \leq \frac{\eta d}{h_{\textrm{min}}}. \]
\end{lemma}
\begin{proof}
By definition, the barycenter of a $(n-1)$-simplex has barycentric coordinates $\lambda_i=1/d$.
This means that it has distance a $h_{\textrm{min}}/n$ to each of its faces.
So the ball centered on the barycenter with radius $h_{\textrm{min}}/n$ is contained in the simplex.
This means that for any direction in $\aff(\tau)$ there exists a line segment of length $2 h_{\textrm{min}}/n$ that lies within $\tau$.
Moreover the end points of this line segments lie at most $\eta$ from $H$ because the vertices of the $\tau$ do. This means that the angle $\xi$ is bounded by
\begin{align}
\sin(\xi) \leq \frac{2 \eta}{2 h_{\textrm{min}}/d} = \frac{\eta d}{h_{\textrm{min}}}.
\nonumber
\end{align}
\end{proof}

We can now give the main result which bounds \emph{Euclidean} dihedral angles, assuming power protection with respect to the arbitrary metric field.

\begin{lemma} \label{lemma-dihedral_angle_metric}
Let $U \subset \Omega$ be open, and $g$ and $g'$ be two Riemannian metric fields on $U$.
Let $\psi_0\geq 1$ be a bound on the metric distortion.
Suppose that $U$ is included in a ball $B_g(p_0, r_0)$, with $p_0 \in U$ and $r_0 \in \R^{+}$, such that ${\forall p\in B(p_0,r_0)}, {\psi(g(p), g'(p)) \leq \psi_0}$.
Let $\mathcal{P}_U$ be a $\delta$-power protected $(\varepsilon,\mu)$-net over $U$, with respect to $g$.
Let $B = B_g(c, r)$ and $B' = B_g(c', r')$ be the geodesic Delaunay balls of $\sigma = \tau \ast p$ and $\sigma' = \tau \ast p'$, with $\sigma,\sigma' \in Del_g(\mathcal{P}_U)$.
Assume that $\mathcal{P}_U$ is sufficiently dense such that $U$ contains $B$ and $B'$.
Then the minimum height of the simplex satisfies
\begin{align*}
h_{\textrm{min}} =&
\sqrt{1-\left(n \, \frac{(r^2 + r'^2)\left(\frac{1}{\psi_0^2} - \psi_0^2\right) }{8 \varepsilon h_{\textrm{min}}} \right)^2} \cdot \\
&\left(\frac{(r^2 + r'^2)\left(\frac{1}{\psi_0^2} - \psi_0^2\right) + \frac{\delta^2}{\psi_0^2}}{4 \varepsilon} - \frac{n \, \frac{(r^2 + r'^2)\left(\frac{1}{\psi_0^2} - \psi_0^2\right) }{8 \varepsilon h_{\textrm{min} }}}{\sqrt{1-\left(n \, \frac{(r^2 + r'^2)\left(\frac{1}{\psi_0^2} - \psi_0^2\right) }{8 \varepsilon h_{\textrm{min}}} \right)^2}} (r+r')
\right).
\end{align*}
Note that this is the height of $p$ in $\sigma$ with respect to the Euclidean metric.
\end{lemma}
We proceed in a similar fashion as the proof Lemma~\ref{lemma-sep_1}.
However, a significant difference is that we are interested here in only proving that power protection with respect to the generic metric field provides a height bound in the Euclidean setting (rather than an equivalence).

\begin{proof}
We use the following notations, illustrated in Figure~\ref{fig-separation_2}:
\begin{itemize}
\item $B^{\pm\psi_0}_{\E}$ and $B'^{\pm\psi_0}_{\E}$ are the two sets of (Euclidean) enclosing balls of respectively $B$ and $B'$ defined as in Lemma~\ref{lemma-euclidean_enclosing_balls}.
\item $B'^{+\delta}$ is the power protected ball of $\sigma'$, given by $B'^{+\delta} = B(c', \sqrt{r'^2 + \delta^2})$.
\item $B'^{+\delta, \pm\psi_0}_{\E}$ are the two Euclidean balls enclosing $B'^{+\delta}$.
\item $S = \partial B^{-\psi_0} \cap \partial B'^{+\psi_0}$, $\tilde{S} = \partial B^{+\psi_0} \cap \partial B'^{+\delta, -\psi_0}$ and $S' = \partial B^{+\psi_0} \cap \partial B'^{-\psi_0}$.
\item $q$ is a point on $S$, $\tilde{q}$ is a point on $\widetilde{S}$ and $q'$ is a point on $S'$.
\item $H$ is the geodesic supporting plane of $\tau$, that is $\{ \argmin(\sum_{v_i \in \tau}\lambda_i d_{g}(x, v_i))\}$.
\item $H_{\E}$, $\widetilde{H}_{\E}$ and $H_\E$ are the two Euclidean hyperplanes orthogonal to $[cc']$ passing through respectively $q$, $\tilde{q}$ and $q'$.
\item $\theta = \widehat{c'cq}$ $\widetilde{\theta} = \widehat{c'c\tilde{q}}$, and $\lambda_c = \abs{cc'} / r$.
\end{itemize}

\begin{figure}[!htb]
\centering
\includegraphics[width=0.7\linewidth]{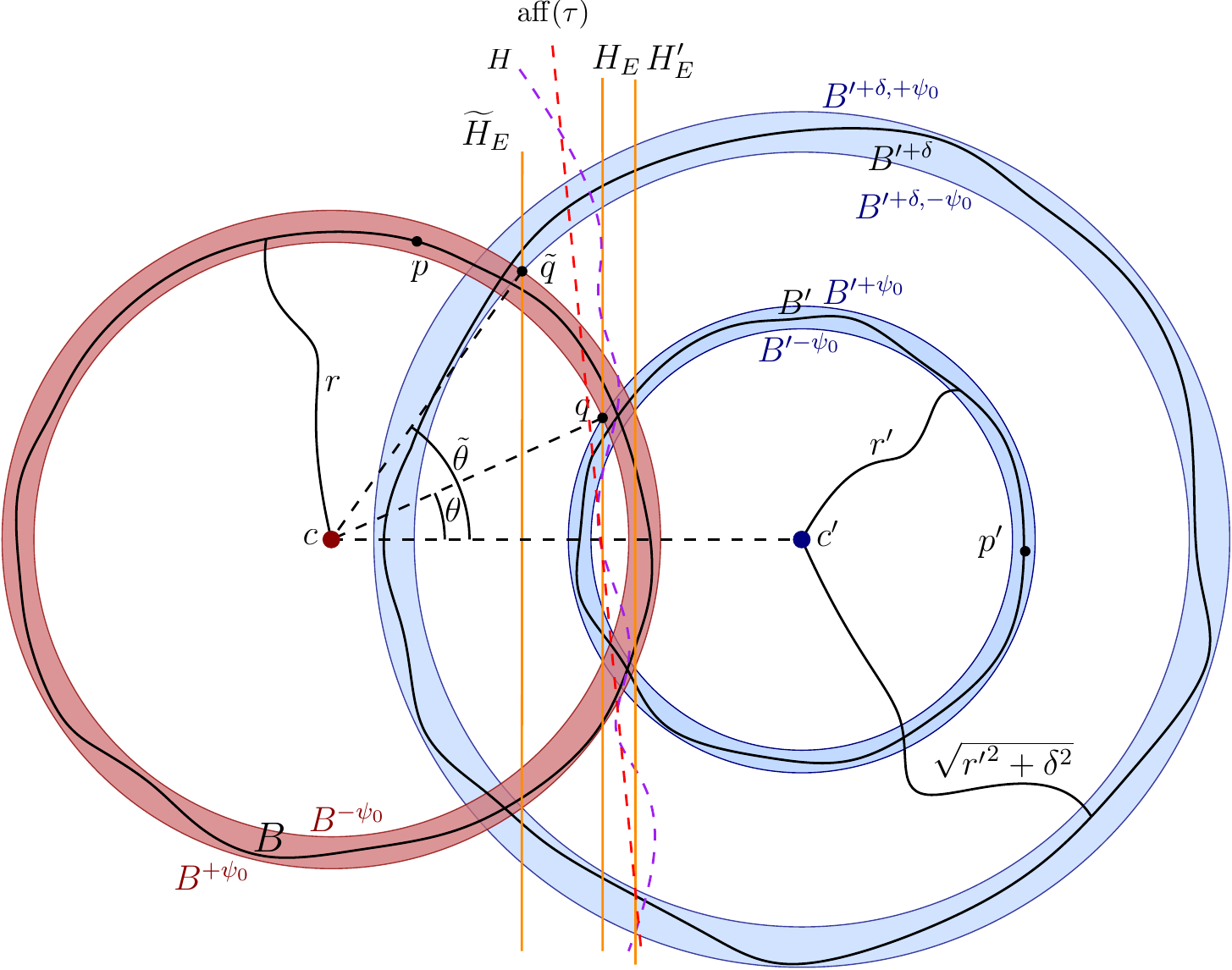}
\caption{Construction and notations used in the proof of Lemma~\ref{lemma-dihedral_angle_metric}} \label{fig-separation_2}
\end{figure}

While the vertices of $\tau$ live on $H$, $\aff_{\E}(\tau)$ is not necessarily orthogonal to~$[cc']$ and consequently
\[ d_{\E}(p, H_{\E}) \leq d_{\E}(p, \tau). \]
The separation between the hyperplanes $H_{\E}$ and $\widetilde{H}_{\E}$ provides a lower bound on the distance $d_{\E}(p, \tau)$, thus on the height $D_{\E}(p,\sigma)$.
We therefore seek to bound $d_{\E}(H_{\E}, \widetilde{H}_{\E})$.

By definition of the enclosing Euclidean balls, we have that
\[ \abs{c q} = \frac{r}{\psi_0},\quad \abs{c\tilde{q}} = \psi_0 r, \quad \abs{c' q} = \psi_0 r', \text{ and } \abs{c'\tilde{q}} = \frac{1}{\psi_0} \sqrt{r'^2 + \delta^2} \]
Using the law of cosines in the triangles $\triangle c c' q$ and $\triangle cc'\tilde{q}$, we find
\begin{align*}
\frac{r'^2 + \delta^2}{\psi_0^2} &= \lambda_c^2 r^2 + \psi_0^2 r^2 - 2\lambda_c \psi_0 r^2 \cos(\widetilde{\theta}) \\
                   \psi_0^2 r'^2 &= \lambda_c^2 r^2 + \frac{r^2}{\psi_0^2} - 2\frac{\lambda_c r^2}{\psi_0} \cos(\theta),
\end{align*}
where $\lambda_c r = \abs{c-c'}$.
Subtracting one from the other, we obtain
\begin{align*}
\psi_0^2 r'^2 - \frac{r'^2 + \delta^2}{\psi_0^2} &= \frac{r^2}{\psi_0^2} - \psi_0^2 r^2 + 2\lambda_c \psi_0 r^2 \cos(\widetilde{\theta}) - 2\frac{\lambda_c r^2}{\psi_0} \cos(\theta) \\
\iff \frac{r}{\psi_0} \cos(\theta) - \psi_0 r \cos(\widetilde{\theta}) &= \frac{\frac{r^2}{\psi_0^2} - \psi_0^2 r'^2 + \frac{r'^2 + \delta^2}{\psi_0^2} - \psi_0^2 r^2}{2\lambda_c r} \\
%\iff \frac{r}{\psi_0} \cos(\theta) - \psi_0 r \cos(\widetilde{\theta}) &= \frac{\frac{1}{\psi_0^2}(r^2 + r'^2 + \delta^2) - \psi_0^2 (r^2 + r'^2)}{2\lambda_c r} \\
\iff \frac{r}{\psi_0} \cos(\theta) - \psi_0 r \cos(\widetilde{\theta}) &= \frac{(r^2 + r'^2)\left(\frac{1}{\psi_0^2} - \psi_0^2\right) + \frac{\delta^2}{\psi_0^2}}{2\lambda_c r}, \\
\end{align*}
so that 
\[ d(H_{\E}, \widetilde{H}_{\E}) = \frac{r}{\psi_0} \cos(\theta) - \psi_0 r \cos(\widetilde{\theta})= \frac{(r^2 + r'^2)\left(\frac{1}{\psi_0^2} - \psi_0^2\right) + \frac{\delta^2}{\psi_0^2}}{2\lambda_c r}. \]
Similarly we can calculate the distance $d_\E (H_\E, H' _\E) $ to be:
\begin{align}
d_\E (H_\E, H' _\E) &=
\frac{\frac{1}{ \psi^2} (r')^2+  |c -c'| ^2- \psi^2 r^2  }{2 |c-c'| } - \frac{\psi^2 (r') ^2+  |c -c'| ^2- \frac{1}{\psi^2} r^2  }{2 |c-c'| }
\nonumber
\\
&=\frac{(r^2 + r'^2)\left(\frac{1}{\psi_0^2} - \psi_0^2\right) }{2 |c-c'| }.
\nonumber
\\
& \leq  \frac{(r^2 + r'^2)\left(\frac{1}{\psi_0^2} - \psi_0^2\right) }{4 \varepsilon }
\nonumber
\end{align}
Lemma~\ref{lemma-dih_angle_intermediary} gives us that the angle $\xi$ between $H_\E$ and $\aff(\tau)$ is bounded by  
\[ \sin(\xi) \leq \frac{n \, d_\E (H_\E, H' _\E) }{2 h_{\textrm{min}}}.\]
The vertices of $\tau$ lie in between $H_\E$ and $H'_\E$ and inside $B_\E(c', \psi r')$.
Thus, if we restrict to the $\tilde{q} c c' $ plane, distance between the point where $\aff(\tau)$ intersects $H_\E$ and the orthogonal projection $\pi_{H_\E}(\tilde{q})$ of $\tilde{q}$ on $H_\E$ is at most $\psi (r+r') $.
This in turn implies that the line connecting $\pi_{H_\E}(\tilde{q})$ and $\tilde{q}$ intersects $\aff(\tau)$ at most distance $(r+r') \tan(\xi) $ from $H_\E$.
This also gives us that the distance from $\tilde{q}$ to its orthogonal projection $\pi_{\aff(\tau)} (\tilde{q})$ on $\aff(\tau)$ is bounded by
\begin{align*}
\cos(\xi) (d(\tilde{H}_\E,H_\E) - \tan(\xi) (r+r')) =& \sqrt{1-\left(\frac{n \, d_\E (H_\E, H' _\E) }{2 h_{\textrm{min}}}\right)^2} \cdot \\
&\left(d(\tilde{H}_\E,H_\E) - \frac{\frac{n \, d_\E (H_\E, H' _\E) }{2 h_{\textrm{min}}}}{\sqrt{1-\left(\frac{n \, d_\E (H_\E, H' _\E) }{2 h_{\textrm{min}}}\right)^2}} (r+r') \right).
\end{align*}
Here we note that $h_{\textrm{min}}$ originally referred to the minimum height of a face, but because the height of a face is always greater than the height in the simplex we may read this in a general way, that is we regard $h_{\textrm{min}}$ as a universal lower bound on the height.
Because $|\tilde{q}-\pi_{\aff(\tau)} (\tilde{q})|$ bounds the height of the simplex we get the following relation:
\begin{align*}
h_{\textrm{min}} &= 
\sqrt{1-\left(\frac{n \, d_\E (H_\E, H' _\E) }{2 h_{\textrm{min}}}\right)^2} 
\left(d(\tilde{H}_\E,H_\E) - \frac{\frac{n \, d_\E (H_\E, H' _\E) }{2 h_{\textrm{min}}}}{\sqrt{1-\left(\frac{n \, d_\E (H_\E, H' _\E) }{2 h_{\textrm{min}}}\right)^2}} (r+r') \right) \\
h_{\textrm{min}} &=
\sqrt{1-\left(n \, \frac{(r^2 + r'^2)\left(\frac{1}{\psi_0^2} - \psi_0^2\right) }{8 \varepsilon h_{\textrm{min}}} \right)^2} \cdot \\
&\left(\frac{(r^2 + r'^2)\left(\frac{1}{\psi_0^2} - \psi_0^2\right) + \frac{\delta^2}{\psi_0^2}}{4 \varepsilon} - \frac{n \, \frac{(r^2 + r'^2)\left(\frac{1}{\psi_0^2} - \psi_0^2\right) }{8 \varepsilon h_{\textrm{min} }}}{\sqrt{1-\left(n \, \frac{(r^2 + r'^2)\left(\frac{1}{\psi_0^2} - \psi_0^2\right) }{8 \varepsilon h_{\textrm{min}}} \right)^2}} (r+r')
\right)
\end{align*}
To make the expression a bit more readable, we introduce
\[ s_1 = \frac{(r^2 + r'^2)\left(\frac{1}{\psi_0^2} - \psi_0^2\right) }{4 \varepsilon } \]
so that
\begin{align*}
h_{\textrm{min}} =
\sqrt{1-\left(n \, \frac{s_1 }{2 h_{\textrm{min}}} \right)^2} \left(\left(s_1+\frac{\delta^2}{4 \varepsilon\psi_0^2} \right) - \frac{n \, \frac{s_1 }{2 h_{\textrm{min}}}}{\sqrt{1-\left(n \, \frac{s_1 }{2 h_{\textrm{min}}} \right)^2}} (r+r') \right) \\
h_{\textrm{min}} + n \, \frac{s_1 }{2 h_{\textrm{min}}} (r+r') = \sqrt{1-\left(n \, \frac{s_1 }{2 h_{\textrm{min}}} \right)^2} \left(s_1+\frac{\delta^2}{4 \varepsilon\psi_0^2}\right).
\end{align*}
Multiplying with $h_\textrm{min}$ and squaring we find:
\[ (h_{\textrm{min}}^2 + n \, \frac{s_1 }{2} (r+r'))^2 = \left (h_{\textrm{min}}^2 -\left(n \, \frac{s_1 }{2 } \right)^2 \right) \left(s_1+\frac{\delta^2}{4 \varepsilon\psi_0^2}\right) ^2 \]
We note that in the limit $\psi \to 1$, $s_1$ tends to zero, we therefore expand $h_{\textrm{min}}$ is terms of $s_1$. Using a computer algebra system we find 
\begin{align*}
h_{\textrm{min}}^2= \frac{\delta ^4}{2^{10} \psi ^4 \varepsilon ^2}&+\left(\frac{\delta ^2 }{128 \psi ^2 \varepsilon }-n  (r+r')\right) s_1 \\ &-\frac{ \delta ^4 \left(16 n^2-1\right)+2^{14} n^2 \psi ^4 \varepsilon ^2 (r+r')^2}{64 \delta ^4}s_1^2+\mathcal{O}(s_1^3)
\end{align*}
We emphasize that this equation gives $h_{\textrm{min}} = \delta^2/4\varepsilon$ as $\psi$ tends to $1$.
This means that for sufficiently small metric distortion the height of a protected simplex will be strictly positive.
\end{proof}

\begin{lemma} \label{lemma-dihedral_angle_from_delta_pp}
Let $g$ be a metric field and $\mathcal{P}$ be a point set defined over $\R^n$.
Let $\psi_0 = \psi(g, \E)$.
Assume that $\mathcal{P}$ is a $\delta$-power protected $(\varepsilon,\mu)$-net over $\R^n$.
Let $\phi$ be the dihedral angle between two faces of a simplex $\tau \in \del_{\E}(\mathcal{P})$.
Then
\[ \arcsin(s_0) \leq \phi \leq \pi - \arcsin(s_0), \]
with $s_0$ detailed in the proof.
\end{lemma}
\begin{proof}
Denote $h_{\textrm{min}}$ the lower bound on $D(q, \sigma)$ obtained in the previous lemma.
We also immediate have that
\[D(q, \sigma) \leq 2\varepsilon. \]
Let $\varphi$ be the dihedral angle between $\aff(\sigma_p)$ and $\aff(\sigma_q)$.
Recall that
\[ \sin(\varphi) = \sin\angle(\aff(\sigma_p), \aff(\sigma_q)) = \frac{D(p,\sigma)}{D(p,\sigma_q)} = \frac{D(q,\sigma)}{D(q,\sigma_p)} \]
Thus
\begin{align*}
\sin(\varphi) \geq \frac{h_{\textrm{min}}}{2\varepsilon} =: s_0.
\end{align*}

For $s_0$ to make sense, we want $s_0 > 0$, which is bound to happen as $\psi_0$ goes to $1$, as shown in Lemma~\ref{lemma-dihedral_angle_metric}: $h_{\textrm{min}}$ goes to $\delta^2 / 4\varepsilon$ thus $h_{\textrm{min}} / 4\varepsilon$ goes to $\iota^2 / 4 > 0$.
\end{proof}

\section{Stability} \label{appendix-stability}
The notion of stability designates the conservation of a property despite changes of other parameters.
In our context, the main assumptions concern the nature of point sets: we assume that point sets are power protected nets and wish to preserve these hypotheses despite (small) metric perturbations.
The stability is important both from a theoretical and a practical point of view.
Indeed, if an assumption is stable under perturbation, we can simplify matters without losing information.
For example, we will prove that if a point set is a net with respect to a metric field $g$, then it is also a net (albeit with slightly different sampling and separation parameters) for a metric field $g'$ that is close to $g$ (in terms of distortion)
In a practical context, the stability of an assumption provides robustness with respect to numerical errors (see, for example, the work of Funke et al.~\cite{funke2005controlled} on the stability of Delaunay simplices).

\subsection{Stability of the protected net hypothesis under metric transformation}
It is rather immediate that the power protected net property is preserved when the point set is transformed (see Section~\ref{section-metric_transformation}) between these spaces, as shown by the following lemma.

\begin{lemma} \label{lemma-protected_net_metric_transformation}
Let $\mathcal{P}$ be a $\delta$-power protected $(\varepsilon,\mu)$-net in $\Omega$.
Let $g = g_0$ be a uniform Riemannian metric field and $F_0$ a square root of $g_0$.
If $\mathcal{P}$ is a $\delta$-power protected $(\varepsilon,\mu)$-net with respect to $g_0$ then $\mathcal{P}^\prime = \{F_0p_i, p_i\in\mathcal{P}\}$ is a $\delta$-power protected $(\varepsilon,\mu)$-net with respect to the Euclidean metric.
\end{lemma}
\begin{proof}
This results directly from the observation that
\[ d_{0}(x,y)^2 = (x-y)^t g_0 (x-y) = \norm{F_0(x-y)}^2 = d(F_0 x, F_0 y)^2. \]
\end{proof}

\subsection{Stability of the protected net hypothesis under metric perturbation}
Metric field perturbations are small modifications of a metric field in terms of distortion.
Since a generic Riemannian metric field $g$ is difficult to study, we will generally consider a uniform approximation $g_0$ of $g$ within a small neighborhood, such that the distortion between both metric fields is small over that neighborhood.
In that context, the perturbation of $g$ is the act of bringing $g$ onto $g_0$.
Stability of the assumption of power protection was previously investigated by Boissonnat et al.~\cite{boissonnat:hal-01096798} in the context of manifold reconstructions.

\subsubsection{Stability of the net property}
The following lemma shows that the ``net'' property is preserved when the metric field is perturbed: a point set that is a net with $g$ is also a net with respect to $g'$, a metric field that is close to $g$.

\begin{lemma} \label{lemma-net_metric_perturbation}
Let $U \subset \Omega$ be open, and $g$ and $g'$ be two Riemannian metric fields on $U$.
Let $\psi_0\geq 1$ be a bound on the metric distortion.
Suppose that $U$ is included in a ball $B_g(p_0, r_0)$, with $p_0 \in U$ and $r_0 \in \R^{+}$, such that ${\forall p\in B(p_0,r_0)}, {\psi(g(p), g'(p)) \leq \psi_0}$.
Let $\mathcal{P}_U$ be a point set in $U$.
Suppose that $\mathcal{P}_U$ is an $(\varepsilon,\mu)$-net of $U$ with respect to $g$.
Then $\mathcal{P}_U$ is an $(\varepsilon_{g'},\mu_{g'})$-net of $U$ with respect to $g'$ with $\varepsilon_{g'} = \psi_0\varepsilon$ and $\mu_{g'} = \mu/\psi_0$.
\end{lemma}
\begin{proof}
By Lemma~\ref{lemma-geodesic_distortion}, we have that
\[ \forall x, y\in U, \frac{1}{\psi_0}d_{g'}(x,y) \leq d_g(x,y) \leq \psi_0 d_{g'}(x,y). \]
Therefore
\[ \forall x\in U, \exists p\in\mathcal{P}_U, d_g(x,p) \leq\varepsilon \Rightarrow \forall x\in U, \exists p\in\mathcal{P}_U, d_{g'}(x,p) \leq\psi_0\varepsilon,\]
and
\[ \forall p,q\in\mathcal{P}_U, d_g(p,q) \geq\mu \Rightarrow \forall p,q\in\mathcal{P}_U, d_{g'}(p,q) \geq\frac{\mu}{\psi_0}.\]
And, with $\varepsilon_{g'} = \psi_0\varepsilon$ and $\mu_{g'} = {\mu/\psi_0}$, $\mathcal{P}$ is an $(\varepsilon_{g'},\mu_{g'})$-net of $U$.
\end{proof}

\begin{remark} \label{remark-epsilon_mu_prime}
We assumed that $\mu = \lambda\varepsilon$.
By Lemma~\ref{lemma-net_metric_perturbation}, we have $\varepsilon' = \psi_0\varepsilon$ and $\mu' = \frac{\mu}{\psi_0}$.
Therefore 
\[ \mu' = \frac{\lambda\varepsilon}{\psi_0} = \frac{\lambda}{\psi_0^2}\varepsilon'. \]
\end{remark}

\subsubsection{Stability of the power protection property}
It is more complex to show that the assumption of power protection is preserved under metric perturbation.
Previously, we have only considered two similar but arbitrary Riemannian metric fields~$g$ and~$g'$ on a neighborhood~$U$.
We now restrict ourselves to the case where~$g'$ is a uniform metric field.
We shall always compare the metric field~$g$ in a neighborhood~$U$ with the uniform metric field $g' = g_0 = g(p_0)$ where $p_0 \in U$.
Because~$g_0$ and the Euclidean metric field differ only by a linear transformation, we simplify matters and assume that~$g_0$ is the Euclidean metric field.

We now give conditions such that the point set is also protected with respect to~$g_0$.
A few intermediary steps are needed to prove the main result:
\begin{itemize}
\item Given two sites, we prove that the bisectors of these two sites in the Voronoi diagrams built with respect to $g$ and with respect to $g' = g_{\E}$ are close (Lemma~\ref{lemma-relaxed_Voronoi_cell}).
\item We prove that the Voronoi cell of a point $p_0$ with respect to $g$, $\V_g(p_0)$ can be encompassed by two scaled versions (one larger and one smaller) of the Euclidean Voronoi cell $\V_{\E}(p_0)$ (Lemma~\ref{lemma-dilated_eroded_Voronoi_cells}).
\item We combine this encompassing with bounds on the dihedral angles of Euclidean Delaunay simplices given by Lemma~\ref{lemma-dihedral_angle_from_delta_pp} to compute a stability region around Voronoi vertices where the same combinatorial Voronoi vertex of lives for both $\V_g(p_0)$ and $\V_{\E}(p_0)$ in $2$D (Lemma~\ref{lemma-voronoi_vertices_metric_perturbation_2D}).
We then extend it to any dimension by induction (Lemma~\ref{lemma-voronoi_vertices_metric_perturbation}).
\end{itemize}
The main result appears in Lemma~\ref{lemma-protection_metric_perturbation} and gives the stability of power protection under metric perturbation.

We first define the scaled version of a Voronoi cell more rigorously.
\begin{definition}[Relaxed Voronoi cell]
Let $\omega\in\R$.
The relaxed Voronoi cell of the site $p_0$ is
\[ \V_g^{+\omega}(p_0) = \{x \in U  \mid d_g(p_0,x)^2 \leq d_g(p_i, x)^2 + \omega \textrm{ for all } i \neq 0 \}. \]
\end{definition}

The following lemma expresses that two Voronoi cells computed in similar metric fields are close.
\begin{lemma} \label{lemma-relaxed_Voronoi_cell}
Let $U \subset \Omega$ be open, and $g$ and $g'$ be two Riemannian metric fields on $U$.
Let $\psi_0\geq 1$ be a bound on the metric distortion.
Suppose that $U$ is included in a ball $B_g(p_0, r_0)$, with $p_0 \in U$ and $r_0 \in \R^{+}$, such that ${\forall p\in B_g(p_0,r_0)}, {\psi(g(p), g'(p)) \leq \psi_0}$.
Let $\mathcal{P}_U = \{p_i\}$ be a point set in $U$.
Let $\V_{p_0, g}$ denote a Voronoi cell with respect to the Riemannian metric field $g$.

Suppose that the Voronoi cell $\V_{g'}^{2\rho^2(\psi_0^2-1)}(p_0)$ lies in a ball of radius $\rho$ with respect to the metric $g'$, which lies completely in $U$.
Let~$\omega_0 = 2\rho^2(\psi_0^2-1)$.
Then~$\V_{g}(p_0)$ lies in~$\V_{g'}^{+\omega_0}(p_0)$ and contains~$\V_{g'}^{-\omega_0}(p_0)$.
\end{lemma}
\begin{proof}
Let $\BS_{g}(p_0, p_i)$ be the bisector between $p_0$ and $p_i$ with respect to $g$.
Let $y \in \BS_{g}(p_0, p_i) \cap B_{g'}(p_0, \rho)$, where $B_{g'}(p_0, \rho)$ denotes the ball centered at $p_0$ of radius $\rho$ with respect to $g'$.
Now $d_{g}(y,p_0)= d_{g}(y,p_i)$, and thus 
\begin{align*}
|d_{g'}(y,p_0)^2- d_{g'}(y,p_i)^2| & = \abs{d_{g'}(y,p_0)^2 - d_{g}(y,p_0)^2 + d_{g}(y,p_i)^2 - d_{g'}(y,p_i)^2} \\
                                   & \leq  \abs{d_{g'}(y,p_0)^2 - d_{g}(y,p_0)^2} + \abs{d_{g'}(y,p_i)^2 - d_{g}(y,p_i)^2} \\
                                   & \leq (\psi_0^2 - 1) (d_{g'}(y,p_0)^2 + d_{g'}(y,p_i))^2 \\
                                   & \leq 2\rho^2(\psi_0^2 - 1).
\end{align*}
Thus $d_{g'}(y,p_0)^2 \leq d_{g'}(y,p_i)^2 + \omega$ and $d_{g'}(y,p_0)^2 \geq d_{g'}(y,p_i)^2 - \omega$ with $\omega = 2\rho^2(\psi_0^2 - 1)$, which gives us the expected result.
\end{proof}

\begin{remark}
Lemma~\ref{lemma-relaxed_Voronoi_cell} does not require $g'$ to be a uniform metric field.
\end{remark}

We clarify in the next lemma that the bisectors of a Voronoi diagram with respect to a uniform are affine hyperplanes.
\begin{lemma} \label{lemma-hyperplanes_bisectors}
Let $U \subset \Omega$ be open, and $g$ and $g'$ be two Riemannian metric fields on $U$.
Let $\psi_0\geq 1$ be a bound on the metric distortion.
Suppose that $U$ is included in a ball $B_g(p_0, r_0)$, with $p_0 \in U$ and $r_0 \in \R^{+}$, such that ${\forall p\in B(p_0,r_0)}, {\psi(g(p), g'(p)) \leq \psi_0}$.
Let $\mathcal{P}_U = \{p_i\}$ be a point set in $U$.
Let $g'$ be a uniform metric field.
We refer to $g'$ as $g_0$ to emphasis its constancy.
Let $p_0\in\mathcal{P}_U$.
The bisectors of $\V_{g_0}^{\pm\omega_0}(p_0)$ are hyperplanes.
\end{lemma}
\begin{proof}
The bisectors of $\V_{g_0}^{\pm\omega_0}(p_0)$ are given by 
\[ \BS_{g_0}^{\pm\omega_0}(p_0,p_i) = \left\{ x\in\Omega \mid d_{g_0}(p_0,x)^2 = d_{g_0}(p_i, x)^2 \pm \omega_0 \right\}. \]
For $x\in\BS_{g_0}^{\pm\omega_0}(p_0,p_i)$, we have by definition that
\begin{align*}
\norm{x-p_0}_{g_0}^2 &= \norm{x-p_i}_{g_0}^2 \pm \omega_0 \\
\iff (x-p_0)^t g_0 (x-p_0) &= (x-p_i)^t g_0 (x-p_i) \pm \omega_0 \\
\iff\quad\quad 2 x^t g_0 (p_i - p_0) &=  p_i^t g_0 p_i - p_0^t g_0 p_0 \pm\omega_0.
\end{align*}
which is the equation of an hyperplane since~$g_0$ is uniform.
\end{proof}

The cells $\V_{g_0}^{\pm\omega_0}(p_0)$ are unfortunately impractical to manipulate as we do not have an explicit distance between the boundaries $\partial \V_{g_0}(p_0)$ and $\partial \V_{g_0}^{\pm\omega_0}(p_0)$.
However that distance can be bounded; this is the purpose of the following lemma.
\begin{lemma} \label{lemma-relaxed_bisector_bound}
Let $U \subset \Omega$ be open, and $g$ and $g'$ be two Riemannian metric fields on $U$.
Let $\psi_0\geq 1$ be a bound on the metric distortion.
Suppose that $U$ is included in a ball $B_g(p_0, r_0)$, with $p_0 \in U$ and $r_0 \in \R^{+}$, such that ${\forall p\in B(p_0,r_0)}, {\psi(g(p), g'(p)) \leq \psi_0}$.
Let $\mathcal{P}_U = \{p_i\}$ be a point set in $U$.
Assume furthermore that $g_0$ is the Euclidean metric field.
Let $p_0\in\mathcal{P}_U$.
We have
\[ d_{g_0}(\partial \V_{g_0}(p_0), \partial \V_{g_0}^{\pm\omega_0}(p_0)) \, := \min_{\substack{x\in\partial \V_{g_0}(p_0) \\
y\in\partial \V_{g_0}^{\pm\omega_0}(p_0)}} d_{g_0} (x, y) \leq \frac{\rho_0^2(\psi_0^2 - 1)}{\mu_0},\]
with $\rho_0$ defined as $\rho$ is in Lemma~\ref{lemma-relaxed_Voronoi_cell}.
\end{lemma}
\begin{proof}
Let $m_{\omega_0}$ be the intersection of the segment $[p_0,p_i]$ and the bisector $\BS_{g_0}^{-\omega_0}(p_0,p_i)$, for $i\neq 0$.
Let $m$ be the intersection of the segment $[p_0,p_i]$ and $\partial \V_{g_0}(p_0)$, for $i\neq 0$.
We have
\[
  \begin{cases}
    m_{\omega_0}\in\BS_{g_0}^{-\omega_0}(p_0,p_i) &\iff 2m_{\omega_0}^T (p_i-p_0) = p_i^Tp_i - p_0^Tp_0 - \omega_0 \\
    m\in\partial \V_{g_0}(p_0) &\iff 2 m^T (p_i-p_0) = p_i^Tp_i - p_0^Tp_0
  \end{cases}
\]
Therefore
\[ 2(m - m_{\omega_0})^T(p_i - p_0) = \omega_0. \]
Since $(m-m_{\omega_0})$ and $(p_i - p_0)$ are linearly dependent,
\[ \omega_0 = 2\abs{m-m_{\omega_0}}\abs{p_i - p_0}.\]
$\mathcal{P}$ is $\mu_0$-separated, which implies that
\[ \omega_0 \geq 2\abs{m-m_{\omega_0}}\mu_0, \]
and
\[ \abs{m-m_{\omega_0}} \leq \frac{\omega_0}{2\mu_0} = \frac{\rho_0^2(\psi_0^2 - 1)}{\mu_0}. \]
\end{proof}

\begin{definition} \label{definition-rho_eta}
In the following, we use $\rho = 2\varepsilon_0$, and therefore $\omega_0 = 8\varepsilon_0^2(\psi_0^2-1)$.
We show that this choice is reasonable in Lemma~\ref{lemma-choice_of_rho}.
Additionally, we define
\[\eta_0 = \frac{\rho_0^2(\psi_0^2 - 1)}{\mu_0}.\]
\end{definition}

We are now ready to encompass the Riemannian Voronoi cell of $p_0$ with respect to an arbitrary metric field $g$ with two scaled versions of the Euclidean Voronoi cell of $p_0$.
The notions of \emph{dilated} and \emph{eroded} Voronoi cells will serve the purpose of defining precisely these scaled cells.
\begin{definition}[Eroded Voronoi cell]
Let $\omega\in\R$.
The eroded Voronoi cell of $p_0$ is the morphological erosion of $\V_{g}(p_0)$ by a ball of radius $\omega$:
\[ \EV_{g}^{-\omega}(p_0) = \{ x\in \V_{g}(p_0) \mid d_{g}(x, \partial \V_{g}(p_0)) > \omega \}. \]
\end{definition}

\begin{definition}[Dilated Voronoi cell]
Let $\omega\in\R$.
The dilated Voronoi cell of $p_0$ is:
\[ \DV_{g}^{+\omega}(p_0) = \bigcap\limits_{i\neq 0} H^\omega(p_0, p_i), \]
where $H^\omega(p_0, p_i)$ is the half-space containing $p_0$ and delimited by the bisector $\mathcal{BS}(p_0,p_i)$ translated away from $p_0$ by $\omega_0$ (see Figure~\ref{fig-encompassing}).
\end{definition}

The second important step on our path towards the stability of power protection is the encompassing of the Riemannian Voronoi cell, and is detailed below.
\begin{lemma} \label{lemma-dilated_eroded_Voronoi_cells}
Let $U \subset \Omega$ be open, and $g$ and $g'$ be two Riemannian metric fields on $U$.
Let $\psi_0\geq 1$ be a bound on the metric distortion.
Suppose that $U$ is included in a ball $B_g(p_0, r_0)$, with $p_0 \in U$ and $r_0 \in \R^{+}$, such that ${\forall p\in B(p_0,r_0)}, {\psi(g(p), g'(p)) \leq \psi_0}$.
Let $\mathcal{P}_U = \{p_i\}$ be a point set in $U$.
We have
\[ \EV_{g_0}^{-\eta_0}(p_0) \subseteq \V_{g_0}^{-\omega_0}(p_0) \subseteq \V_{g_0}(p_0) \subseteq \V_{g_0}^{+\omega_0}(p_0) \subseteq \DV_{g_0}^{+\eta_0}(p_0). \]
These inclusions are illustrated in Figure~\ref{fig-encompassing}.
\end{lemma}
\begin{proof}
Using the notations and the result of Lemma~\ref{lemma-relaxed_bisector_bound}, we have
\[ \abs{m - m_{\omega_0}} \leq \eta_0. \]
Since the bisectors $\BS_{g'}^{\omega_0}(p_0,p_i)$ are hyperplanes, the result follows.
\end{proof}

In Figures~\ref{fig-encompassing} and~\ref{fig-Voronoi_vertex_stability}, $\DV_{g_0}^{+\eta_0}$ and $\EV_{g_0}^{-\eta_0}$ are shown in green and $\V_{g_0}$ in yellow.

\begin{figure}[!htb]
\centering
\includegraphics[width=0.7\linewidth]{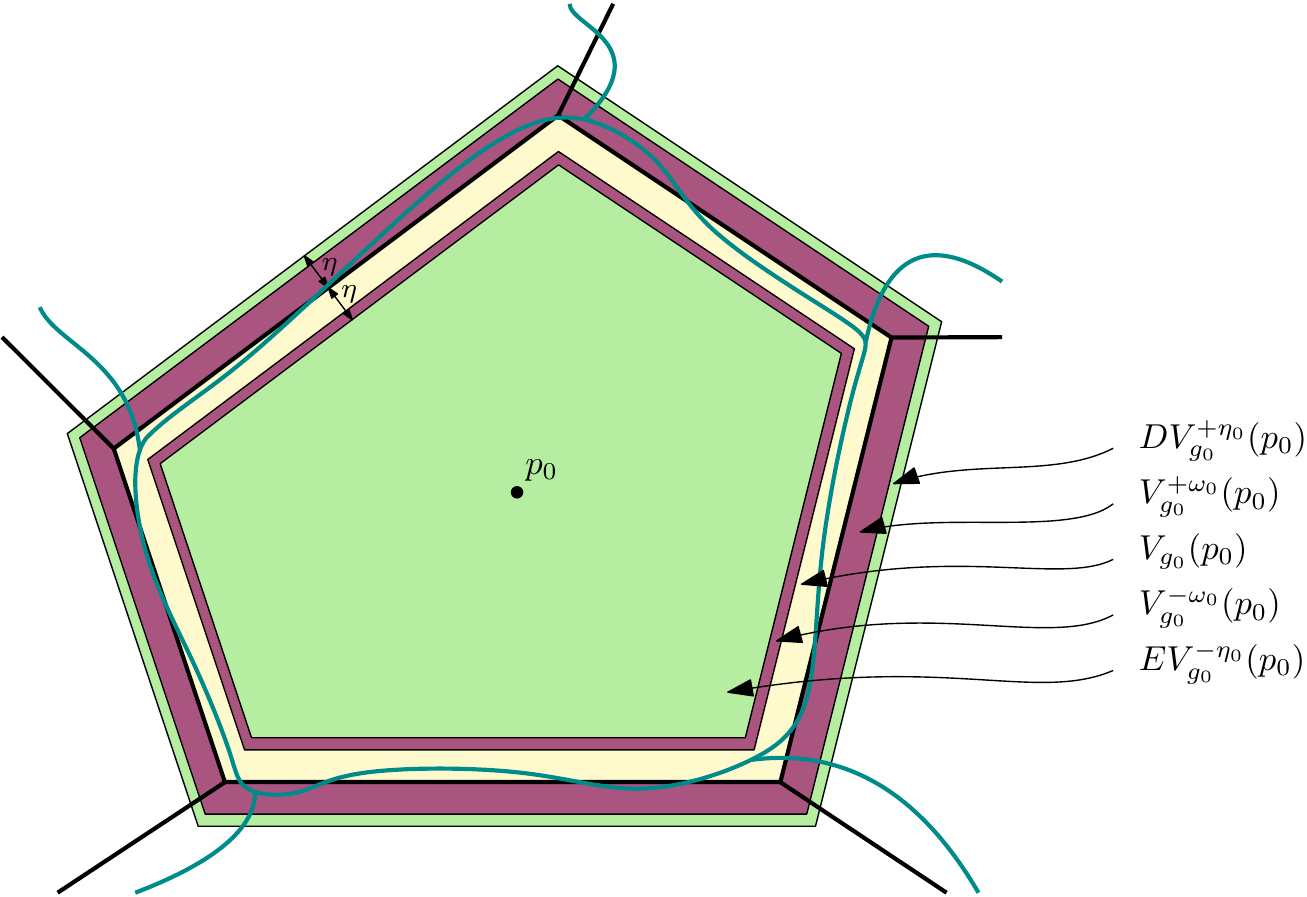}
\caption{The different encompassing cells around $p_0$.
         The Riemannian Voronoi diagrams with respect to $g$ and $g_0$ are traced in dark cyan and black respectively.
         The Voronoi cell $\V_{g_0}(p_0)$ is colored in yellow.
         The cells $\DV_{g_0}^{+\eta_0}$ and $\EV_{g_0}^{-\eta_0}$ are colored in green, and the cells $\V_{g_0}^{\pm\omega_0}(p_0)$ are colored in purple.}
\label{fig-encompassing}
\end{figure}

The next step is to prove that we have stability of the Voronoi vertices of $\V_g(p_0)$, meaning that a small perturbation of the metric only creates a small displacement of any Voronoi vertex of the Voronoi cell of $p_0$.
The following lemma shows that Voronoi vertices are close if the distortion between the metric fields $g$ and $g_0$ is small.
The approach is to use Lemmas~\ref{lemma-relaxed_Voronoi_cell} and~\ref{lemma-dilated_eroded_Voronoi_cells}: we know that each $(n-1)$-face of the Riemannian Voronoi cell $\V_g(p_0)$ and shared by a Voronoi cell $V_g(q)$ is enclosed within the bisectors of $\DV_{g_0}(p_0)$ and $\EV_{g_0}(p_0)$ (which are translations of the bisectors of $\V_{g_0}(p_0)$) for the sites $p_0$ and $q$.
These two bisectors create a ``band'' that contains the bisector $\BS_g(p_0, q)$.
Given a Voronoi vertex $c$ in $\V_g(p_0)$, $c$ can be obtained as the intersection of $n+1$ Voronoi cells, but also as the intersection of $n$ Voronoi $(n-1)$-faces of $\V_g(p_0)$.
The intersection of the bands associated to those $n$ $(n-1)-$faces is a parallelotopic-shaped region which by definition contains the same (combinatorially speaking) Voronoi vertex, but for $\V_{g_0}(p_0)$.
Lemmas~\ref{lemma-voronoi_vertices_metric_perturbation_2D} and~\ref{lemma-voronoi_vertices_metric_perturbation} express this reasoning, which is illustrated in Figure~\ref{fig-Voronoi_vertex_stability} for $2$D and~\ref{fig-cylinder_intersection} for any dimension.
\begin{lemma} \label{lemma-voronoi_vertices_metric_perturbation_2D}
We consider here $\Omega = \R^2$.
Let $U \subset \Omega$ be open, and $g$ and $g'$ be two Riemannian metric fields on $U$.
Let $\psi_0\geq 1$ be a bound on the metric distortion.
Suppose that $U$ is included in a ball $B_g(p_0, r_0)$, with $p_0 \in U$ and $r_0 \in \R^{+}$, such that ${\forall p\in B(p_0,r_0)}, {\psi(g(p), g'(p)) \leq \psi_0}$.
Let $\mathcal{P}_U = \{p_i\}$ be a point set in $U$.
Assume furthermore that $\mathcal{P}_U$ is a $\delta_0$-power protected $(\varepsilon_0, \mu_0)$-net with respect to $g_0$ (the Euclidean metric).
Let $p_0\in\mathcal{P}_U$.
Let $c$ and $c_0$ be the same Voronoi vertex in respectively $\V_g(p_0)$ and $\V_{g_0}(p_0)$.
Then $d_{g_0}(c,c_0) \leq \chi_2$ with
\[ \chi_2 = \frac{2\eta_0}{\sqrt{1 + \frac{\mu_0}{2\varepsilon_0}} - \sqrt{1 - \frac{\mu_0}{2\varepsilon_0}}} 
          = \frac{2\eta_0}{\sqrt{1 + \frac{\lambda}{2\psi_0^2}} - \sqrt{1 - \frac{\lambda}{2\psi_0^2}}}. \]
where $\lambda$ is given by $\mu_0 = \frac{\lambda}{\psi_0^2} \varepsilon_0$ (see Lemma~\ref{remark-epsilon_mu_prime}).
\end{lemma}
\begin{proof}
We use Lemma~\ref{lemma-dilated_eroded_Voronoi_cells}.
$\V_{g}(p_0)$ lies in $\DV_{g_0}^{+\eta_0}$ and contains $\EV_{g_0}^{-\eta_0}$.
The circumcenters $c$ and $c_0$ lie in a parallelogrammatic region centered on $c_0$, itself included in the ball centered on $c_0$ and with radius $\chi$.
The radius $\chi$ is given by half the length of the longest diagonal of the parallelogram (see Figure~\ref{fig-Voronoi_vertex_stability}).
By Lemma~\ref{lemma-net_metric_perturbation}, $\mathcal{P}$ is an $(\varepsilon_{0},\mu_{0})$-net with respect to $g_0$.
Let $\theta$ be the angle of the Voronoi corner of $\V_{g_0}(p_0)$ at $c_0$.
By Lemma~\ref{lemma-Voronoi_angle_bounds}, that angle is bounded:
\[ \theta_m = 2\arcsin\left(\frac{\mu_0}{2\varepsilon_0}\right) \leq \theta \leq \pi-\arcsin\left(\frac{\mu_0}{2\varepsilon_0}\right) = \theta_M. \]
Since $\pi - \theta_M < \theta_m$, $\chi$ is maximal when $\theta > \pi / 2$.
We thus assume $\theta > \pi / 2$, and compute a bound on $\chi$ as follows:
\begin{align*}
\sin\left(\frac{\pi - \theta}{2}\right) = \frac{\eta_0}{\chi} \Longrightarrow \chi &= \frac{\eta_0}{\sin\left(\frac{\pi - \theta}{2}\right)} \\
                                                                                &\leq \frac{\eta_0}{\sin\left(\frac{1}{2}\arcsin\left(\frac{\mu_0}{2\varepsilon_0}\right)\right)} \\
                                                                                &\leq \frac{2\eta_0}{\sqrt{1 + \frac{\mu_0}{2\varepsilon_0}} - \sqrt{1 - \frac{\mu_0}{2\varepsilon_0}}} =: \chi_2,
\end{align*}
using $\sin(\frac{1}{2}\arcsin(\alpha)) = \frac{1}{2}\left(\sqrt{1+\alpha} - \sqrt{1-\alpha}\,\right)$.

\begin{figure}[!htb]
\centering
\includegraphics[width=\linewidth]{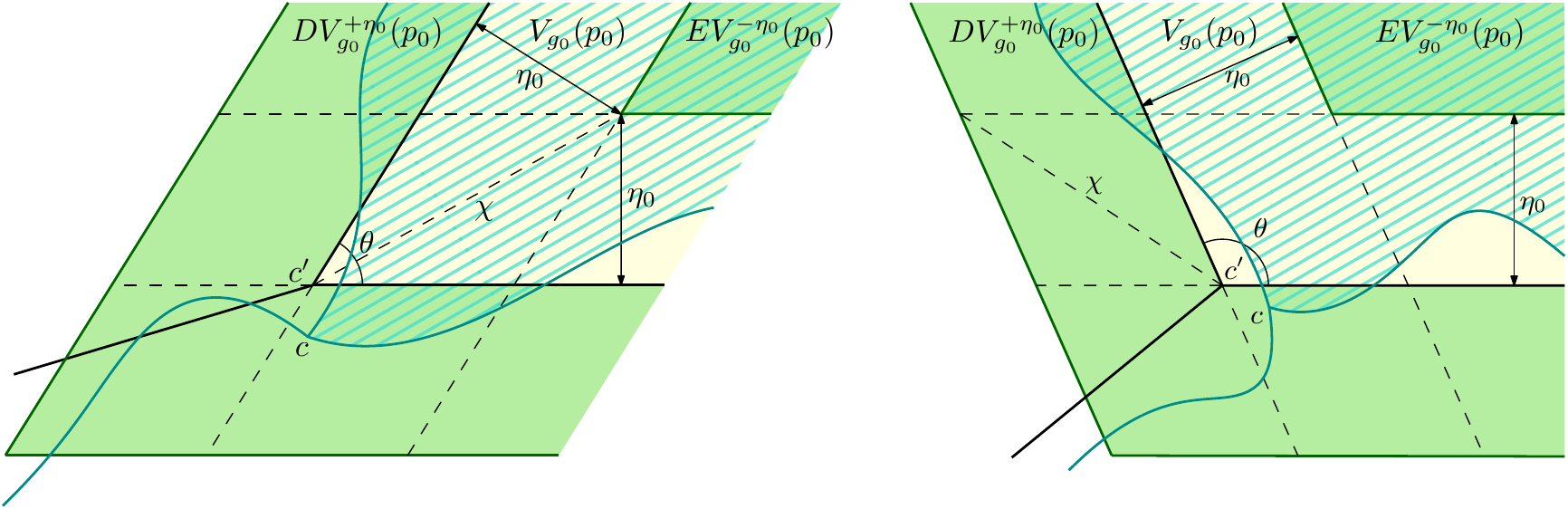}
\caption{Black lines trace $\RVD_{g_0}$ and cyan lines trace $\RVD_{g}$.
         The cell $\V_{g_0}(p_0)$ is colored in yellow and the cell $\V_{g}(p_0)$ is dashed.
         The green regions correspond to $\DV_{g_0}^{+\eta}(p_0)$ and $\EV_{g_0}^{-\eta}(p_0)$.
         On the left, the configuration where $\theta<\pi/2$; on the right, $\theta > \pi/2$.}
\label{fig-Voronoi_vertex_stability}
\end{figure}
\end{proof}

The result obtained in Lemma~\ref{lemma-voronoi_vertices_metric_perturbation_2D} can be extended to any dimension using induction and the stability of the Voronoi vertices of facets.
\begin{lemma} \label{lemma-voronoi_vertices_metric_perturbation}
We consider here $\Omega = \R^n$.
Let $U \subset \Omega$ be open, and $g$ and $g'$ be two Riemannian metric fields on $U$.
Let $\psi_0\geq 1$ be a bound on the metric distortion.
Suppose that $U$ is included in a ball $B_g(p_0, r_0)$, with $p_0 \in U$ and $r_0 \in \R^{+}$, such that ${\forall p\in B(p_0,r_0)}, {\psi(g(p), g'(p)) \leq \psi_0}$.
Let $\mathcal{P}_U = \{p_i\}$ be a $\delta_0$-power protected $(\varepsilon_0, \mu_0)$-net with respect to $g_0$ (the Euclidean metric).
Let $p_0\in\mathcal{P}_U$.
Let $c$ and $c_0$ be the same Voronoi vertex in respectively $\V_g(p_0)$ and $\V_{g_0}(p_0)$.
Then $d_{g_0}(c,c_0) \leq \chi$ with
\[ \chi = \frac{\chi_2}{\sin^{n-2}\left(\frac{\varphi_0}{2}\right)}.\]
where $\chi_2$ is defined as in Lemma~\ref{lemma-voronoi_vertices_metric_perturbation_2D}, and $\varphi_0$ is the maximal dihedral angle between two faces of a simplex.
\end{lemma}
\begin{proof}
We know from Lemma~\ref{lemma-relaxed_Voronoi_cell} that $\V_{g}(p_0)$ lies in $\DV_{g_0}^{+\eta_0}$ and contains $\EV_{g_0}^{-\eta_0}$.
The circumcenters~$c$ and~$c_0$ lie in a parallelotopic region centered on~$c_0$ defined by the intersection of~$n$ Euclidean thickened Voronoi faces.
This parallelotope and its circumscribing sphere are difficult to compute.
However, it can be seen as the intersection of two parallelotopic tubes defined by the intersection of $n-1$ Euclidean thickened Voronoi faces.
From another point of view, this is the computation of the intersection of the thickened duals of two facets $\tau_1$ and $\tau_2$ incident to $p_0$ of the simplex $\sigma\in Del(\mathcal{P})$, dual of $c$ (and $c_0$), see Figure~\ref{fig-cylinder_intersection} (left).

\begin{figure}[!htb]
\centering
\includegraphics[clip, trim=4.cm 11cm 4.cm 6cm, height=5cm]{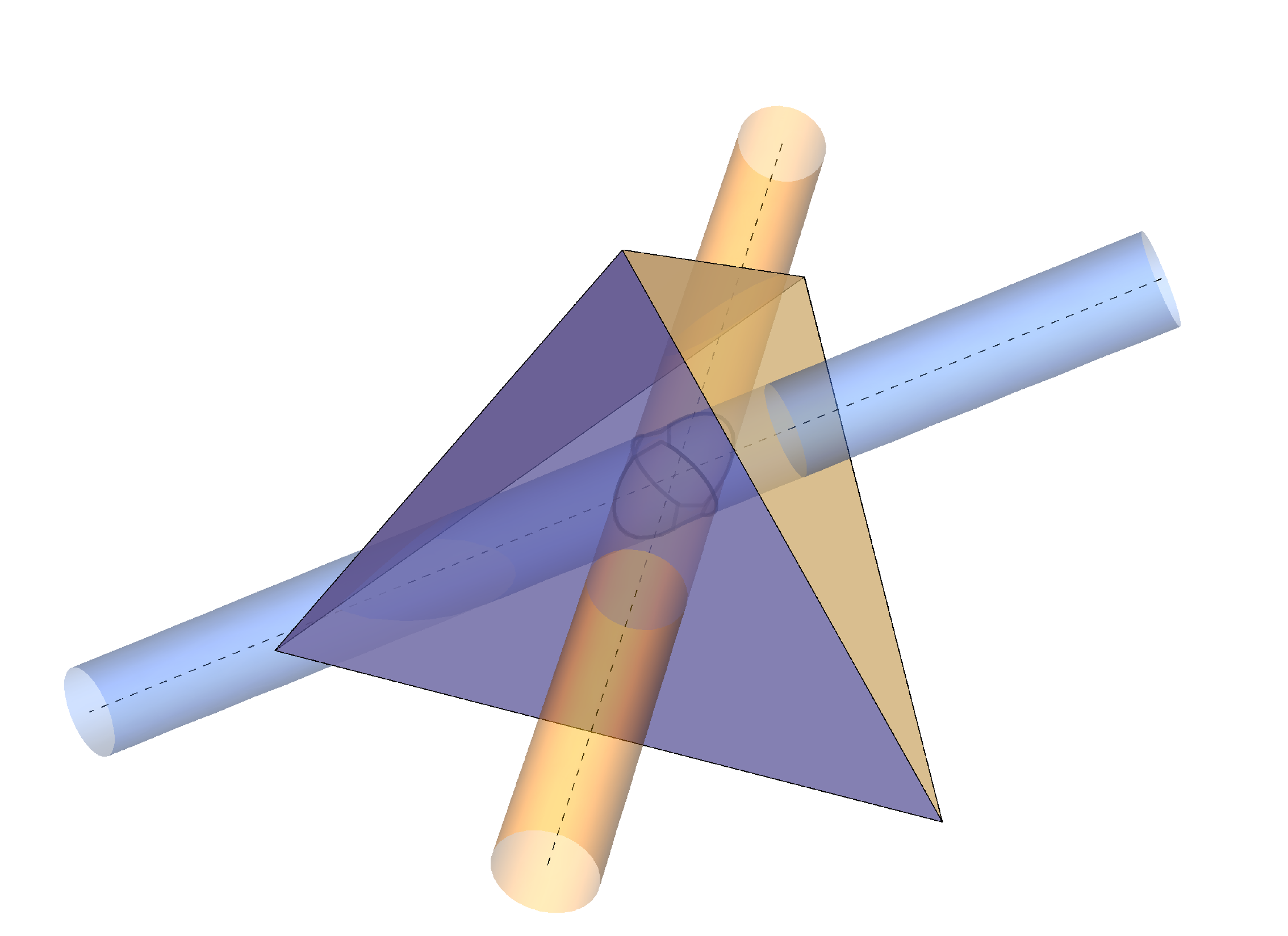}
\includegraphics[clip, trim=4.cm 5cm 2.cm 5cm, height=4.5cm]{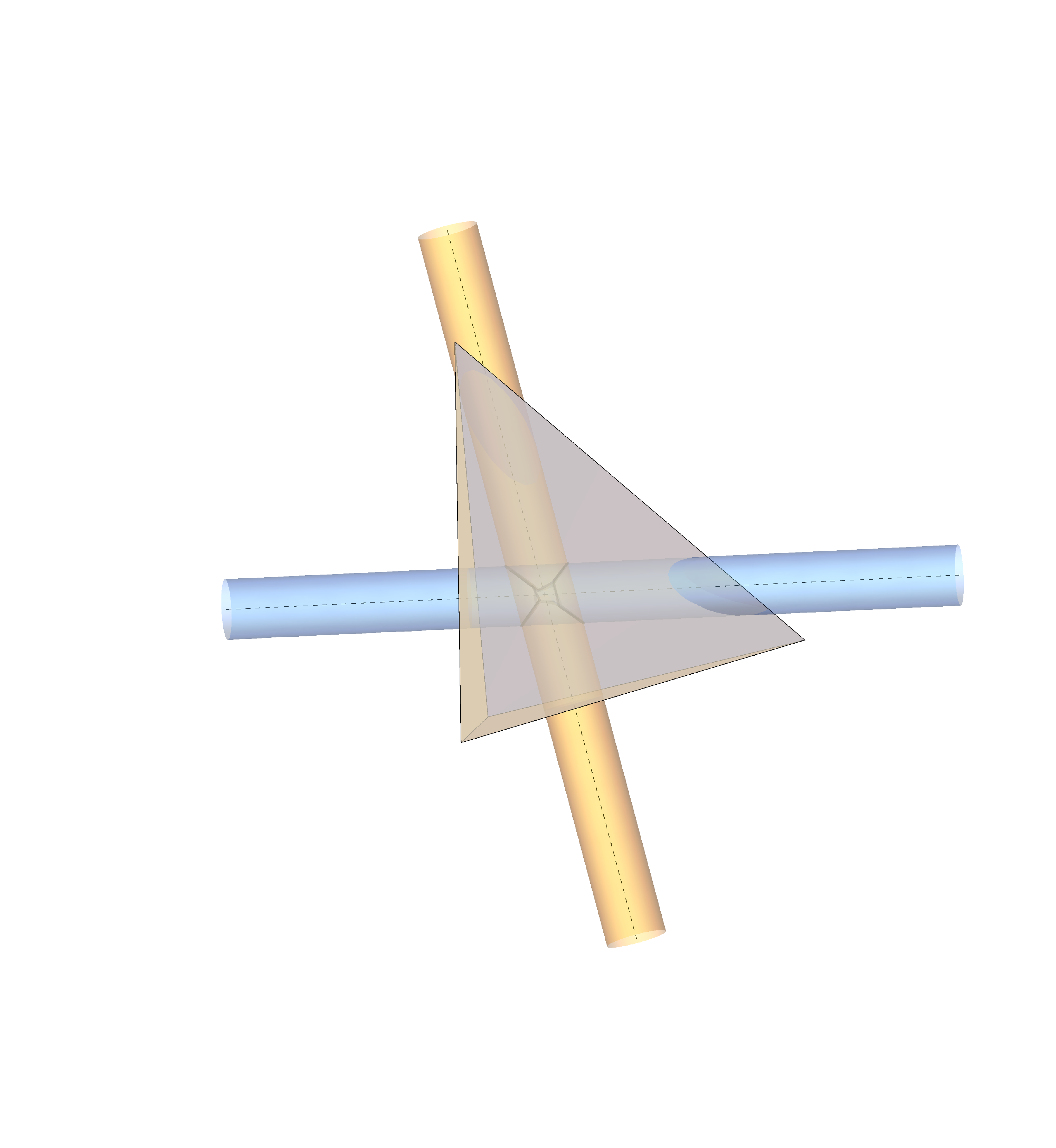}
\caption{Left, a simplex, the duals of two faces and their respective thickened duals (cylinders); the orange thickened dual is orthogonal to the purple face (and inversely).
         Right, the intersection of the two tubes, seen from above, illustrates the proof of Lemma~\ref{lemma-voronoi_vertices_metric_perturbation}.}
\label{fig-cylinder_intersection}
\end{figure}

The stability radius $\chi$ is computed incrementally by increasing the dimension and proving stability of the circumcenters of the faces of the simplices.
We prove the formula by induction.

The radius of each tube is given by the stability of the radius of the circumcenter in the lower dimension of the facet.
The base case, $\R^n = \R^3$, is solved by Lemma~\ref{lemma-voronoi_vertices_metric_perturbation_2D}, and gives
\[ \chi_2 = \frac{2\eta_0}{\sqrt{1 + \frac{\lambda}{2\psi_0^2}} - \sqrt{1 - \frac{\lambda}{2\psi_0^2}}}. \]
We now consider two facets $\tau_1$ and $\tau_2$ that are incident to $p_0$.
Denote $\mathcal{D}_1$ and $\mathcal{D}_2$ their respective duals.
By Lemma~\ref{lemma-relaxed_Voronoi_cell}, $\mathcal{D}_1$ and $\mathcal{D}_2$ lie in two cylinders $C_1$ and $C_2$ of radius $\chi_2$.
$C_1$ and $C_2$ are also orthogonal to $\tau_1$ and $\tau_2$ and $c$ and $c_0$ lie in $C_1 \cap C_2$.
The angle $\varphi$ between $C_1$ and $C_2$ is exactly the dihedral angle between $\tau_1$ and $\tau_2$.
By Lemma~\ref{lemma-dihedral_angle_from_delta_pp}, we have
\[ \arcsin(s_0) \leq \varphi \leq \pi - \arcsin(s_0) \text{ with } s_0 = \frac{1}{2} \left[ \frac{\iota^2}{4\psi_0^2} - \frac{1}{2}\left(\psi_0^2 - \frac{1}{\psi_0^2}\right) \right] \]
Let $\varphi_0 = \arcsin(s_0) = \pi-\arcsin(s_0)$.
We encompass the intersection of the cylinders, difficult to compute, with a sphere whose radius can be computed as follows (see Figure~\ref{fig-cylinder_intersection}, right):
\[ \chi_3 = \frac{\chi_2}{\cos(\alpha)} \text{ with } \alpha = \frac{\pi}{2} - \frac{\varphi_0}{2}. \]
Thus,
\[ \chi_3 = \frac{\chi_2}{\sin\left(\frac{\varphi_0}{2}\right)}. \]
Recursively,
\[ \chi = \frac{\chi_2}{\sin^{n-2}\left(\frac{\varphi_0}{2}\right)}. \]
\end{proof}

We have assumed in different lemmas that we could pick values of $\rho_0$ or $\omega_0$ that fit our need.
The following lemma shows that these assumptions were reasonable.
\begin{lemma} \label{lemma-choice_of_rho}
Lemmas~\ref{lemma-dilated_eroded_Voronoi_cells} and~\ref{lemma-voronoi_vertices_metric_perturbation} allow us to characterize the parameter $\rho_0$ more precisely.
Indeed, an assumption of Lemma~\ref{lemma-relaxed_Voronoi_cell} was that $\V_{g_0}^{\omega_0}(p_0)$ is included in a ball $B_{g_0}(p_0, \rho_0)$.
If the sampling of $\mathcal{P}$ is sufficiently dense, such an assumption is reasonable.
\end{lemma}
\begin{proof}
By definition, the Voronoi cell $\V_{g_0}^{\omega_0}(p_0)$ is included in the dilated cell $\DV_{g_0}^{+\eta_0}(p_0)$.
Since the point set is an $\varepsilon$-sample, we have $d_{g_0}(p_0, x) \leq \varepsilon_0$ for $x\in \V_{g_0}(p_0)$.
By Lemma~\ref{lemma-voronoi_vertices_metric_perturbation}, we have for $x\in \DV_{g_0}^{+\eta_0}(p_0)$
\[d_{g_0}(p_0, x) \leq \varepsilon_0 + \chi.\]
Recall from Lemma~\ref{lemma-voronoi_vertices_metric_perturbation} that
\[ \chi = \frac{\chi_2}{\left[\sin\left(\frac{1}{2}\arcsin\left(s_0\right)\right)\right]^{n-2}}. \]
with $\eta_0 = \frac{\rho_0^2(\psi_0^2-1)}{\mu_0}$ and $s_0 = \frac{1}{2} \left[ \frac{\iota^2}{4\psi_0^2} - \frac{1}{2}\left(\psi_0^2 - \frac{1}{\psi_0^2}\right) \right]$.
We require $\V_{g_0}^{\omega_0}(p_0) \subset B_{g_0}(p_0,\rho_0)$, which is verified if $\DV_{g_0}^{+\eta_0}(p_0) \subset B_{g_0}(p_0, \rho_0)$, that is if
\begin{align*}
\varepsilon_0 + \frac{\chi_2}{\left[\sin\left(\frac{1}{2}\arcsin\left(s_0\right)\right)\right]^{n-2}} &\leq \rho_0 \\
\iff \frac{4\eta_0}{\left( \sqrt{1 + \frac{\lambda}{2\psi_0^2}} - \sqrt{1 - \frac{\lambda}{2\psi_0^2}} \right) \Big(\sqrt{1 + s_0} - \sqrt{1 - s_0}\,\Big)^{n-2}} &\leq \rho_0 - \varepsilon_0 \\
\iff \frac{4\rho_0^2(\psi_0^2 - 1)}{\mu_0 \left( \sqrt{1 + \frac{\lambda}{2\psi_0^2}} - \sqrt{1 - \frac{\lambda}{2\psi_0^2}} \right) \Big(\sqrt{1 + s_0} - \sqrt{1 - s_0}\,\Big)^{n-2}} &\leq \rho_0 - \varepsilon_0 \\
\iff \frac{4(\psi_0^2 - 1)}{\left( \sqrt{1 + \frac{\lambda}{2\psi_0^2}} - \sqrt{1 - \frac{\lambda}{2\psi_0^2}} \right) \Big(\sqrt{1 + s_0} - \sqrt{1 - s_0}\,\Big)^{n-2}} &\leq \frac{\mu_0(\rho_0 - \varepsilon_0)}{\rho_0^2} \\
\iff \frac{4\psi_0(\psi_0^2 - 1)}{\left( \sqrt{1 + \frac{\lambda}{2\psi_0^2}} - \sqrt{1 - \frac{\lambda}{2\psi_0^2}} \right) \Big(\sqrt{1 + s_0} - \sqrt{1 - s_0}\,\Big)^{n-2}} &\leq \frac{\lambda\varepsilon(\rho_0 - \varepsilon_0)}{\rho_0^2} \text{, by Remark~\ref{remark-epsilon_mu_prime}.}
\end{align*}

The parameter $\rho_0$ can be chosen arbitrarily as long as it is greater than $\varepsilon_0$, and we have taken $\rho_0 = 2\varepsilon_0$ (see Definition~\ref{definition-rho_eta}), which imposes
\begin{align} \label{eq-conditions_for_rho}
\frac{\psi_0^2(\psi_0^2 - 1)}{\left( \sqrt{1 + \frac{\lambda}{2\psi_0^2}} - \sqrt{1 - \frac{\lambda}{2\psi_0^2}} \right) \Big(\sqrt{1 + s_0} - \sqrt{1 - s_0}\,\Big)^{n-2}} &\leq \frac{\lambda}{16}.
\end{align}
Recall that the parameter $\lambda$ is fixed.
By continuity of the metric field, $\lim_{\varepsilon \to 0} \psi_0 = 1$, therefore the left hand side goes to $0$ and Inequality~\eqref{eq-conditions_for_rho} is eventually satisfied as the sampling is made denser.
\end{proof}

Finally, we can now show the main result: the power protection property is preserved when the metric field is perturbed.
\begin{lemma} \label{lemma-protection_metric_perturbation}
Let $U \subset \Omega$ be open, and $g$ and $g'$ be two Riemannian metric fields on $U$.
Let $\psi_0\geq 1$ be a bound on the metric distortion.
Suppose that $U$ is included in a ball $B_g(p_0, r_0)$, with $p_0 \in U$ and $r_0 \in \R^{+}$, such that ${\forall p\in B(p_0,r_0)}, {\psi(g(p), g'(p)) \leq \psi_0}$.
Assume that $\mathcal{P}_U$ is a $\delta$-power protected $(\varepsilon,\mu)$-net in $U$ with respect to $g$.
If $\delta$ is well chosen, then $\mathcal{P}_U$ is a $\delta_0$-power protected net with respect to $g_0$, with 
\[ \delta_0^2 = \left(\frac{1}{\psi_0^2} - \psi_0^2\right) (\varepsilon + \chi)^2 - \frac{4\varepsilon\chi}{\psi_0^2} + \frac{\delta^2}{\psi_0^2}. \]
\end{lemma}
\begin{proof}
By Lemma~\ref{lemma-net_metric_perturbation}, we know that $\mathcal{P}_U$ is $(\varepsilon_0, \mu_0)$-net with respect to~$g_0$.
Let $q\in\mathcal{P}_U$, with $q$ not a vertex in the dual of $c$.
Let $c_0$ be the combinatorial equivalent of $c$ in $\V_{g_0}(\mathcal{P})$
Since $\mathcal{P}$ is a $\delta$-power protected net with respect to $g$, we have $d_g(c,q) > \sqrt{r^2 + \delta^2}$, where~$r= d_g(c,p)$.
On the one hand, we have
\begin{align*}
d_{g_0}(c_0, q) &\geq d_{g_0}(q,c) - d_{g_0}(c, c_0) \\
              &\geq \frac{1}{\psi_0} d_g(q, c) - \chi \\
              &\geq \frac{1}{\psi_0} \sqrt{r^2 + \delta^2} - \chi.
\end{align*}
by Lemma~\ref{lemma-voronoi_vertices_metric_perturbation_2D}.
On the other hand, for any~$p\in\mathcal{P}_U$ such that $p$ is a vertex of the dual of $c$, we have
\begin{align*}
r_0 = d_{g_0}(c_0, p) &\leq d_{g_0}(c, p) + d_{g_0}(c_0, c) \\
                   &\leq \psi_0 d_g(c,p) + \chi \\
                   &\leq \psi_0 r + \chi.
\end{align*}
Thus $\delta_0$-power protection of $\mathcal{P}_U$ with respect to $g_0$ requires
\begin{align*}
&\frac{1}{\psi_0} \sqrt{r^2 + \delta^2} - \chi > \chi + \psi_0 r \\
\iff &\sqrt{r^2 + \delta^2} > \psi_0(2\chi + \psi_0 r).
\end{align*}
This is verified if
\begin{align*}
\delta^2 > \left(\psi_0(2\chi + \psi_0 r)\right)^2 - r^2 &= 4\chi^2\psi_0^2 + 4\chi\psi_0^3 r + \psi_0^4 r^2 - r^2 \\
                                                         &= 4\chi^2\psi_0^2 + 4\chi\psi_0^3 r + (\psi_0^4 - 1) r^2, 
\end{align*}
for all $r\in [ \mu / 2, \varepsilon ]$.
This gives us
\begin{equation}
\delta^2 > 4\chi^2\psi_0^2 + 4\chi\psi_0^3 \varepsilon + (\psi_0^4 - 1) \varepsilon^2. \label{first-bound-Delta}
\end{equation}

This condition on $\delta$ is only reasonable if the right hand side is not too large.
Indeed, since $\mathcal{P}$ is an $\varepsilon$-sample, we must have $d_{g}(c, q) < 2\varepsilon$.
However, we have that $d_{g}(c, q)^2 > d_{g}(c,p)^2 + \delta^2$ by $\delta$-power protection of $\mathcal{P}$ with respect to $g$.
Because $d_{g}(c, p) < \varepsilon$, it suffices that $\delta < \varepsilon$.
We will now show that this is reasonable by examining the limit of the right hand side of Inequality~\eqref{first-bound-Delta}.

We note, see Lemma~\ref{lemma-voronoi_vertices_metric_perturbation}, that 
\begin{align*}
\chi = \frac{\chi_2}{\sin^{n-2}\left(\frac{\varphi}{2}\right)} = \frac{4\eta}{\left( \sqrt{1 + \frac{\lambda}{2\psi_0^2}} - \sqrt{1 - \frac{\lambda}{2\psi_0^2}} \right) \Big(\sqrt{1 + s_0} - \sqrt{1 - s_0}\,\Big)^{n-2}},
\end{align*}
where $\varepsilon_0 = \psi_0\varepsilon$ and $\mu_0 = \mu/\psi_0 = \lambda\varepsilon/\psi_0$ (see Remark~\ref{remark-epsilon_mu_prime}).
So that
\begin{align}
4\chi^2\psi_0^2 + 4\chi\psi_0^3\varepsilon + (\psi_0^4 - 1) \varepsilon^2 &= 
4 \left(\frac{16 \varepsilon\psi_0^4(\psi_0^2-1)}{\lambda \left( \sqrt{1 + \frac{\lambda}{2\psi_0^2}} - \sqrt{1 - \frac{\lambda}{2\psi_0^2}} \right) \Big(\sqrt{1 + s_0} - \sqrt{1 - s_0}\,\Big)^{n-2}}\right)^2 \nonumber \\
&+ 4 \frac{16\varepsilon \psi_0^3(\psi_0^2 - 1)}{\lambda \left( \sqrt{1 + \frac{\lambda}{2\psi_0^2}} - \sqrt{1 - \frac{\lambda}{2\psi_0^2}} \right) \Big(\sqrt{1 + s_0} - \sqrt{1 - s_0}\,\Big)^{n-2}} \psi_0^3 \nonumber \\
&+ (\psi_0^2 - 1)(\psi_0^2 + 1) \varepsilon^2. \label{Exegesis-Bound}
\end{align}
This means that the right hand side of \eqref{first-bound-Delta} is of the form $f(\psi_0)( \psi_0 ^2- 1) \varepsilon^2$, where $f(\psi_0)$ is a function that tends to a constant as $\psi_0$ goes to $1$:
\begin{align*}
f(\psi_0) = 4&\left(\frac{16 \psi_0^4}{\lambda \left( \sqrt{1 + \frac{\lambda}{2\psi_0^2}} - \sqrt{1 - \frac{\lambda}{2\psi_0^2}} \right) \Big(\sqrt{1 + s_0} - \sqrt{1 - s_0}\,\Big)^{n-2}}\right)^2 \nonumber \\
          &+ 4 \frac{16 \psi_0^6}{\lambda \left( \sqrt{1 + \frac{\lambda}{2\psi_0^2}} - \sqrt{1 - \frac{\lambda}{2\psi_0^2}} \right) \Big(\sqrt{1 + s_0} - \sqrt{1 - s_0}\,\Big)^{n-2}} \nonumber \\
          &+ (\psi_0^2 + 1) \\
\overset{\psi_0 \to 1}{\longrightarrow}
          4 &\left(\frac{16}{\lambda \left( \sqrt{1 + \frac{\lambda}{2}} - \sqrt{1 - \frac{\lambda}{2}} \right) \left(\sqrt{1 + \frac{\iota^2}{4}} - \sqrt{1 - \frac{\iota^2}{4}}\,\right)^{n-2}}\right)^2 \nonumber \\
          &+ 4 \frac{16}{\lambda \left( \sqrt{1 + \frac{\lambda}{2}} - \sqrt{1 - \frac{\lambda}{2}} \right) \left(\sqrt{1 + \frac{\iota^2}{4}} - \sqrt{1 - \frac{\iota^2}{4}}\,\right)^{n-2}} \nonumber \\
          &+ 2
\end{align*}
So the bound given in Equality~\eqref{first-bound-Delta} may be easily satisfied if the metric distortion is sufficiently small.

We now provide an explicit value for~$\delta_0$ in terms of~$\delta$.
Let~$\xi = d_{g}(c,q)$ and~$\xi_0 = d_{g_0}(c_0, q_0)$.
We have the following bounds on~$r_0$ and~$\xi_0$:
\begin{align*}
\frac{1}{\psi_0} (r - \xi) &\leq r_0 \leq \psi_0 (r + \xi) \\
\frac{1}{\psi_0} \sqrt{(r-\chi)^2 + \delta^2} &\leq \xi_0 \leq \psi_0 \sqrt{ (r+\chi)^2 + \delta^2}.
\end{align*}
If we had $\tilde{\delta}$-power protection, we would have
\begin{align*}
r_0^2 + \tilde{\delta}^2 \leq \xi_0^2 &\iff \tilde{\delta}^2 \leq \xi_0^2 - r_0^2 \\
                             &\iff \tilde{\delta}^2 \leq \frac{1}{\psi_0^2}\left( (r-\chi)^2 + \delta^2 \right) - \psi_0^2(r+\chi)^2 \\
                             &\iff \tilde{\delta}^2 \leq \frac{1}{\psi_0^2} (r + \chi)^2 - \frac{4 r\chi}{\psi_0^2} + \frac{\delta^2}{\psi_0^2} - \psi_0^2(r+\chi)^2 \\
                             &\Longleftarrow \tilde{\delta}^2 \leq \left(\frac{1}{\psi_0^2} - \psi_0^2\right) (\varepsilon + \chi)^2 - \frac{4\varepsilon\chi}{\psi_0^2} + \frac{\delta^2}{\psi_0^2}.
\end{align*}

Therefore we can take $\delta_0^2 = \frac{\delta^2}{\psi_0^2}+ \left(\frac{1}{\psi_0^2} - \psi_0^2\right) (\varepsilon + \chi)^2 - \frac{4\varepsilon\chi}{\psi_0^2}$.
Note that with this definition, $\delta_0$ goes to $\delta$ as $\psi_0$ goes to $1$, which proves that our value of $\delta_0$ is legitimate.
\end{proof}

\section{Embeddability of the straight Delaunay triangulation (Proofs of Section~\ref{section-embeddability_straight})} \label{appendix-proofs_straight_embeddability}
We first prove Lemma~\ref{lemma-geo_straight_proximity}, recalled below, which bounds the distance between the same point on a the Karcher and the straight simplex.

\begin{lemma*}
Let $\mathcal{P}$ be an $\varepsilon$-sample with respect to $g$ on $\Omega$.
Let $\{ p_i \}$ be a set of $n+1$ vertices in $\mathcal{M}$ such that $N = \cap_{p_i\in\mathcal{P}} V(p_i) \neq \emptyset$.
Let $\bar{\sigma}$ and $\widetilde{\sigma}$ be the straight and Karcher simplices that realize $N$.
Let $\widetilde{x}$ be a point on the Karcher simplex $\widetilde{\sigma}$ determined by the barycentric coordinates $\{ \lambda_i \}$ (see Equation~\ref{equation-hattilde}).
Let $x_e$ be the point uniquely determined by~$\{ \lambda_i \}$ as $x_e = \sum_i \lambda_i p_i$.
If the geodesic distance $d_g$ is close to $d_{\E}$ -- that is the distortion $\psi(g, g_{\E})$ is bounded by $\psi_0$ -- then $\abs{\widetilde{x} - x_e} \leq \sqrt{ 2 \cdot 4^3 (\psi_0 - 1) \varepsilon^2 }$.
\end{lemma*}
\begin{proof}
The key observation is that given a convex function $f$ and a function $f'$ that is close, that is $f-f' < \alpha$ with $\alpha$ small, then the minimum value of $f'$ is at most of $\min f+\alpha$.
If we observe that at any point $x$ where $f(x)> \min f +2 \alpha$, we also have $f'(x)>\min f +\alpha$ so $x$ is not a minimum of $f'$, we see that the minima of $f$ and $f'$ can not be far apart.
In particular, we have that if $x_{f',\min}$ is the point where $f'$ attains its minimum, then $f'(x_{f',\min}) \leq \min f + \alpha$.
The precise argument goes as follows.

We again assume that (possibly after a linear transformation) the metric is close to the Euclidean one, that is:
\[ d_g(x,y) = \abs{x-y} + \delta d_g(x,y), \]
with $\abs{\delta d_g(x,y)} \leq (\psi_0 - 1) \abs{x-y} $.
If we assume that $\abs{x-y} \leq 4 \varepsilon$ and $\psi_0 \leq 2$, it follows that
\[ d_g(x,y) ^2= |x-y|^2+ \delta d_g^2 (x,y), \]
with
\[ \delta d_g^2 (x,y)\leq 4^3 (\psi_0 - 1) \varepsilon^2. \]

Recall that $\widetilde{x}$ is the point where the functional
\[ \mathcal{E}_{\lambda} (x) = \sum_i \lambda_i d_g(x, p_i)^2 \]
attains its minimum.

Using the bounds above, we find that
\[ \sum_i \lambda_i d_g(x, p_i)^2 = \sum_i \lambda_i \abs{x - p_i}^2 + \sum_i \lambda_i \delta d_g^2 (x, p_i), \]
where
$\sum_i \lambda_i \delta d_g^2 (x,p_i) \leq 4^3 (\psi_0 - 1) \varepsilon^2$.
We also see that
\[ \abs{ \sum_i \lambda_i d_g(\widetilde{x} ,p_i)^2 - \sum_i \lambda_i \abs{\widetilde{x} -p_i}^2 } \leq 4^3 (\psi_0 - 1) \varepsilon^2. \]
Taking $f'$ to be $\sum_i \lambda_i d_g(\widetilde{x} ,p_i)^2$ in the explanation above, we find that 
\[ \abs{ \sum_i \lambda_i d_g(\widetilde{x} ,p_i)^2- \sum_i \lambda_i \abs{\sum_j \lambda_j v_j -p_i}^2 } \leq 4^3 (\psi_0 - 1) \varepsilon^2, \]
because the Euclidean barycenter $x_e = \sum_i \lambda_i p_i$ is where the function $\sum_i \lambda_i |\widetilde{x} -p_i|^2$ attains its minimum.
Combining these results yields
\begin{equation}
\abs{ \sum_i \lambda_i \abs{ \widetilde{x} -p_i}^2- \sum_i \lambda_i \abs{ \sum_j \lambda_j v_j - p_i}^2 } \leq 2  \cdot 4^3 (\psi_0 - 1) \varepsilon^2. \label{equation-geo_straight}
\end{equation}
This bounds the distance between $\widetilde{x}$ and $x_e$.
An explicit bound can be found by observing that
\begin{align*}
 \sum_i \lambda_i \abs{ x- p_i }^2  = &\sum_i \lambda_i \left(\left(x^1- p_i^1 \right)^2 + \ldots + \left(x^n - p_i^n \right)^2\right) \\
                                    = &\sum_i \lambda_i \left(x^1- p_i^1 \right)^2 + \sum_i \lambda_i \left(x^2- p_i^2 \right)^2 + \ldots + \sum_i \lambda_i \left( x^n- p_i^n \right)^2 \\
                                    = &\left(x^1- \sum_i \lambda_i p^1_i\right)^2 - \left(\sum_i \lambda_i p^1_i\right)^2 + \sum_i \lambda_i \left(p_i^1 \right)^2+ \ldots \\
                                      &\qquad + \left(x^n- \sum_i \lambda_i p^n_i\right)^2 - \left(\sum_i \lambda_i p^n_i\right)^2 + \sum_i \lambda_i \left(p_i^n \right)^2 \\
                                    = &\abs{x- \sum_i \lambda_i p_i}^2 + \sum_j \left[- \left(\sum_i \lambda_i p^j_i\right)^2 + \sum_i \lambda_i \left(p_i^j \right)^2 \right]. \numberthis \label{equation-geo_straight_2}
\end{align*}
Then, applying Equation~\ref{equation-geo_straight_2} for both both $x = \widetilde{x}$ and $x = x_e = \sum_j \lambda_j p_j$ in Equation~\eqref{equation-geo_straight}, we obtain:
\begin{align*}
\Bigg\lvert &\abs{ \widetilde{x} - \sum_i \lambda_i p_i }^2 + \sum_j \left[- \left(\sum_i \lambda_i p^j_i \right)^2 + \sum_i \lambda_i \left(p_i^j \right)^2 \right] \\
            &- \left( \abs{\sum_j \lambda_j p_j - \sum_i \lambda_i p_i}^2 + \sum_j \left[- \left(\sum_i \lambda_i p^j_i \right)^2 + \sum_i \lambda_i \left(p_i^j \right)^2 \right] \right) \Bigg\rvert \\
            & \qquad = \abs{ \widetilde{x} - \sum_i \lambda_i p_i}^2 \leq 2 \cdot 4^3 (\psi_0 - 1) \varepsilon^2.
\end{align*}
which yields a distance bound of $\sqrt{2 \cdot 4^3 (\psi_0 - 1) \varepsilon^2}$.
\end{proof}

Although we have formulated this metric distortion result for simplices, the same proof extends almost verbatim to continuous distributions.
By this we mean that the barycenter with respect to a metric $g$ of a continuous distribution is close to the barycenter with respect to the Euclidean metric, if $g$ is close to the Euclidean metric.
Furthermore, note that the proof does not depend on the weights being positive.

We now prove Theorem~\ref{theorem-SRDT_embedding_anyD}, recalled below.
\begin{theorem*}
Let $\mathcal{P}$ be a $\delta$-power protected $(\varepsilon, \mu)$-net with respect to $g$ on $\Omega$.
Let $\{ p_i \}$ be a set of $n+1$ vertices in $\Omega$ such that $\cap_{p_i\in\mathcal{P}} \V(p_i) \neq \emptyset$.
Let $\bar{\sigma}$ and $\widetilde{\sigma}$ be the straight and Karcher simplices with vertices $\{ p_i \}$.
Let $\widetilde{\tau}$ be a facet of $\widetilde{\sigma}$, opposite of the vertex $p_i$.
If for all $\widetilde{x} \in \widetilde{\tau}$, we have $\abs{\widetilde{x} - x_e}$ smaller than the lower bound on $D(p_i, \sigma)$, where $x_e$ is the corresponding point on $\bar{\sigma}$ (as defined in Lemma~\ref{lemma-geo_straight_proximity}), then there is no inversion created when $\widetilde{\sigma}$ is straightened onto $\bar{\sigma}$.
Furthermore, if this condition is fulfilled for all $\widetilde{\sigma}\in\RDT$, then $\SRDT$ is embedded.
\end{theorem*}
\begin{proof}
The lower bound on $D(p,\sigma)$ given in Appendix~\ref{appendix-dihedral_angles} is proportional to $\varepsilon$.
The proximity (upper) bound from Lemma~\ref{lemma-geo_straight_proximity} is proportional to $\sqrt{(\psi_0 - 1)} \varepsilon$, therefore going to $0$ much faster.
The embeddability is thus satisfied once
\begin{align*}
\sqrt{2 \cdot 4^3 (\psi_0 - 1) \varepsilon^2} &< \frac{\delta^2}{4\varepsilon} \iff \psi_0 < 1 + \frac{\iota^4}{32 \cdot 4^3}.
\end{align*}
\end{proof}

\section{Deforming the triangulation $\mathcal{T}_{v}$} \label{appendix-deforming}
An extreme configuration can have a sphere in Figure~\ref{fig-polytope_deformation_extreme} separating parts of $\mathcal{T}_v$ from the Voronoi vertex $v$.
In that case, we can chose different spheres to ``push away'' faces.
Figure~\ref{fig-polytope_deformation_extreme} shows the construction.
We do not detail the computations.

\begin{figure}[!htb]
  \centering
  \includegraphics[width=\linewidth]{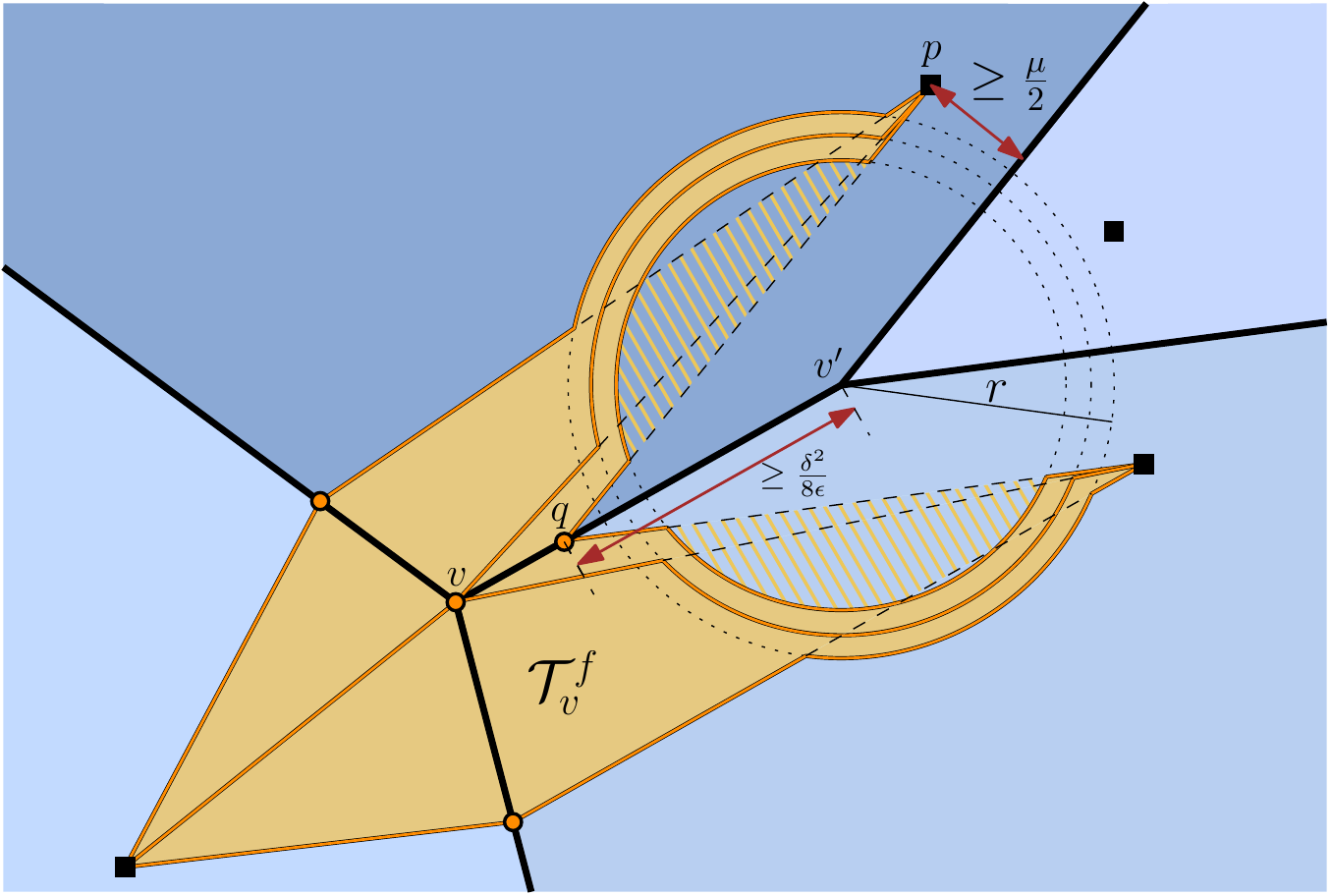}
  \caption{Limit case of the deformation of $\mathcal{T}_{v}$ into $\mathcal{T}_v$.}
  \label{fig-polytope_deformation_extreme}
\end{figure}

\section{Equality of the Riemannian Delaunay complexes in advanced geometric settings} \label{appendix-extension}
In this section, we explain precisely how to obtain conditions on the maximal length of a canvas edge and the quality of the sample set such that the $\DDC$ and the $\DC$ are equal.

\subsection{Uniform metric field} \label{section-uniform_anyD}
We first investigate the setting of a subdomain of $\R^n$ endowed with a uniform metric field.

If, for all $\sigma\in\del_g^{\textrm{d}}(\mathcal{P})$, there exists $\simp_{\canvas}\in\canvas$ such that $\simp_{\canvas}$ witnesses $\sigma$, and $\forall\simp_{\canvas}\in\canvas$ the simplex witnessed by $\simp_{\canvas}$ belongs to $\del_g^{\textrm{d}}(\mathcal{P})$, then we say that $\canvas$ \emph{captures} $\del_g^{\textrm{d}}(\mathcal{P})$.
A Voronoi cell is said to be \emph{captured} if all its Voronoi vertices are witnessed by the canvas.

By Lemma~\ref{lemma-protected_net_metric_transformation}, if a point set $\mathcal{P}$ in $\R^n$ is a $\delta$-power protected $(\varepsilon,\mu)$-net with respect to a uniform metric field $g_0$ then the point set $\mathcal{P}' = \{F_0p_i, p_i\in\mathcal{P}\}$, where $F_0$ is the square root of $g_0$, is also a $\delta$-power protected $(\varepsilon,\mu)$-net, but with respect to the Euclidean metric.
We can deduce an upper bound on the maximal length of an edge of $\canvas$ for a uniform metric field using the results of Theorem~\ref{theorem-basic_anyD} for the Euclidean setting.
The main result is given by the following theorem:
\begin{theorem} \label{theorem-uniform_metric_sizing_field_anyD}
Let $\mathcal{P}$ be a point set in $\Omega$.
Let $g$ be a uniform Riemannian metric field on $\Omega$ ($\forall x\in\Omega$, $g(x) = g_0$).
Let $\canvas$ be the canvas and let $e_{\canvas}$ be the length of its longest edge.
If
\[ e_{\canvas} < \left(\min\limits_{i}\sqrt{\lambda_i}\right) \min \left\{ \frac{\mu}{3}, \frac{\delta^2}{32\varepsilon} \right\} , \]
then $\DDC = \DC$.
\end{theorem}
\begin{proof}
Consider $\mathcal{R} = \Vor^{\textrm{d}}_{g_0}(\mathcal{P})$ over $\canvas$.
Let $F_0$ be a square root of $G_0$.
The matrix $F_0$ provides a stretching operator between the Euclidean and the metric spaces.
Let $\mathcal{P}_{0} = \{ F_0p, p\in\mathcal{P}\}$ be the transformed point set and $\canvas_{0}$ be a canvas (a dense triangulation) of the transformed space.
Denote by $\mathcal{R}_0$ the discrete Riemannian Voronoi diagram of $\mathcal{P}_{0}$ with respect to $g_{\E}$ over $\canvas_{0}$.
Let $e_{\canvas, 0}$ be the upper bound on the canvas edge length of $\canvas_{0}$ provided by Theorem~\ref{theorem-basic_anyD} such that $\mathcal{D}_{0}$ is captured by $\canvas_{0}$.

Since $\mathcal{P}_0$ is a $\delta$-power protected net with respect to $g_{\E}$, we can invoke Theorem~\ref{theorem-basic_anyD} and we must thus have
\[ e_{\canvas,0} < \min \left\{ \frac{\mu}{3}, \frac{\delta^2}{32\varepsilon} \right\}, \]
for the canvas $\canvas_0$ to capture $\del_{\E}(\mathcal{P}_0)$.

Let $\canvas'_{0}$ be the image of $\canvas$ by $F_0$.
Note that $\canvas'_{0}$ and $\canvas_{0}$ are two different triangulations of the same space.
If any edge of $\canvas'_{0}$ is smaller (with respect to the Euclidean metric) than $e_{\canvas, 0}$, then $\canvas'_0$ satisfies
\[ e_{\canvas', 0} < \min \left\{ \frac{\mu}{3}, \frac{\delta^2}{32\varepsilon} \right\} \]
and thus $\canvas'_0$ captures $\del_{\E}(\mathcal{P}_0)$.

Recall that given an eigenvector $v_i$ of $G_0$ with corresponding eigenvalue $\lambda_i$, a unit length in the direction $v_i$ in the metric space has length $1/\sqrt{\lambda_i}$ in the Euclidean space.
Therefore, if the bound $e_{\canvas}$ of $\canvas$ is smaller than $\alpha e_{\canvas_0}$, with
\[ \alpha = \frac{1}{\max\limits_{i}\left(\frac{1}{\sqrt{\lambda_i}}\right)} = \min\limits_{i} \sqrt{\lambda_i}, \]
then every edge of $\canvas'_{0}$ is smaller than $e_{\canvas, 0}$.
This implies that $\canvas'_{0}$ captures $\del_{\E}(\mathcal{P}_0)$ and therefore that $\canvas$ captures $\del_{g_0}(\mathcal{P})$.
\end{proof}
This settles the case of a uniform metric field.

\subsection{Arbitrary metric field} \label{section-generic_anyD}
We now consider an arbitrary metric field $g$ over the domain $\R^n$.
The key to proving the equality of the discrete Riemannian Delaunay complex and the Riemannian Delaunay complex in this setting is to locally approximate the arbitrary metric field with a uniform metric field, a configuration that we have dealt with in the previous section.
We shall always compare the metric field $g$ in a neighborhood $U$ with the uniform metric field $g'= g(p_0)$ where $p_0 \in U$.
Because $g'$ and the Euclidean metric field differ by a linear transformation, we can simplify matters and assume that $g'$ is the Euclidean metric field.
The main argument of the proof will be once again that a power protected has stable and separated Voronoi vertices.

We recall the main result of this section, Theorem~\ref{theorem-generic_anyD}.
\begin{theorem*}
Let $g$ be an arbitrary metric field on $\Omega$.
Assume that $\mathcal{P}$ is a $\delta$-power protected $(\varepsilon,\mu)-$net in $\Omega$ with respect to $g$.
Denote by $\canvas$ the canvas, and $e_{\canvas}$ the length of its longest edge.
If
\[ e_{\canvas} < \min\limits_{p\in\mathcal{P}} e_{\canvas, p}, \]
where $e_{\canvas, p}$ is given by Lemma~\ref{lemma-generic_single_cell_anyD}, and if $\varepsilon$ is sufficiently small and $\delta$ is sufficiently large (both values will be detailed in the proof), then $\DDC = \DC$.
\end{theorem*}

We prove Theorem~\ref{theorem-generic_anyD} by computing for each point $p\in\mathcal{P}$ the maximal edge length of the canvas such that the Voronoi cell $\V_g(p)$ is captured correctly.
Conditions on $\varepsilon$ and $\delta$ shall emerge from the intermediary results on the stability of power protected nets.

\begin{lemma} \label{lemma-generic_single_cell_anyD}
Let $\psi_0 \geq 1$ be a bound on the metric distortion and $g$ be a Riemannian metric field on $U$.
Let $U$ be an open neighborhood of $\Omega = \R^n$ that is included in a ball $B_g(p_0, r_0)$, with $p_0 \in U$ and $r_0 \in \R^{+}$ such that ${\forall p\in B(p_0,r_0)}, {\psi(g(p), g_{\E}(p)) \leq \psi_0}$.
Let $\mathcal{P}_U$ be a point set in $U$ and let $p_0 \in \mathcal{P}_U$.

Suppose that $\mathcal{P}_U$ is a $\delta$-power protected $(\varepsilon,\mu)-$net of with respect to~$g$.
Let $\V_g(p_0)$ be the Voronoi cell of~$p_0$ in $\Vor^{\textrm{d}}_g(\mathcal{P}_U)$.
If
\[ e_{\canvas, p_0} < \min\limits_{i} \left( \sqrt{\lambda_i} \,\right) \min \left\{ \frac{\mu}{3} , \frac{\ell_0}{2} \right\}, \]
with $\{\lambda_i \}$ the eigenvalues of $g_0$ and $\ell_0$ that is made explicit in the proof, and if $\varepsilon$ is sufficiently small and $\delta$ is sufficiently large (both values will also be detailed in the proof), then $\DDC = \DC$.
\end{lemma}

\subsubsection{Approach}
The many intermediary results needed to prove Theorem~\ref{theorem-generic_anyD} are presented in Appendices~\ref{appendix-separation},~\ref{appendix-dihedral_angles} and~\ref{appendix-stability}.
We refer to them at appropriate times.

We use the fact that a Riemannian Voronoi cell $\V_g(p_0)$ can be encompassed into two Euclidean Voronoi cells $\DV_{\E}^{+\eta}(p_0)$ and $\EV_{\E}^{-\eta}(p_0)$ that are scaled up and down versions of $\V_{\E}(p_0)$ (Lemma~\ref{lemma-dilated_eroded_Voronoi_cells}).
Specifically, $\EV_{\E}^{-\eta}(p_0)$ and $\DV_{\E}^{+\eta}(p_0)$ are defined by
\[ \EV_{\E}^{-\omega}(p_0) = \{ x \in \V_{\E}(p_0) \mid d_{\E}(x, \partial \V_{\E}(p_0)) > \omega \}, \]
and
\[ \DV_{\E}^{+\omega}(p_0) = \bigcap\limits_{i\neq 0} H^\omega(p_0, p_i), \]
where $H^\omega(p_0, p_i)$ is the half-space containing $p_0$ and delimited by the bisector $\BS(p_0, p_i)$ translated away from $p_0$ by $\omega_0$.
The constant $\eta$ is the thickness of this encompassing and depends on the bound on the distortion $\psi_0$ in the neighborhood, and on the sampling and separation parameters $\varepsilon$ and $\mu$.
We have that $\eta$ goes to $0$ as $\psi_0$ goes to $1$.

The (local) stability of the power protected nets assumption is proved here again (Lemmas~\ref{lemma-net_metric_perturbation} and~\ref{lemma-protection_metric_perturbation}).
From this observation, we can deduce that the Voronoi vertices of the Euclidean Voronoi cell $\V_{\E}(p_0)$ are separated, and thus that the Voronoi vertices of $\EV_{\E}(p_0)$ are separated.
A bound on the maximal length of a canvas edge can then be computed such that $\EV_{\E}(p_0)$ is captured and thus $\V_g(p_0)$ is captured.

\paragraph{Difficulties}
The difficulty almost entirely comes from proving the stability of the assumption of power protection under metric perturbation in any dimension, that is proving that if we assume that $\mathcal{P}$ is $\delta$-power protected with respect to the arbitrary metric field $g$, then, in a small neighborhood around $p_0$, the point set $\mathcal{P}$ is $\delta_0$-power protected with respect to $g_0 = g(p)$.
Assuming a power protected net does give us some bounds, but creates a tricky circular dependency as the coefficient $\delta$ appears in the dihedral angles (see Lemma~\ref{lemma-dihedral_angle_phi}).
We remedy this issue by proving that Euclidean dihedral angles are bounded assuming power protection with respect to the arbitrary metric field, with Lemmas~\ref{lemma-dihedral_angle_metric} and~\ref{lemma-dihedral_angle_from_delta_pp}.

\subsubsection{Proof of Lemma~\ref{lemma-generic_single_cell_anyD}}
Lemma~\ref{lemma-relaxed_Voronoi_cell} gives us that $\V_{g}(p_0)$ lies in $\DV_{\E}^{+\eta}(p_0)$ and contains $\EV_{\E}^{-\eta}(p_0)$.
Since $\V_g(p_0)$ contains $\EV_{\E}^{-\eta}(p_0)$, if $e_{\canvas}$ is small enough such that $\EV_{\E}^{-\eta}(p_0)$ is captured, then $\V_g(p_0)$ is also captured.
Proving that $\EV_{\E}^{-\eta}(p_0)$ is captured is done similarly to the Euclidean setting.
While we do not explicitly have the power protected net property for the relaxed Voronoi cells (and specifically, $\EV_{\E}^{-\eta}(p_0)$), we can still extract the critical property that the Voronoi vertices are separated, as shown by the next lemma.

\begin{lemma} \label{lemma-EV_separation}
Assume $U$, $g$, and $\psi_0$ as in Lemma~\ref{lemma-generic_single_cell_anyD}.
Assume that the point set $\mathcal{P}_U$ is a $\delta_0$-power protected $(\varepsilon,\mu)$-net with respect to the Riemannian metric field $g$.
Then the Voronoi vertices of $\EV_{\E}^{-\eta}(p_0)$ are separated.
\end{lemma}
\begin{proof}
By Lemmas~\ref{lemma-net_metric_perturbation} and~\ref{lemma-protection_metric_perturbation}, we have local stability of the power protection and net properties.
Hence, $\mathcal{P}$ is $\delta_{0}$-power protected $(\varepsilon_{0},\mu_{0})$-net with respect to $g_{\E}$ in $U$.

Let $L_0 = \delta_0^2 / 4\varepsilon_0$ be the separation bound induced by the $\delta_0$-power protection property of $\mathcal{P}_U$ (see Lemma~\ref{lemma-global_separation_bound}).
Let $l$ be the distance between any two adjacent Voronoi vertices of ${\EV_{\E}^{-\eta}}(p_0)$.
We know by Lemma~\ref{lemma-voronoi_vertices_metric_perturbation} that the parallelotopic region around a Voronoi vertex is included in a ball centered on the Voronoi vertex and of radius $\chi$.
The protection parameter $\iota$ is given by $\delta = \iota \varepsilon$.
We have that
\begin{align}
l &\geq L - 2 \chi \nonumber \\
  &\geq \frac{\delta_0^2}{4\varepsilon_0} - 2 \frac{\chi_2}{\sin^{n-2}\left( \frac{\varphi}{2} \right)} \nonumber \\
  &\geq \frac{\delta_0^2}{4\varepsilon_0} - 2 \frac{4\eta}{\left( \sqrt{ 1+\frac{\lambda}{2\psi_0^2} } - \sqrt{ 1-\frac{\lambda}{2\psi_0^2} } \right) \left( \sqrt{1 + s_0} - \sqrt{1 - s_0} \,\right)^{n-2}} =: \ell_0, \label{equation-l0}
\end{align}
where $\varphi$ represents the dihedral angle and $s_0 = \frac{1}{2} \left[ \frac{\iota^2}{4\psi_0^2} - \frac{1}{2}\left(\psi_0^2 - \frac{1}{\psi_0^2}\right) \right]$.
For the stability regions not to intersect, we require $l$ to be positive.
This can be ensured by enforcing that the lower bound is positive:
\[ \frac{\delta_0^2}{4\varepsilon_0} > \frac{8\eta}{\left( \sqrt{ 1+\frac{\lambda}{2\psi_0^2} } - \sqrt{ 1-\frac{\lambda}{2\psi_0^2} } \right) \left( \sqrt{1 + s_0} - \sqrt{1 - s_0} \,\right)^{n-2}} \]
Recall that $\varepsilon_0 = \psi_0\varepsilon$, $\mu_0 = \lambda\varepsilon/\psi_0$ and $\rho_0 = 2\psi_0\varepsilon$.
Using these notations, we see that $l > 0$ if
\begin{align}
\delta_{0}^2 &> \frac{32 \varepsilon_0 \rho_0^2 (\psi_0^2 -1) }{\mu_0 \left( \sqrt{ 1+\frac{\lambda}{2\psi_0^2} } - \sqrt{ 1-\frac{\lambda}{2\psi_0^2} } \right) \left( \sqrt{1 + s_0} - \sqrt{1 - s_0} \,\right)^{n-2}} \nonumber \\
\delta_{0}^2 &> \frac{128 \varepsilon^2 \psi_0^4 (\psi_0^2 -1) }{\lambda \left( \sqrt{ 1+\frac{\lambda}{2\psi_0^2} } - \sqrt{ 1-\frac{\lambda}{2\psi_0^2} } \right) \left( \sqrt{1 + s_0} - \sqrt{1 - s_0} \,\right)^{n-2}}. \label{eq-eroded_separation_anyD}
\end{align}

This condition is easy to satisfy when $\psi_0$ goes to $1$ because the right hand side of Inequality~\eqref{eq-eroded_separation_anyD} is proportional to $(\psi_0^2 - 1) \varepsilon^2 $.

The intermediary results that we use already impose some conditions on $\delta$ and $\varepsilon$, and we thus would like to give the condition in Equation~\ref{eq-eroded_separation_anyD} in terms of $\delta$, so that it may be compared with Inequality~\eqref{first-bound-Delta}.
In Lemma~\ref{lemma-protection_metric_perturbation}, we have seen that
\[ \delta_{0}^2 = \frac{\delta^2}{\psi_0^2}+ \left(\frac{1}{\psi_0^2} - \psi_0^2\right) (\varepsilon + \chi)^2 - \frac{4\varepsilon\chi}{\psi_0^2}, \]
Thus
\begin{align*}
\frac{\delta^2}{\psi_0^2} + \left(\frac{1}{\psi_0^2} - \psi_0^2\right) (\varepsilon + \chi)^2 - \frac{4\varepsilon\chi}{\psi_0^2} > \frac{128 \varepsilon^2 \psi_0^4 (\psi_0^2 -1) }{\lambda \left( \sqrt{ 1+\frac{\lambda}{2\psi_0^2} } - \sqrt{ 1-\frac{\lambda}{2\psi_0^2} } \right) \left( \sqrt{1 + s_0} - \sqrt{1 - s_0} \,\right)^{n-2}},
\end{align*}
which is equivalent to
\begin{align}
\delta^2 > &\,\frac{128 \varepsilon^2 \psi_0^6 (\psi_0^2 -1) }{\lambda \left( \sqrt{ 1+\frac{\lambda}{2\psi_0^2} } - \sqrt{ 1-\frac{\lambda}{2\psi_0^2} } \right) \left( \sqrt{1 + s_0} - \sqrt{1 - s_0} \,\right)^{n-2}} \nonumber \\
           &+ 4\varepsilon\chi \nonumber \\
           &+ (\psi_0^4 - 1)(\varepsilon + \chi)^2. \nonumber \\
\iff \delta^2 > \,&\,8\varepsilon\psi_0^3\chi  + 4\varepsilon\chi + (\psi_0^4 - 1)(\varepsilon + \chi)^2. \label{second-bound-Delta_anyD}
\end{align}
This bound is again proportional to $(\psi_0^2 -1)\varepsilon$ and is very similar to the bound given by Inequality~\eqref{first-bound-Delta}, made explicit in Inequality~\eqref{Exegesis-Bound},
%namely
% \begin{align*}
% \delta^2 > &
% 4 \left(\frac{16 \varepsilon\psi_0^4(\psi_0^2-1)}{\lambda \left( \sqrt{1 + \frac{\lambda}{2\psi_0^2}} - \sqrt{1 - \frac{\lambda}{2\psi_0^2}} \right) \left(\sqrt{1 + s_0} - \sqrt{1 - s_0}\,\right)^{n-2}}\right)^2 \nonumber \\
% &+ 4 \frac{16\varepsilon \psi_0^3(\psi_0^2 - 1)}{\lambda \left( \sqrt{1 + \frac{\lambda}{2\psi_0^2}} - \sqrt{1 - \frac{\lambda}{2\psi_0^2}} \right) \left(\sqrt{1 + s_0} - \sqrt{1 - s_0}\,\right)^{n-2}} \psi_0^3 \nonumber \\
% &+ (\psi_0^2 - 1)(\psi_0^2 + 1) \varepsilon^2
% \end{align*}
but Inequality~\eqref{second-bound-Delta_anyD} provides the tougher bound due to the $(\varepsilon + \chi)$ coefficient.

% After simplifying the $(\varepsilon + \chi)$-term, we find
% \begin{align}
% 1 + 2 > 2\label{second-bound-Delta}
% \end{align}

% If we compare this to the bound given by Inequality~\eqref{first-bound-Delta}, made explicit in Inequality~\eqref{Exegesis-Bound}, namely 
% \begin{align}
% \delta^2 >
%   &\frac{256 ( \psi_0 ^2- 1) \psi_0^2} {\lambda^2 \left( \sqrt{1+\frac{\lambda}{2\psi_0^2}} - \sqrt{1-\frac{\lambda}{2\psi_0^2} } \right)^2 } ( \psi_0 ^2- 1) \varepsilon^2 \nonumber \\
%  + &\frac{32 \psi_0^3 } {\lambda\left( \sqrt{1+\frac{\lambda}{2\psi_0^2}} - \sqrt{1-\frac{\lambda}{2\psi_0^2} } \right) } ( \psi_0 ^2- 1) \varepsilon^2 \nonumber \\
%  + &(\psi_0^2 + 1) (\psi_0^2 - 1) \varepsilon^2, \label{rep-first-Bound-Delta} 
% \end{align}
% we see that the two final terms in the sum in Inequalities~\eqref{second-bound-Delta} and~\eqref{rep-first-Bound-Delta} are of equal magnitude.
% The final terms would in fact agree if we would not have simplified the equation, because $\chi_2$ is of order $\varepsilon (\psi_0^2-1)$.
% The first term disappears for Inequality~\eqref{rep-first-Bound-Delta} in the limit where $\psi_0$ tends to $1$, while this is not the case for the first term in Inequality~\eqref{second-bound-Delta}.

% This means that at least when $\psi_0$ goes to $1$, Equation~\eqref{second-bound-Delta} provides the tougher bound.
% However both inequalities have to be verified explicitly in a practical setting.
\end{proof}

We can now provide an upper bound on the length of any canvas edge so that it captures $\EV_{\E}^{-\eta}(p_0)$ and prove Lemma~\ref{lemma-generic_single_cell_anyD}.
From Theorem~\ref{theorem-basic_anyD}, we have that if the canvas edge length is bounded as: $e_{\canvas} < \min \{ \mu_0 / 16, \delta_0^2 / 64\varepsilon_0 \}$, then $\V_{\E}(p_0)$ is captured as $\mathcal{P}$ is a $\delta_0$-power protected $(\varepsilon_0, \mu_0)$-net with respect to the Euclidean metric field.
As we want to capture $\EV_{\E}(p_0)$, we cannot directly use this result.
We have nevertheless obtained the separation between the Voronoi vertices of the eroded Voronoi cell (Equation~\ref{equation-l0}).
It is then straightforward to modify the result of Theorem~\ref{theorem-Euclidean_capture} by using the separation bound provided in Lemma~\ref{lemma-EV_separation} instead of the one provided by Lemma~\ref{lemma-global_separation_bound}.
We thus choose
\[ e_{\canvas, p_0}^0 = \min \left\{ \frac{\mu}{3} , \frac{\ell_0}{2} \right\}. \]

\begin{remark}
We here ignore the consequences of Appendix~\ref{appendix-deforming} as it only complicates formulas without changing the logical steps.
\end{remark}

We should not forget that we have assumed that $g_0$ is the Euclidean metric field, which is generally not the case, we must in fact proceed like for the case of a uniform metric field (Theorem~\ref{theorem-uniform_metric_sizing_field_anyD}):
\[ e_{\canvas, p_0} < \min\limits_{j} \left( \sqrt{\lambda_i} \right) (e_{\canvas, p_0}^0), \]
with $\{\lambda_i \}$ the eigenvalues of $g_0$.

Therefore, if the site set satisfies the previous conditions on $\varepsilon$ and $\delta$ and the canvas is enough for all of its edges to have a length smaller than $e_{\canvas, i}$, then $\EV_{\E}^{-\eta}(p_0)$ is captured, and thus $\V_g(p_0)$ is captured, which proves Lemma~\ref{lemma-generic_single_cell_anyD}.

Taking the minimum of all the bounds $e_{\canvas}$ over all $p_i\in\mathcal{P}$, we obtain an upper bound on the length of the longest canvas edge,
\[ e_{\canvas} = \min_{p_i\in\mathcal{P}} e_{\canvas, p_i}, \]
such that all the Voronoi cells are captured.

In all the Lemmas necessary to obtain the local results, we have imposed conditions on $\varepsilon$ and $\delta$.
Similarly to the bound $e_{\canvas}$, the domain-wide bounds on $\varepsilon$ and $\delta$ are computed from the local values of $\varepsilon$ and $\delta$.

Finally, this proves that $\DDC = \DC$ in the general setting (Theorem~\ref{theorem-basic_anyD}), when geodesics are exactly computed.

\subsection{Approximate geodesic distance computations} \label{section-approx_geo}
We have so far assumed that geodesics are computed exactly, which is generally not the case in practice.
Nevertheless, once the error in the approximation of the geodesic distances is small enough, the computation of the discrete Riemannian Voronoi diagram with approximate geodesic distances can be equivalently seen as the computation of a second discrete Riemannian Voronoi diagram using exact geodesic distances but for a slightly different metric field.

Denote by $\widetilde{d_g}$ the geodesic approximation and $d_g$ the exact geodesic distance with respect to the metric field $g$.
Assume that in a small enough neighborhood $U$ (see Lemma~\ref{lemma-geodesic_distortion}), the distances can be related as
\[ \abs{ d_g(p_0, x) - \widetilde{d_g}(p_0, x) } \leq \xi d_g(p_0, x), \]
where $\xi$ is a function of $x$ that goes to $0$ as the sampling parameter $\varepsilon$ goes to $0$.
We can formulate a lemma similar to Lemma~\ref{lemma-relaxed_Voronoi_cell} to bound the distance between the same bisectors between sites for the exact and the approximate diagrams.

\begin{lemma} \label{lemma-approximate_relaxed_Voronoi_cell}
Let $\psi_0 \geq 1$ be a bound on the metric distortion and $g$ and $g'$ be two Riemannian metric fields on $U$.
Let $U$ be an open neighborhood of $\Omega = \R^2$ that is included in a ball $B_g(p_0, r_0)$, with $p_0 \in U$ and $r_0 \in \R^{+}$ such that ${\forall p\in B(p_0,r_0)}, {\psi(g(p), g'(p)) \leq \psi_0}$.
Let $\mathcal{P}_U$ be a point set in $U$ and let $p_0 \in \mathcal{P}_U$.
Let $\mathcal{P}_U = \{p_i\}$ be a point set in $U$.
Let $\widetilde{V_{p_0, g}}$ denote a Voronoi cell with respect to the approximate geodesic distance.

Suppose that the Voronoi cell $V_{g}^{+2\xi\tilde{\rho}}(p_0)$ lies in a ball of radius $\tilde{\rho}$ with respect to the metric $g$, which lies completely in $U$.
Then~$\widetilde{V_{p_0, g}}$ lies in~$V_{g}^{+2\xi\tilde{\rho}}(p_0)$ and contains~$V_{g}^{-2\xi\tilde{\rho}}(p_0)$.
\end{lemma}
\begin{proof}
Let $\widetilde{\BS_{g}}(p_0, p_i)$ be the bisector between $p_0$ and $p_i$ with respect to the approximate geodesic distance.
Let $y \in \widetilde{\BS_{g}}(p_0, p_i) \cap B_{g}(p_0, \rho)$, where $B_{g}(p_0, \rho)$ denotes the ball centered at $p_0$ of radius $\rho$ with respect to the exact geodesic distance.
Now $\widetilde{d_{g}}(y,p_0) = \widetilde{d_{g}}(y,p_i)$, and thus 
\begin{align*}
|d_{g}(y,p_0)^2 - d_{g}(y,p_i)^2| & = \abs{d_{g}(y,p_0)^2 - \widetilde{d_{g}}(y,p_0)^2 + \widetilde{d_{g}}(y,p_i)^2 - d_{g}(y,p_i)^2} \\
                                  & \leq  \abs{d_{g}(y,p_0)^2 - \widetilde{d_{g}}(y,p_0)^2} + \abs{d_{g}(y,p_i)^2 - \widetilde{d_{g}}(y,p_i)^2} \\
                                  & \leq 2 \xi (d_{g}(y,p_0)^2 + d_{g}(y,p_i))^2 \\
                                  & \leq 2 \xi \tilde{\rho}^2 .
\end{align*}
Thus $d_{g}(y,p_0)^2 \leq d_{g}(y,p_i)^2 + \tilde{\omega}$ and $d_{g}(y,p_0)^2 \geq d_{g}(y,p_i)^2 - \tilde{\omega}$ with $\tilde{\omega} = 2\xi\tilde{\rho}^2$, which gives us the expected result.
\end{proof}

Denote by $\widetilde{V}(p_0)$ the Voronoi cell of $p_0$ with the approximate metric.
We can then incorporate Lemma~\ref{lemma-approximate_relaxed_Voronoi_cell} to obtain a result similar to Lemma~\ref{lemma-dilated_eroded_Voronoi_cells}.

\begin{lemma}
Let $U$, $\psi_0$, $g$, $g_0$ and $\mathcal{P}_U$ be defined as in Theorem~\ref{theorem-generic_anyD}.
Then we can find $\eta'_0$ and $\omega'_0$ such that
\[ \EV_{g_0}^{-\eta'_0}(p_0) \subseteq \V_{g_0}^{-\omega'_0}(p_0) \subseteq \V_{g}^{-\tilde{\omega}}(p_0) \subseteq \widetilde{\V_{g}}(p_0) \subseteq \V_{g}^{+\tilde{\omega}}(p_0) \subseteq \V_{g_0}^{+\omega'_0}(p_0) \subseteq \DV_{g_0}^{+\eta'_0}(p_0). \]
These inclusions are illustrated in Figure~\ref{fig-encompassing_approx}.
\end{lemma}

\begin{figure}[!htb]
\centering
\includegraphics[width=0.7\linewidth]{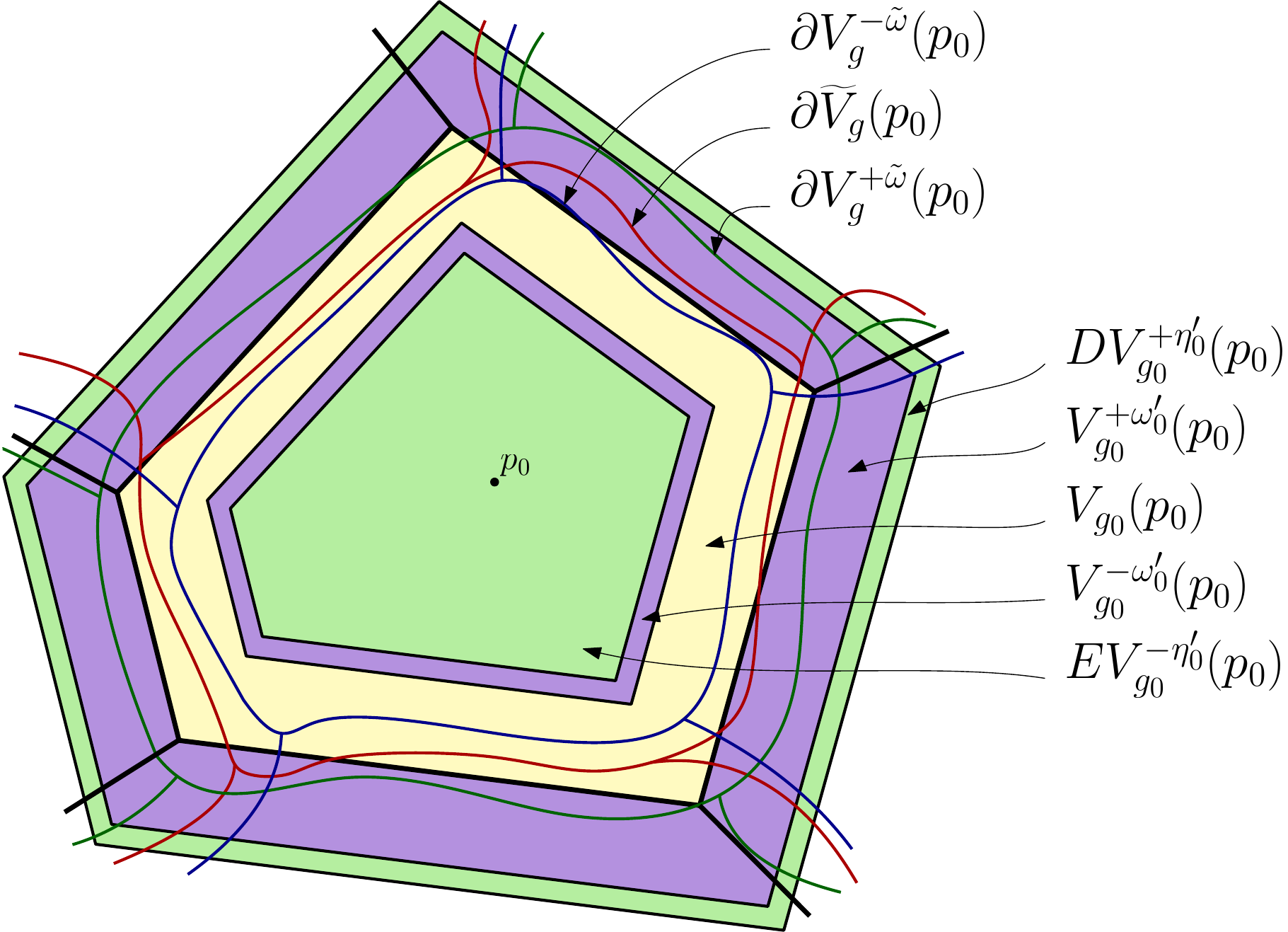}
\caption{Illustration of the different encompassing cells around $p_0$ in the context of an approximate geodesic distance.
         The RVDs with respect to $g$ and $g_0$ are respectively traced in red and black.
         The dilated Voronoi cells are traced in blue and green.
         The Voronoi cell $V_{g_0}(p_0)$ is colored in yellow.
         The cells $DV_{g_0}^{+\eta'_0}$ and $EV_{g_0}^{-\eta'_0}$ are colored in green, and the cells $V_{g_0}^{\pm\omega'_0}(p_0)$ are colored in purple.}
\label{fig-encompassing_approx}
\end{figure}

The subsequent lemmas and proofs are similar to what was done in the case of exact geodesic computations and we do not explicit them.

\end{document}
\endinput
%%
%% End of file `squelette-rr.tex'.